\documentclass[11pt, letterpaper]{article}
\title{Optimal tiling of the Euclidean space using symmetric bodies}
\author{Mark Braverman
\thanks{Department of Computer Science, Princeton University.
Research supported in part by the NSF Alan T. Waterman Award, Grant No. 1933331, a Packard Fellowship in Science and Engineering, and the Simons Collaboration on Algorithms and Geometry. }
\and
Dor Minzer
\thanks{Department of Mathematics, Massachusetts Institute of Technology.
Part of this work was done while the author was a member in the Institute for Advanced Study, Princeton,
supported by NSF grant DMS-1638352 and Rothschild Fellowship.}
}
\date{\vspace{-5ex}}
\usepackage{fullpage}
\usepackage{amsthm}
\usepackage{amsmath,amssymb,amsfonts,nicefrac}
\usepackage{xspace}
\usepackage{color}
\usepackage{url}
\usepackage{hyperref}

\usepackage{bm}
\usepackage{bbm}
\usepackage{times}

\usepackage{enumitem}

\newtheorem{thm}{Theorem}[section]

\newtheorem{lemma}[thm]{Lemma}

\newtheorem{claim}[thm]{Claim}
\newtheorem{proposition}[thm]{Proposition}
\newtheorem{definition}[thm]{Definition}

\newtheorem{fact}[thm]{Fact}
\newtheorem{problem}{Problem}

\newcommand\E{\mathop{\mathbb{E}}}

\newcommand\card[1]{\left| {#1} \right|}
\newcommand\sett[2]{\left\{ \left. #1 \;\right\vert #2 \right\}}

\newcommand\set[1]{{\left\{ #1 \right\}}}
\newcommand\Prob[2]{{\Pr_{#1}\left[ {#2} \right]}}

\newcommand\cProb[3]{{\Pr_{#1}\left[ \left. #3 \;\right\vert #2 \right]}}

\newcommand\Expect[2]{{\mathop{\mathbb{E}}_{#1}\left[ {#2} \right]}}

\newcommand\cExpect[3]{{\mathbb{E}_{#1}\left[ \left. #3 \;\right\vert #2 \right]}}
\newcommand\cbExpect[3]{{\mathbb{E}_{#1}\big[ \big. #3 \;\big\vert #2 \big]}}

\newcommand\norm[1]{\| #1 \|}
\newcommand\power[1]{\set{0,1}^{#1}}

\newcommand\half{{1\over2}}

\newcommand\defeq{\stackrel{def}{=}}

\newcommand\eps{\varepsilon}

\renewcommand\geq{\geqslant}
\renewcommand\ge{\geqslant}
\renewcommand\leq{\leqslant}
\renewcommand\le{\leqslant}

\newcommand{\rom}[1]{\uppercase\expandafter{\romannumeral #1\relax}}

\makeatletter
\def\moverlay{\mathpalette\mov@rlay}
\def\mov@rlay#1#2{\leavevmode\vtop{%
   \baselineskip\z@skip \lineskiplimit-\maxdimen
   \ialign{\hfil$\m@th#1##$\hfil\cr#2\crcr}}}
\newcommand{\charfusion}[3][\mathord]{
    #1{\ifx#1\mathop\vphantom{#2}\fi
        \mathpalette\mov@rlay{#2\cr#3}
      }
    \ifx#1\mathop\expandafter\displaylimits\fi}
\makeatother

\newcommand{\ZZ}{\mathbb{Z}}
\newcommand{\RR}{\mathbb{R}}
\newcommand{\ga}{\gamma}
\newcommand{\la}{\lambda}
\newcommand{\al}{\alpha}
\newcommand{\si}{\sigma}
\newcommand{\De}{\Delta}

\newcommand{\ve}{\varepsilon}

\def\gg{\gtrsim}
\def\ll{\lesssim}
\def\ee{\asymp}

\begin{document}
\maketitle
\begin{abstract}
What is the least surface area of a symmetric body $B$ whose $\mathbb{Z}^n$ translations tile $\mathbb{R}^n$? Since any such body must have volume $1$, the isoperimetric inequality implies that its surface area must be at least $\Omega(\sqrt{n})$. Remarkably, Kindler et al.\ showed that for general bodies $B$ this is tight, i.e.\ that there is a tiling body of $\mathbb{R}^n$ whose surface area is $O(\sqrt{n})$.

In theoretical computer science, the tiling problem is intimately to the study of parallel repetition theorems (which are an important component in PCPs),
and more specifically in the question of whether a ``strong version'' of the parallel repetition theorem holds.
Raz showed, using the odd cycle game, that strong parallel repetition fails in general, and subsequently these ideas were used
in order to construct non-trivial tilings of $\mathbb{R}^n$.

In this paper, motivated by the study of a symmetric parallel repetition, we consider the symmetric variant of the tiling problem in $\mathbb{R}^n$. We show that any symmetric
body that tiles $\mathbb{R}^n$ must have surface area at least $\Omega(n/\sqrt{\log n})$, and that this bound is tight, i.e.\ that there is a symmetric tiling body of
$\mathbb{R}^n$ with surface area $O(n/\sqrt{\log n})$. We also give matching bounds for the value of the symmetric parallel repetition of Raz's odd cycle game.

Our result suggests that while strong parallel repetition fails in general, there may be important special cases where it still applies.
\end{abstract}

\section{Introduction}

A body $D\subseteq\mathbb{R}^n$ is said to be tiling the Euclidean space $\mathbb{R}^n$, if its translations by $\mathbb{Z}^n$
cover the entire space and have disjoint interiors. The foam problem asks for the least surface area a tiling body
$D$ can have. The problem had been considered by mathematicians already in the 19th century \cite{Thomson87},
and it also appears in chemistry, physics and engineering \cite{chemistry}. More recently, the problem had received
significant attention in the theoretical computer science community due to its strong relation with the parallel repetition problem~\cite{FKO,KROW,AK}.

The simplest example for a body that tiles the Euclidean space is the solid cube, $D = [0,1]^n$, which has surface
area $2n$. At first glance, one may expect the solid cube to be the best example there is, or more modestly
that any tiling body would need to have surface area $\Omega(n)$. The main results of \cite{KROW,AK}
show that this initial intuition is completely false, and that there are far more efficient tiling bodies whose surface area is $O(\sqrt{n})$.
This is surprising, since spheres --- which are the minimizers of surface area among all bodies with a given, fixed volume
(in this case volume $1$), have $\Theta(\sqrt{n})$ surface area and seem to be very far from forming a tiling of $\mathbb{R}^n$.
As we will shortly discuss, the existence of such surprising tiling body is intimately related to the existence of another
surprising object -- namely non-trivial strategies for $2$-prover-$1$-round games, repeated in parallel. The main goal of this
paper is to understand the symmetric variant of the foam problem, which is closely related to the symmetric variant of
parallel repetition.

\subsection{$2$-Prover-$1$-Round Games and Parallel Repetition}
\begin{definition}\label{def:game}
A $2$-Prover-$1$-Round Game $G=(L\cup R,E,\Phi,\Sigma_L,\Sigma_R)$ consists of a bipartite graph $(L\cup R,E)$,
alphabets $\Sigma_L,\Sigma_R$, and a constraint $\Phi(u,v)$ for every edge $(u,v) \in E$. The goal is to find assignments
$A_L: L \to \Sigma_L$, $A_R: R \to \Sigma_R$ that satify the maximum fraction of the constraints. A constraint $\Phi(u,v)$
is satisfied if $ (A_L(u), A_R(v)) \in \Phi(u,v)$, where by abuse of notation, $\Phi(u,v) \subseteq \Sigma_L\times \Sigma_R$
denotes the subset of label pairs that are deemed satisfactory.

The value of a game, denoted by ${\sf val}(G)$, is the maximum fraction of constraints that can be satisfied in
$G$ by any pair of assignments $A_L,A_R$.
\end{definition}
Equivalently, a $2$-Prover-$1$-Round Game can be viewed as a ``game" between two provers and a verifier.
The verifier picks a constraint $(u,v)$ at random, asks the ``question" $u$ to the left prover, the ``question" $v$ to the right prover,
receives ``answers" $A_L(u), A_R(v)$ respectively from the provers; the verifier accepts if and only if $(A_L(u), A_R(v)) \in \Phi(u,v)$. It is
easy to see that in this language, ${\sf val}(G)$ represents the maximum probability a verifier will accept, where the maximum
is taken over all of the strategies of the provers.

$2$-Prover-$1$-Round games play an important role in the study of PCPs and Hardness of approximation, and in fact an equivalent
statement of the seminal PCP Theorem~\cite{FGLSS,AS,ALMSS} can be stated in that language. It will be convenient for us to use the notation of gap problems:
for $0<s<c\leq 1$, denote by ${\sf Gap2Prover1Round}(c,s)$ the promise problem in which the input is a $2$-Prover-$1$-Round game
$G$ promised to either satisfy ${\sf val}(G)\geq c$ or ${\sf val}(G)\leq s$, and the goal is to distinguish between these two cases.
The parameters $c$ and $s$ are referred to as the completeness and soundness parameters of the problem, respectively.

\begin{thm}[PCP Theorem, \cite{FGLSS,AS,ALMSS}]
There exists $k\in\mathbb{N}$, $s<1$ such that ${\sf Gap2Prover1Round}(1,s)$ is NP-hard on instances with alphabet size at most $k$.
\end{thm}

The PCP Theorem, as stated above, can be used to establish some hardness of approximation results. However it turns out that to get
strong hardness results, one must prove a variant of the theorem with small soundness, i.e.\ with $s$ close to $0$. One way to do
that is by amplifying hardness using \emph{parallel repetition}.

The $t$-fold repetition of a game $G$, denoted by $G^{\otimes t}$, is the game in which the verifier picks
$t$ independently chosen challenges, $(u_1,v_1),\ldots,(u_t,v_t)$ and sends them to the provers in a single
bunch, i.e.\ $\vec{u} = (u_1,\ldots,u_t)$ to one prover and $\vec{v} = (v_1,\ldots,v_t)$ to the second one. The provers are supposed
to give an answer to each one of their questions, say $A_L(\vec{u}) = (a_1,\ldots,a_t)$ and $A_R(\vec{v}) = (b_1,\ldots,b_t)$,
and the verifier accepts with only if $(a_i,b_i)\in\Phi(u_i,v_i)$ for all $i=1,\ldots,t$. What is the value of the $t$-fold
repeated game, as a function of ${\sf val}(G)$ and $t$?

The idea of parallel repetition was first introduced in~\cite{FRS}, wherein it was originally suggested that ${\sf val}(G^{\otimes t}) \approx {\sf val}(G)^t$.
Alas, in a later version of that paper it was shown to be false, leaving the question wide open. Raz~\cite{Razrep} was the first to prove that the value of
the repeated game decreases exponentially with $t$, and with many subsequent works improving the result \cite{Holenstein,Rao,DS,BG}.
The most relevant version for our purposes is the result of Rao~\cite{Rao}, which makes the following statement. First, we say a game $G$
is a projection game, if all of the constraints $\Phi(u,v)$ can be described by a projection map, i.e.\ there is a mapping $\pi_{u,v}\colon\Sigma_L\to\Sigma_R$
such that $\Phi(u,v) = \sett{(a,b)}{b = \phi_{u,v}(a)}$.

\begin{thm}\label{thm:Rao}
If $G$ is a projection game, and ${\sf val}(G) = 1-\eps$, then ${\sf val}(G^{\oplus t}) \leq (1-\eps^2)^{\Omega(t)}$.
\end{thm}
Rao's result seems nearly optimal, in the sense that a-priori, the best bound one can hope for is that ${\sf val}(G^{\oplus t}) \leq (1-\eps)^{\Omega(t)}$.
Quantitatively speaking, one may think that for all intents and purposes, Rao's bound is just as good as the best one can hope for. However,
as it turns out, there is at least one prominent problem where this quadratic gap is what makes the difference, which we describe next.

\paragraph{The Unique Games Conjecture and the Max-Cut Conjecture.} The Unique Games problem is a specific type of projection $2$-Prover-$1$-Round Game,
in which the projection maps $\phi_{u,v}$ are also bijections. The Unique Games Conjecture of Khot~\cite{Khot} (abbreviated UGC henceforth)
asserts that a strong PCP theorem holds for Unique-Games, and
more specifically that for any $\eps,\delta>0$, the problem ${\sf GapUG}(1-\eps,\delta)$ is NP-hard, when the alphabet sizes depend only on $\eps,\delta$.
This conjecture is now of central importance in complexity theory, and it is known to imply many, often tight inapproximability results
(see~\cite{KhotICM-survey,TrevisanUGCSurvey} for more details).
A prominent example is the result of~\cite{KKMO}, stating that assuming UGC, the Goemans-Williamson algorithm~\cite{GW} for Max-Cut is optimal.
In particular, for small enough $\eps>0$, if UGC is true, then ${\sf GapMaxCut}(1-\eps,1-\frac{2}{\pi}\sqrt{\eps}+o(1))$ is NP-hard. Does the converse hold? I.e.,
does the assumption that ${\sf GapMaxCut}(1-\eps,1-\frac{2}{\pi}\sqrt{\eps}+o(1))$ is NP-hard imply UGC? If so, that would be a promising avenue of attack on the
Unique-Games Conjecture.

Noting that Max-Cut is a Unique-Game and that Parallel repetition preserves uniqueness, one may hope that a reduction from
${\sf GapMaxCut}(1-\eps,1-\frac{2}{\pi}\sqrt{\eps}+o(1))$ to ${\sf GapUniqueGames}(1-\eps',\delta)$
would simply follow by appealing to a parallel repetition theorem, such as Rao's result~\cite{Rao}.
Alas, the quadratic loss there exactly matches the quadratic gap we have in Max-Cut, thereby nullifying it completely. This possibility
was discussed in~\cite{SS}, who among other things proposed that perhaps a stronger version of Theorem~\ref{thm:Rao} should hold for Unique-Games, in which
the $\eps^2$ is replaced with $\eps$. This conjecture was referred to as the Strong Parallel Repetition Conjecture, and unfortunately it turns out to be false.

\paragraph{A Strong parallel repetition theorem?}
The problem of understanding
parallel repetition over a very simple game, called the odd cycle game and denoted below by $C_n$, was shown to be closely related to the foam problem~\cite{FKO}.
In this game, we have a graph $G$ which is an odd cycle of length $n$, and the provers try to convince the verifier that $G$ is
a bipartite graph (while it is clearly not). To test the provers, the verifier picks a vertex $u$ from the cycle uniformly at random,
and then picks $v$ as $v=u$ with probability $1/2$, and otherwise $v$ is one of the neighbours of $u$ with equal probability. The verifier sends $u$ as a question
to one prover, and $v$ as a question to the other prover, and expects to receive a bit from each one $b_1, b_2$. The verifier checks
that $b_1 = b_2$ in case $u=v$, or that $b_1\neq b_2$ in case $u\neq v$.

Note that clearly, ${\sf val}(C_n) = 1-\Theta(1/n)$, and so the Strong Parallel Repetition Conjecture would predict that
the value of the $t$-fold repeated game is $1-\Theta(t/n)$ so long as $t\leq n$. Alas, this turns out to be false. First, in~\cite{FKO}, it was shown
that non-trivial solutions to the foam problem imply non-trivial strategies for the $t$-fold repeated game, and in particular the existence of a tiling
body with surface area $o(n)$ would refute the Strong Parallel Repetition Conjecture. Subsequently, Raz~\cite{Raz} showed that
the value of the $t$-fold repeated odd-cycle game is in fact at least $1-O(\sqrt{t}/n)$ so long as $t\leq n^2$, and that
Theorem~\ref{thm:Rao} is optimal (i.e., the quadratic gap is necessary, even for Unique-Games, and more specifically for Max-Cut).
Subsequent works were able to use these insights to solve the foam problem for the integer lattice~\cite{KROW,AK} and lead to better understanding of parallel repetition
and its variants~\cite{BHHRRS,BRRRS}. From the point of view of UGC, these results were very discouraging since they eliminate one of the main available venues
(perhaps the main one) for the proof of UGC.

Partly due to this issue, the best partial results towards UGC had to take an entirely different approach~\cite{KMS1,DKKMS1,DKKMS2,KMS2,BKS}, and currently
can only prove that ${\sf GapUG}(1/2,\delta)$ is NP-hard for every $\delta>0$.

\subsection{A symmetric variant of Parallel Repetition}
One may try to revive the plan for showing the equivalence of UGC and the hardness of Max-Cut by considering variants of parallel repetition. Ideally,
for that approach to work, one should come up with a variant of parallel repetition, in which (a) the value decreases exponentially with
the number of repetitions, and (b) the operation preserves uniqueness. One operation that had been considered in the literature, for example, is called fortification~\cite{Moshkovitz,BSVV}. Using this operation, the value of the game indeed decreases exponentially, however this operation does not preserve
uniqueness and therefore is not useful for showing the equivalence of UGC and the Max-Cut Conjecture.

More relevant to us is the symmetric variant of parallel repetition that had been previously suggested as a replacement for parallel repetition.
In this variant, given a basic game $G$,
the verifier chooses the challenges $(u_1,v_1),\ldots,(u_t,v_t)$, and sends the questions to the provers as unordered tuples, i.e.\ $U = \set{u_1,\ldots,u_t}$ and
$V = \set{v_1,\ldots,v_t}$. The verifier expects to receive a label for each element in $U$ and each element in $V$, and checks that they satisfy each one
of the constraints $(u_i,v_i)$.
We denote this game by $G^{\otimes_{{\sf sym}} t}$, and note that it
clearly preserves uniqueness; also, we note that the arguments used to refute the strong Parallel Repetition Conjecture do not immediately apply to it.
While a naive application of this variant can still be shown to fail in general,\footnote{This can be seen by considering a graph which is the disjoint union of many odd cycles (instead of a single odd cycle), say $M$, so that one would get a canonical ordering on most subsets of $t$ vertices from this graph, so long as $t=o(\sqrt{M})$.}
there is still a hope that it can be used in a more clever way and establish the equivalence of UGC and the Max-Cut Conjecture.
Our work is partly motivated by seeking such possibilities.

We are thus led to investigate the effect on symmetric repetition on the odd cycle game, and more specifically the symmetric variant of the foam problem
which again is very much related.

\subsection{Our results}
In this paper, our main object of study mainly are tilings of $\mathbb{R}^n$ using a \emph{symmetric body}.
\begin{definition}
  A set $D\subseteq\mathbb{R}^n$ is called symmetric if
  for any $\pi\in S_n$ and $x\in\mathbb{R}^n$, it holds that
  $x\in D$ if and only if $\pi(x)\in D$.
\end{definition}

The main question we consider, is what is the least surface area a symmetric tiling body can have.
Again, one has the trivial example of the solid cube $D = [0,1]^n$, but inspired by the non-symmetric
variant of the problem, one may expect there to be better examples. We first show that while this is
possible, the savings are much milder, and can be at most a multiplicative factor of $\sqrt{\log n}$.
\begin{thm}\label{thm:lb}
	Any symmetric tiling body $D$ of volume $1$
	with piecewise smooth surface has surface area at least $\Omega\left(\frac{n}{\sqrt{\log n}}\right)$.
\end{thm}

Besides the quantitative result itself, we believe the argument used in the proof of Theorem~\ref{thm:lb}
carries with it a lot of intuition regarding the additional challenge that the symmetric variants of the foam problem
and the parallel repetition posses, and we hope that this intuition will help us to develop better understanding
of symmetric parallel repetition in general. We remark that our proof actually shows a lower bound on the ``noise sensitivity''
parameter of the body, which is known to be smaller than the surface area of the body.

We complement Theorem~\ref{thm:lb} with a randomized construction showing that $O(\sqrt{\log n})$ savings
are indeed possible.
\begin{thm}\label{thm:ub}
	There exists a symmetric tiling body $D$ of volume $1$
	with piecewise smooth surface that has surface area $O\left(\frac{n}{\sqrt{\log n}}\right)$.
\end{thm}

Our results also imply tight bounds for the value of the $t$-fold symmetric repetition of the odd cycle game, which we discuss next.

\subsection{Significance of our results for symmetric parallel repetition.}
Using our techniques, one may give sharp estimates to the value of the
$t$-fold symmetric repetition of the odd cycle game, as follows.
\begin{thm}\label{thm:val_sym_cycle_ub}
  There is $c>0$, such that for an odd $n$, if $t \leq c n\sqrt{\log n}$ then
  ${\sf val}(C_n^{\otimes_{{\sf sym}} t})\leq 1 - c\frac{t}{n\sqrt{\log t}}$.
\end{thm}

\begin{thm}\label{thm:val_sym_cycle_lb}
  For all $n, t\in\mathbb{N}$ it holds that
  ${\sf val}(C_n^{\otimes_{{\sf sym}} t})\geq 1 - O\left(\frac{t}{n\sqrt{\log t}}\right)$.
\end{thm}
We remark that a similar connection between the standard foam problem and the value of the $t$-fold repeated game
is well known. More precisely, in~\cite{FKO} the authors show that
(1) tilings of the Euclidean space with small surface area can be used to derive good strategies for $C_n^{\otimes t}$,
and (2) the Euclidean isoperimetric inequality (which gives a lower bound of $\Theta(\sqrt{n})$ on the surface area of a tiling body) can be used
to prove upper bounds on the value of $C_n^{\otimes t}$. We remark that while (1) above is derived in a black-box way, the converse direction, i.e. (2),
is done in a white-box way. That is, the authors in~\cite{FKO} do not actually use the Euclidean isoperimetric inequality, but rather convert one of its
proofs into an upper bound of the value of the $t$-fold repeated odd cycle game.

In contrast to~\cite{FKO}, our proof of Theorems~\ref{thm:val_sym_cycle_ub},~\ref{thm:val_sym_cycle_lb} follow more direct adaptations of the proofs of
Theorems~\ref{thm:lb},~\ref{thm:ub}. This is partly because our arguments  work from scratch and are therefore more flexible.
We outline these adaptations in Section~\ref{sec:val_sym_cycle}.

We believe that Theorem~\ref{thm:val_sym_cycle_ub} gives some new life to the
possible equivalence between the Max-Cut Conjecture and UGC. For example, this would follow if such rate of amplification
would hold for all graphs if we allow for a ``mild'' preprocessing phase first (i.e., preprocessing
that doesn't change the value of the instance by much). For this reason, we believe it would be interesting to investigate
other graph topologies on which symmetric parallel repetition performs well,
and hope that the techniques developed herein will be useful.

On the flip side, Theorem~\ref{thm:val_sym_cycle_lb} asserts that even symmetric parallel repetition on
the odd cycle game admits non-trivial strategies.
Thus, we cannot hope to use it in order to establish the equivalence of weaker forms of the Max-Cut Conjecture and UGC.
Here, by weaker forms of the Max-Cut Conjecture, we mean the conjecture that
${\sf GapMaxCut}[1-\eps,1-\delta(\eps)]$ is NP-hard for small enough $\eps$,
and $\delta(\eps)$ is a nearly linear function of $\eps$, e.g.\ $\delta(\eps) = 100\eps$
or $\delta(\eps) = \eps\sqrt{\log(1/\eps)}$. Given that the best known NP-hardness results
for Max-Cut in this regime are only known for $\delta = (1+\Omega(1))\eps$, this means that there is still
a significant road ahead to establish even the weakest version of the Max-Cut Conjecture that may be useful for UGC.

\subsection{Techniques}
In this section, we explain some of the intuition and idea that go into the proof of Theorems~\ref{thm:lb} and~\ref{thm:ub}, focusing mostly on the former.

Let $D$ be a symmetric tiling body. To prove that the surface area of $D$ is at least $A$, it is enough to prove that $D$ is sensitive to noise rate
$1/A$. I.e., that if we take a point $x\in_R D$, and walk along a random direction $u$ of (expected) length $1/A$, then with constant probability we escape
$A$ at some point on the line $\ell_{x,u}(t) = x+t\cdot u$.

We begin by describing an argument showing a worse bound than the one proved in Theorem~\ref{thm:lb}, which is nevertheless helpful in conveying some of the intuition.
To prove that a random line $\ell_{x,u}(t)$ crosses $D$ with noticeable probability, we argue that for appropriate length of $u$, with constant probability
the line $\ell_{x,u}$ will contain a point in which there are two coordinates differing by a non-zero integer, say $y$ with the coordinates being $i,j$.
Note that this is enough, since then if we assumed that $y\in D$, then the point $y'$ in which the value of coordinates $i,j$ is switched also lies in $D$ (by symmetry),
and then the difference of $y$ and $y'$  is a non-zero lattice vector, so they must be in different cells of the tiling. Therefore we conclude that $y\not\in D$.

With this plan in mind, let $x = (x_1,\ldots,x_n)\in_R D$, and consider the coordinates of $x$ modulo $1$, i.e.\ $B = \set{x_1\pmod{1},\ldots,x_n\pmod{1}}$, as points
in the torus $\mathbb{T}$.
First, it can be shown without much difficulty that they are jointly distributed as uniform random points on $\mathbb{T}$, hence standard probabilistic tools
tell us that any interval of length $100\log n/n$ on the circle contains at least two points from $B$. Now, regardless of how the body $D$ looks like,
there would be two coordinates, say $i$ and $j$, that almost differ by a non-zero integer, yet appear very close when projected on the circle,
i.e.\ in distance at most $100\log n/n$. In this case, with constant probability the coordinates $i,j$ get even closer along a random line $\ell_{x,u}(t) = x+t\cdot u$,
and provided the length of $u$ is long enough to cover the distance between $x_i,x_j$ on the circle (i.e.\ each coordinate of magnitude $\Theta(\log n/n)$),
the line $\ell_{x,u}(t)$ would contain a point as desired.

The above argument can indeed be formalized to yield a lower bound of $\Omega\left(\frac{n}{\log n}\right)$ on the surface area of $D$, but it carries more intuition
than just the bound itself. In a sense, this argument says that if we project $x$ onto the torus, we should be wary of coordinates whose projections are too close,
and make sure that it would only occur if the coordinates themselves are close (as opposed to almost differing by a non-zero integer). Analyzing the event that
two coordinates meet on the circle while being different is easily seen however to not yield a better bound than $\Omega(n/\log n)$, hence to prove Theorem~\ref{thm:lb}
we must look at a different event. That being said, the argument does tell us that we should look at pairwise distances between coordinates of
$x$ when projected on the circle, and in particular on pairs that ``relatively close'' and the way they move along a line in a random direction.

%To capture this intuition however we need to look at more refined events than the event that ``$\ell_{x,u}(t)$ contains a point in which two coordinates differ by a non-zero integer''.
It turns out that it is enough to come up with some parameter that behaves differently on the endpoints of the line,
assuming the line does not escape $D$. This is because that if the escape probability from $D$ is small, then the distributions of
$x$ and $x+u$ are close in statistical distance, and in particular any parameter should behave roughly the same on $x$ and on $x+u$.
Indeed, our proof utilizes an energy function (inspired by the previous argument) that considers the pairwise distances between coordinates of $x$;
the contribution from a pair of coordinates that are in distance $d$ in the circle is proportional to $e^{-Z\cdot d}$, where $Z\sim \frac{n}{\sqrt{\log n}}$.
We show that with high probability, the energy increases along a random line $\ell_{x,u}(t)$ provided it does not escape $D$,
while on the other hand, if the escape probability is small, then $x$ and $x+u$ are close in statistical distance and hence
$\Prob{x,u}{{\sf Energy}(x+u) > {\sf Energy}(x)}\approx \half$. This implies that the escape probability must be constant.

We remark that the above high-level intuition also plays a role in the proof of Theorem~\ref{thm:lb}. I.e., when constructing a symmetric tiling body $D$,
all we really need to care about are the pairwise distances between coordinates, and that we must make sure that somewhat far coordinates will project to far points
on the torus. Indeed, given a point $x\in\mathbb{R}^n$, in order to decide which integer lattice point $y\in\mathbb{Z}^n$ we round $x$ to, we only look at this pairwise distances
of $x$ on the torus. We try to find a point $z$ on the torus that is far from all the coordinates of $x$, and do the rounding according to it.
One naive attempt would be to take $z$ that is furthest from
all coordinates of $x$, however this point turns out to be very noise sensitive and therefore yield a body with large surface area. Instead, we consider a probability
distribution that only puts significant weight on $z$'s that are somewhat far from all $x_i$'s, yet is not too concentrated around the maximizers. Coming up and analyzing
a construction along these lines turns out to require considerable technical effort, and we defer a more elaborate discussion to Section~\ref{sec:ub}

\paragraph{Organization of the paper.} In Section~\ref{sec:prelim}, we set up basic notations and preliminaries.
Section~\ref{sec:lb} is devoted to the proof of Theorem~\ref{sec:lb}, and Section~\ref{sec:ub} is devoted for the proof of Theorem~\ref{sec:ub}.
In Section~\ref{sec:val_sym_cycle} we prove Theorems~\ref{thm:val_sym_cycle_ub},~\ref{thm:val_sym_cycle_lb},
and in Section~\ref{sec:open} we state some open problems.

\section{Preliminaries}\label{sec:prelim}
\paragraph{Notations.}
We write $X\ll Y$ or $X = O(Y)$ to say that
there exists an absolute constant $C >0$ such that $X\leq C\cdot Y$,
and similarly write $X\gg Y$ or $X = \Omega(Y)$ to say that
there exists an absolute constant $c >0$ such that $X\geq c\cdot Y$.
We write $X\ee Y$ or $X = \Theta(Y)$ to say that $Y\ll X\ll Y$.

We denote random variables by boldface letters such as ${\bf x}$ and $\bm{\De}$.
We denote by $\mathcal{N}(\mu,\sigma^2)$ the distribution of a standard
Gaussian random variable with mean $\mu$ and variance $\sigma^2$, and by
$\mathcal{N}(\vec{\mu}, \Sigma)$ the distribution of a multi-dimensional
Gaussian random variable with means $\vec{\mu}$ and covariance matrix $\Sigma$.
\subsection{Needles}

\begin{definition}
  Let $\delta>0$, and let $a\in\mathbb{R}^n$. A random $\delta$-needle is a line $\ell_{a,{\bf u}} = \sett{a+t\cdot {\bf u}}{t\in[0,1]}$
  where the direction vector ${\bf u}$ is a chosen as a standard Gaussian $\mathcal{N}(0,\delta I_n)$.
\end{definition}

Given a tiling body $D$, a random $\delta$-needle from $D$ is a random $\delta$-needle $\ell_{{\bf a}, {\bf u}}$
where ${\bf a}\in D$ is chosen uniformly. Random needles are a useful tool to measure the surface area of a $D$, as shown
in the following two lemmas. First, given a tiling body $D$ and a needle $\ell_{a, u}$, we may think of the needle
as ``wrapping around'' around $D$,  i.e.\ its points are taken modulo $D$. We denote this ``wrapped around'' line by
$\tilde{\ell}_{a, u}$. We will use the following formula from \cite{Santalo}; the case $n=2$ is formula (8.10) therein,
and the extension to general $n$ is discussed in page 274.
\begin{lemma}\label{lem:santalo}
  There is a constant $C_n = \Theta(1)$, such that the following holds.
  Let $S$ be a piecewise smooth surface in a tiling body $D$ of volume $1$, and let $\delta>0$.
  Then
  \[
   \Expect{{\bf a}\in D, {\bf u}\sim\mathcal{N}(0,\delta I_n)}{\card{\tilde{\ell}_{{\bf a}, {\bf u}}\cap S}} = C_n \cdot\sqrt{\delta} \cdot {\sf area}(S).
  \]
\end{lemma}

\begin{lemma}\label{lem:area_bd}
  Let $D$ be a tiling body of volume $1$, and let $\delta>0$. Then
  \[
    \Prob{{\bf a}\in D, {\bf u}\sim\mathcal{N}(0,\delta I_n)}{\ell_{{\bf a},{\bf u}}\cap \partial D\neq \emptyset} \leq \Theta(\sqrt{\delta}) {\sf area}(\partial D).
  \]
\end{lemma}
\begin{proof}
  Set $S = \partial D$, and note that whenever $\ell_{a,\delta u}\cap \partial D\neq \emptyset$, we have that $\card{\tilde{\ell}_{a,\delta u}\cap S}\geq 1$.
  Hence by the previous lemma we get that
  \[
  \Prob{{\bf a}\in D, {\bf u}\sim\mathcal{N}(0,\delta I_n)}{\ell_{{\bf a},{\bf u}}\cap \partial D\neq \emptyset}
  \leq
  \Expect{{\bf a}\in D, {\bf u}\sim\mathcal{N}(0,\delta I_n)}{\card{\tilde{\ell}_{{\bf a},{\bf u}}\cap \partial D}}
  \leq  \Theta(\sqrt{\delta}) \cdot {\sf area}(\partial D).
  \qedhere
  \]
\end{proof}
We will use the above lemma to prove lower bounds on the surface area of a tiling body, by finding $\delta$ such that
the probability on the left hand side of Lemma~\ref{lem:area_bd} is at least $\Omega(1)$; this would imply that
${\sf area}(\partial D)\geq \Omega(1/\sqrt{\delta})$.

\subsection{Basic useful properties of tiling bodies}
\begin{lemma}\label{lem:diff}
  Let $D\subseteq \mathbb{R}^n$ be a symmetric body, such that for all $z\in\mathbb{Z}^n\setminus\set{0}$
  we have $D\cap (D+z)  = \emptyset$, and let $x\in D$. Then for every $1\leq i,j\leq n$,
  if $x_i - x_j \in \mathbb{Z}$, then $x_i = x_j$.
\end{lemma}
\begin{proof}
  Assume towards contradiction $x_i - x_j$ is a non-zero integer $k$, and let $S_{i,j}\in S_n$ be the permutation that
  maps $i$ to $j$, $j$ to $i$ and has any $r\neq i,j$ as a fixed point. Since $D$ is symmetric, we have
  that $S_{i,j}(x)\in D$. Also, we have
  \[
    x - S_{i,j}(x) = (x_i - x_j) (e_i - e_j) = k(e_i-e_j),
  \]
  where $e_i$ is the $i$th element in the standard basis. In other words, we get that $x = S_{i,j}(x) + z$ for non-zero $z\in \mathbb{Z}^n$,
  and therefore $x\in D + z$. This contradict the fact that $D$ and $D+z$ are disjoint.
\end{proof}

\begin{lemma}\label{lem:uniform}
  Let $D$ be a volume $1$ tiling body, and choose $a=(a_1,\ldots,a_n)\in D$ uniformly
  at random. Then the random variable $(a_1\hspace{-1ex}\pmod{1},\ldots,a_n\hspace{-1ex}\pmod{1})$ is uniform over $[0,1)^n$.
\end{lemma}
\begin{proof}
  Sample ${\bf x}\in[0,1)^n$, and take ${\bf a} = {\bf x} \pmod{D}$. Note that the distribution of
  ${\bf a}$ is uniform over $D$. Indeed, for that we note that the map $x\to x\pmod{D}$ is bijection from
  $[0,1)^n$ to $D$:
  otherwise, there were $x\neq x'$ in $[0,1)^n$ that are equal mod $D$, and therefore differ by non-zero lattice point
  (which is clearly impossible). Now as the distribution of ${\bf a}\pmod{1}$ is just ${\bf x}$, the claim follows.
  \end{proof}

\section{The lower bound: proof of Theorem~\ref{thm:lb}}\label{sec:lb}

In this section, we prove the lower bound on the surface area of a
symmetric tiling body $D$.
Throughout, we will have two parameters: $\sigma$, which is magnitude of each coordinates in the needle we consider
(which will be of order $\frac{\sqrt{\log n}}{n}$), and an auxiliary parameter $Z$ (which will be of order $\frac{n}{\log n}$).
Let $D$ be a symmetric tiling body containing $0$.
We denote by ${\bf a}$ a random point in $D$,
and by ${\bf u}$ a Gaussian vector $\mathcal{N}(0,\sigma^2 I_n)$. We will prove that
$\Prob{{\bf a}, {\bf u}}{\ell_{{\bf a},{\bf u}}\not\subseteq D} = \Omega(1)$,
which by Lemma~\ref{lem:area_bd} implies that ${\sf area}(\partial D)\geq \Omega(1/\sigma)$. As
$\sigma = \Theta(\sqrt{\log n}/n)$, this would establish Theorem \ref{thm:lb}.

\paragraph{Notations.}
For $x,y\in\mathbb{R}$, define
\[
d(x,y):= \min_{z\in \ZZ, z\neq 0} |(x+z)-y|\in [0,1].
\]
To gain some intuition for the definition of $d(x,y)$, suppose $x$ and $y$ are two entries of a point $a\in D$.
Clearly, if $d(x,y)$ is small, then $x,y$ nearly differ by an integer $z\neq 0$, and this says that the point $a$ is
somewhat close to the boundary of $D$ (in the sense that Lemma~\ref{lem:diff} could kick in if we move along a direction
that decreases this distance).

Our argument will indeed inspect $d(a_i,a_j)$ for all distinct $i,j\in[n]$ and the way they change along a random direction.
A key measure that we will keep track of is the energy of a point $a\in D$, defined by
\begin{equation*}
\Psi(a):= \sum_{i<j} e^{-Z \cdot d(a_i,a_j)}.
\end{equation*}
We show that for ${\bf a}\in_R D$ and ${\bf u}\sim \mathcal{N}(0,\sigma^2 I_n)$, if $\ell_{{\bf a}, {\bf u}}\subseteq D$ with
probability close to $1$, then the energy of ${\bf a}$ increases along the line $\ell_{{\bf a},{\bf u}}$ with high probability, and
in particular that $\Psi({\bf a}+{\bf u}) > \Psi({\bf a})$. We then argue that with high probability, this should be the case for
the point ${\bf a}$ as well as for ${\bf a} - {\bf u}$, hence $\Psi({\bf a}+{\bf u}) > \Psi({\bf a}-{\bf u})$ with high probability.
This event however can happen with probability at most $0.5$ by symmetry, hence completing the proof.

\subsection{Analyzing the energy along a random line}
By definition of $d(x,y)$, we either have $d(x,y) = (x+z-y)$ or $d(x,y) = -(x+z-y)$ for some $z\in\mathbb{Z}\setminus\set{0}$,
and this sign determines whether $x,y$ need to move in different directions or the same direction for $d(x,y)$ to get smaller.
To capture this, we denote
\[
\ga(x,y):=\left\{ \begin{array}{ll}
+1 & \text{ if $d(x,y)=x+z-y$ for some $z\in\ZZ$, $z\neq 0$},\\
-1 & \text{ otherwise}.	
\end{array} \right.
\]
Next, we discuss the {\em energy} of a configuration, which is the key concept used in the proof.
Let $Z$ be a parameter to be chosen later (of the order $n/\log n$).
As stated earlier, our goal is to analyze the behaviour of $\Psi(a)$ along a random $\sigma^2$-needle from $a$
in direction $u$. Towards this end, note that we expect (at least if $u_i, u_j$ are small) that $d(a_i+u_i, a_j+u_j) = d(a_i,a_j) + \ga(a_i,a_j) (u_j-u_j)$,
hence expect $\Psi(a+u)$ to be close to
\begin{equation*}
\Psi(a,u):=  \sum_{i<j} e^{-Z \cdot (d(a_i,a_j)+\ga(a_i,a_j)\cdot (u_i-u_j))}.
\end{equation*}
Indeed, this is the content of the following claim.

\begin{claim}
	\label{cl:1}
	Suppose $|u_i|\leq 1/20$ for all $i$, and $a+[0,1]\cdot u \subset D$, then
	$$
	|\Psi(a+u)-\Psi(a,u)|\leq n^2\cdot e^{-Z/4}.
	$$
\end{claim}

\begin{proof}
   We consider the contribution of each pair $(i,j)$ to $\Psi(a+u)$ and $\Psi(a,u)$ separately.
   Without loss of generality we may only consider pairs $i,j$ that $\ga(a_i,a_j)=1$,
   and thus $d(a_i,a_j)= a_i - a_j + z$ for some $z\in \ZZ$, $z\neq 0$.
   Let
   \[
   d= a_i - a_j + z + (u_i-u_j) = (a_i+u_i)- (a_j+u_j)+z.
   \]
   First, we argue that $d\geq 0$. Otherwise, since $a_i-a_j + z \geq 0$ it follows by continuity that there is
   $\lambda\in[0,1]$ such that $a_i - a_j + z + \lambda(u_i-u_j) = 0$, and hence the point $a+\lambda u$ has entries
   that differ by an integer $z\neq 0$, and this contradicts Lemma \ref{lem:diff} (as $a+\lambda u \in D$).
   We now consider two cases:
   \begin{itemize}
     \item {\bf Case 1}: $d \in [0,0.5]$. In this case, we have $d(a_i+u_i, a_j+u_j) = d$, and thus the contribution of the pair $(i,j)$ to both sums is the same ($e^{-Z\cdot d}$).
     \item {\bf Case 2}: $d > 0.5$. Since $\card{u_i-u_j}\leq 0.1$, it follows that $d(a_i,a_j) = d -(u_i-u_j) > 0.4$, which implies $d(a_i+u_i,a_j+u_j) > 0.3$. Therefore, the
     contribution to $\Psi(a,u)$ from $i,j$ is at most $e^{-0.4\cdot Z}$ and to $\Psi(a+u)$ is at most $e^{-0.3 \cdot Z}$, and in particular $(i,j)$ contributes (in absolute value)
     at most $e^{-Z/4}$ to the difference between the sums.
   \end{itemize}
   Taking a sum over all pairs $(i,j)$ concludes the proof.
\end{proof}

\subsection{Analyzing the expectation and variance of $\Psi(a,{\bf u})$}
Next,
we consider $\Psi(a,{\bf u})$ as a random variable over the choice of ${\bf u}$
and compute its expectation and variance. In both computations we will use
the well-known fact that $\E[e^{-Z\cdot N(0, c^2)}]=e^{Z^2 c^2/2}$ for all $c> 0$.

\begin{claim}
\label{cl:2}
For every $a\in\mathbb{R}^n$ we have $\Expect{{\bf u}\sim \mathcal{N}(0,\sigma^2 I_n)}{\Psi(a,{\bf u})} =\Psi(a)\cdot e^{(Z \cdot \si)^2}$.
\end{claim}

\begin{proof}
    By linearity of expectation we have that
	\[
	\Expect{{\bf u}\sim \mathcal{N}(0,\sigma^2I_n)}{\Psi(a,{\bf u})}
= \sum_{i<j} e^{-Z \cdot d(a_i,a_j)}\cdot \Expect{{\bf u}\sim \mathcal{N}(0,\sigma^2I_n)}{e^{-Z\cdot \ga(a_i,a_j)\cdot ({\bf u}_i-{\bf u}_j)}}.
	\]
    Note that the above expectation does not depend on $i,j$: for every $i,j$ the distribution of
    ${\bf u}_i - {\bf u}_j$ is $N(0,\si^2)-N(0,\si^2)\sim N(0,2\si^2)$, so it is symmetric around $0$ and thus
    the sign $\ga(a_i,a_j)$ does not affect the expectation. Hence we have
    \[
    \Expect{{\bf u}\sim \mathcal{N}(0,\sigma^2I_n)}{\Psi(a,{\bf u})}   = \Psi(a) \cdot \E [e^{Z\cdot N(0,2\sigma^2)}]
	=\Psi(a)\cdot e^{Z^2\si^2}.\qedhere
    \]
\end{proof}

Next, we turn our attention into upper bounding the variance of $\Psi(a,{\bf u})$, and
for that we first define the notion of {\em good} points $a\in D$ and prove two preliminary
claims.
We say a point $a$ is {\em good} if any interval of length $(10\log n)/n$
on the torus contains at least $\log n$ and at most $100 \log n$ coordinates from $a\hspace{-1ex}\pmod{1}$.
Note by Lemma \ref{lem:uniform}, if ${\bf a}$ is chosen randomly from $D$ then ${\bf a}\pmod{1}$ is uniform
over $[0,1)^n$ and by Chernoff bound is easily shown to be good with probability $>0.999$.

We first show that good points have high energy.
\begin{claim}\label{cl:6}
There exists $c_2>0$, such that for
$Z = 0.1 \frac{\log n}{n}$, if
$a$ is good then $\Psi(a)> c_2 \log^2 n$.
\end{claim}

\begin{proof}
 Partition the torus $[0,1)$ into $m = n/(10 \log n)$ disjoint intervals of length $1/m = (10 \log n)/n$ each.
 We say that $I_i$ is {\em unanimous}, if there is $b_i\in \RR$ (called anchor) such that
 (1) $b_i\hspace{-1ex}\pmod{1}$ is the middle of $I_i$, and (2) for the majority of points $a_j\in I_i$, $|a_j-b_i|< 1/m$.

 We consider two cases:

\noindent {\bf Case 1: There is an interval $I_i$ that is not unanimous.} Note that there are at least $\log n$ coordinates $j$ of $a$ such that $a_j\in I_i$.
Let $j^{\star}$ be such coordinate, and write $a_{j^{\star}} = z_{j^{\star}} + \{a_{j^{\star}}\}$ where $z_{j^{\star}}\in\mathbb{Z}$ and
$\{a_{j^{\star}}\}$ is the fractional part of $a_{j^{\star}}$. Consider
$b = z_{j^{\star}} + m_i$ where $m_i$ is the middle of $I_i$. Then since $I_i$ is not unanimous, $b$ is not an anchor of it and so there are at least $\half \log n$
coordinates of $a$, say $(a_k)_{k\in K_{i,j^{\star}}}$ that mod $1$ are in $I_i$, and $\card{a_k - b}\geq 1/m$. Writing $a_k = z_k + \{a_k\}$, we observe that $z_k\neq z_{j^{\star}}$,
since otherwise $\card{a_k - b}= \card{\{a_k\} - m_i} \leq 1/(2m)$. Hence the difference $a_k-a_{j^{\star}}$
is $10\log n/n$ close to an integer $z_k-z_{j^{\star}}\neq 0$, and so $d(a_k,a_{j^{\star}})\leq 10\log n/n$, and the contribution of $\Psi(a)$ is
at least $e^{-1}$. Summing we get
\[
\Psi(a)\geq
\half
\sum\limits_{j^{\star}: a_{j^{\star}} \in I_i}
\sum\limits_{k\in K_{i,j^{\star}}}{e^{-Z d(a_k,a_{j^{\star}})}}
\geq
\half
\sum\limits_{j: a_j \in I_i}
e^{-1} \card{K_{i,j^{\star}}}
\geq \frac{1}{4e} \log^2 n.
\]

\noindent {\bf Case 2: All intervals are unanimous.} Let $b_i$ be an anchor of $I_i$. Note that since the fractional part of two adjacent anchors,
i.e.\ of $b_i,b_{i+1}$, are $1/m$ apart, we have that either $\card{b_i - b_{i+1}}\leq 1/m$ or $\card{b_i-b_{i+1}}\geq 1-1/m$. We claim there exists $i$ for
which the latter condition holds. To see this, assume that for all $i=1,\ldots,m-1$ we have that the first condition holds. Then we have
$b_i = z + i \frac{10\log n}{n}$ for some $z\in \mathbb{Z}$ for all $i=1,\ldots,m$, and hence $\card{b_m - b_1}\geq 1-1/m$ (and the condition holds for $i=m$).

Thus, we fix $i$ such that $\card{b_i-b_{i+1}}\geq 1-1/m$, and thus $b_i - b_{i+1} = z + \alpha$ for $z\neq 0$ and $\card{\alpha}\leq 1/m$.
Let $K_i$ be the coordinates $j$ of $a$ such that $\card{a_j - b_i}\leq 1/m$ for $j\in K_i$
and similarly define $K_{i+1}$.
We have that $a_r - a_j = z + \alpha + (a_r - b_{i+1}) + (a_j - b_i)$, hence $a_r - a_j = z + \beta$ for $\card{\beta}\leq 3/m$ for all $r\in K_{i+1}$, $j\i K_i$.
Thus $d(a_r,a_j)\leq 3/m$, and we get
\[
\Psi(a)\geq \card{K_i}\card{K_{i+1}} e^{-Z\cdot 3/m}\geq \frac{1}{4} e^{-3} \log^2 n\qedhere
\]
\end{proof}

Let $C_i = \sum_{j\neq i} e^{-Z\cdot d(a_i,a_j)}$
be the contribution of $a_i$ to $\Psi(a)$. Note that $\Psi(a)= \frac{1}{2}\sum_i C_i$.

\begin{claim}\label{cl:7}
There exists $c_3>0$, such that if $a$ is good, then for all $i$ we have
$C_i < c_3 \Psi(a)/\log n$.
\end{claim}
\begin{proof}
	Note that $d(a_i,a_j)\ge | \{a_i\} - \{a_j\}|$. Since any interval of length $10\log n/n$ on the torus contains at most $100\log n$ points of $a$,
    we have that the number of $j$'s such that $|\{a_i\} - \{a_j\}|$ is between $10 \log n/n \cdot k$ and $10 \log n/n \cdot (k+1)$ is at most $200\log n$
    (for all $k$). Therefore,
	\[
	C_i < 200 \log n\cdot \sum_{k=0}^{\infty} e^{-Z \cdot k \cdot (10 \log n )/n} = 200 \log n \cdot \sum_{k=0}^{\infty}e^{-k}
    \leq 400 \log n.
	\]
Using Claim~\ref{cl:6}, we may bound $\log n\leq \frac{1}{c_2}\frac{\Psi(a)}{\log n}$, finishing the proof.
\end{proof}

We are now ready to bound the variance of $\Psi(a,{\bf u} )$.
\begin{claim}\label{cl:3}
There exists $c_1>0$ such that the following holds.
Let $Z= n/10\log n$, let $a\in \mathbb{R}^n$ be good and let ${\bf u} \sim \mathcal{N}(0,\sigma^2 I_n)$.
Then
\[
	{\sf var}_{\bf u} [\Psi(a,{\bf u} )] \leq \frac{c_1}{\log n} \cdot( e^{4(Z \cdot \si)^2}
	-e^{2(Z \cdot \si)^2})\cdot \Psi(a)^2.
\]
\end{claim}

\begin{proof}
Using Claim \ref{cl:2} to compute the expectation of $\Psi(a,{\bf u})$, we have by definition that
\begin{align*}
&{\sf var}_{\bf u}  ( \Psi(a,{\bf u} ))
= \E_{\bf u}  \left[\left(\sum_{i<j} e^{-Z\cdot d(a_i,a_j)}\cdot (e^{Z\cdot\gamma(a_i,a_j)\cdot({\bf u} _i-{\bf u} _j)}-e^{(Z\cdot \si)^2}) \right)^2\right]\\
& =\sum_{i<j} e^{-2Z\cdot d(a_i,a_j)} \cdot\Expect{{\bf u}}{\left(e^{Z\cdot\gamma(a_i,a_j)\cdot({\bf u}_i-{\bf u}_j)}-e^{(Z\cdot \si)^2}\right)^2}\\
& + \sum_{\substack{(i,j,k) \\ \text{distinct}}}e^{-Z\cdot (d(a_i,a_j)+d(a_i,a_k))}\cdot
\Expect{{\bf u}}{(e^{Z\cdot\gamma(a_i,a_j)\cdot({\bf u}_i-{\bf u}_j)}-e^{(Z\cdot \si)^2})(e^{Z\cdot\gamma(a_i,a_k)\cdot({\bf u}_i-{\bf u}_k)}-e^{(Z\cdot \si)^2})}.
\end{align*}
Here, we used that fact that if $i,j,k,r$ are all distinct then $e^{Z\cdot\gamma(a_i,a_j)\cdot({\bf u}_i-{\bf u}_j)}$, $e^{Z\cdot\gamma(a_k,a_r)\cdot({\bf u}_k-{\bf u}_r)}$
are independent with expectation $e^{(Z\cdot \sigma)^2}$, hence the contribution of these terms is $0$. Computing, we see that
\[
\Expect{{\bf u}}{\left(e^{Z\cdot\gamma(a_i,a_j)\cdot({\bf u}_i-{\bf u}_j)}-e^{(Z\cdot \si)^2}\right)^2}
=\E\left[e^{Z\cdot N(0,8\sigma^2)}\right]-e^{2(Z\cdot \si)^2}
=e^{4(Z\cdot \si)^2}-e^{2(Z\cdot \si)^2},
\]
and
\begin{align*}
&\Expect{{\bf u}}{(e^{Z\cdot\gamma(a_i,a_j)\cdot({\bf u}_i-{\bf u}_j)}-e^{(Z\cdot \si)^2})
 (e^{Z\cdot\gamma(a_i,a_k)\cdot({\bf u}_i-{\bf u}_k)}-e^{(Z\cdot \si)^2})}\\
 &=
 \E\left[e^{(\gamma(a_i,a_j) + \gamma(a_i,a_k))Z\cdot N(0,\sigma^2)}\right]
 \E\left[e^{Z\cdot N(0,2\sigma^2)}\right]
 -e^{2(Z\cdot \si)^2}\\
 &\leq
 \E\left[e^{2Z\cdot N(0,\sigma^2)}\right]
 \E\left[e^{Z\cdot N(0,2\sigma^2)}\right]
 -e^{2(Z\cdot \si)^2}\\
 &= e^{3(Z\cdot \sigma)^2}-e^{2(Z\cdot \si)^2}.
\end{align*}
Thus, we get that
\begin{align*}
  {\sf var}_{\bf u} ( \Psi(a,{\bf u}))
  &\leq
 \sum_{i<j} e^{-2Z\cdot d(a_i,a_j)} (e^{4 (Z\cdot \si)^2}-e^{2 (Z\cdot \si)^2}) +
 \sum_{\text{$(i,j,k)$ distinct}} e^{-Z\cdot (d(a_i,a_j)+d(a_i,a_k))}(e^{3(Z\cdot \sigma)^2}-e^{2(Z\cdot \si)^2}) \\
 & \leq
 (e^{4 (Z\cdot \si)^2}-e^{2 (Z\cdot \si)^2})
 \sum_{i}
 \left(\sum_{j\neq i} e^{-2Z\cdot d(a_i,a_j)} + \sum_{j,k\neq i}
e^{-Z\cdot (d(a_i,a_j)+d(a_i,a_k))}\right)\\
&=(e^{4 (Z\cdot \si)^2}-e^{2 (Z\cdot \si)^2})\sum\limits_{i}\left(\sum_{j\neq i} e^{-2Z\cdot d(a_i,a_j)}\right)^2\\
&=  (e^{4 (Z\cdot \si)^2}-e^{2 (Z\cdot \si)^2}) \cdot \sum_i C_i^2.
\end{align*}
Therefore, using Claim~\ref{cl:7} we conclude that
\begin{multline*}
{\sf var}_{\bf u} ( \Psi(a,{\bf u}))
\leq
( e^{4 (Z\cdot \si)^2}-e^{2 (Z\cdot \si)^2}) \frac{c_3\Psi(a)}{\log n}\cdot \sum_i C_i =
\frac{2 c_3}{\log n}\cdot
( e^{4 (Z\cdot \si)^2}-e^{2 (Z\cdot \si)^2}) \cdot \Psi(a)^2.
\end{multline*}
Setting $c_1:=2c_3$ completes the proof.
\end{proof}

Putting the last two claims together, we have:

\begin{claim}
	\label{cl:4}
Let $\sigma = 10^4\sqrt{c_1} \frac{\sqrt{\log n}}{n}$ and let $a\in\mathbb{R}^n$ be good.
Then
\[
\Pr_{{\bf u}}[\Psi(a,{\bf u})>\Psi(a)+\frac{(Z\si)^4}{2}\Psi(a)]\geq 0.96.
\]
\end{claim}
\begin{proof}
	We upper bound the probability of the complement event. Using Claim \ref{cl:2} (and $e^{t}\geq 1+t+t^2/2$ for $t\geq 0$), we get
	\[
	\E_{\bf u}[\Psi(a,{\bf u})]\geq \Psi(a)\cdot \left(1+(Z \si)^2+\frac{(Z\si)^4}{2}\right).
	\]
	Hence
	\[
	\Pr_{{\bf u}}\left[\Psi(a,{\bf u})\leq \Psi(a)+\frac{(Z\si)^4}{2}\Psi(a)\right] \leq \Pr_{{\bf u}}\left[\card{\Psi(a,{\bf u})- \E_{{\bf u'}}[\Psi(a,{\bf u'})]}\geq
\Psi(a)\cdot  (Z \si)^2\right].
	\]

    We want to upper bound the probability of the last event using Chebyshev's inequality.
    Since $a$ is good, the conclusion of Claim~\ref{cl:3} holds.
    Since $Z \si =o(1)$, for large enough $n$ we get
	\[
	{\sf var}_{\bf u} [\Psi(a,{\bf u})]
    \leq
    \frac{c_1}{\log n} \cdot( e^{4(Z \cdot \si)^2}
	-e^{2(Z \cdot \si)^2})\cdot \Psi(a)^2
    \leq \frac{c_1}{\log n} \cdot \Psi(a)^2\cdot 8(Z \si)^2.
	\]
    Therefore, applying Chebyshev's inequality we see the probability in question is at most
    \[
	\frac{{\sf var}_{\bf u} [\Psi(a,{\bf u})]}{\Psi(a)^2 \cdot  (Z \si)^4} \leq
	\frac{c_1 \cdot \Psi(a)^2 \cdot 8 (Z \si)^2}{(\log n)\cdot \Psi(a)^2 \cdot  (Z \si)^4 }
    = \frac{8c_1}{(\log n)\cdot (Z\sigma)^2 } = \frac{4 c_1}{10^2 c_1} = 0.04.\qedhere
    \]
\end{proof}

\subsection{Finishing the argument}
For each $u$, denote $\eps_u = \Prob{{\bf a}\in D}{\ell_{{\bf a},u}\not\subseteq D}$,
and denote $\eps = \Expect{{\bf u}\sim\mathcal{N}(0,\sigma^2 I_n)}{\eps_{\bf u}} = \Prob{{\bf a}, {\bf u}}{\ell_{{\bf a}, {\bf u}}\not\subseteq D}$.
\begin{claim}\label{cl:tv5}
For each $u$, $\mathcal{D}_{TV}[{\bf a}; {\bf a}-u]\leq \eps_u + \eps_{-u}$.
\end{claim}
\begin{proof}
  Let $K$ be a Borel set.
  Note that it is enough to show that
  (1) if $K\subseteq D$ then $0\leq \Prob{{\bf a}\in D}{{\bf a}\in K} - \Prob{{\bf a}\in D}{{\bf a}-u\in K}\leq \eps_u$,
  and (2) if $K\subseteq\bar{D}$, then $-\eps_{-u}\leq \Prob{{\bf a}\in D}{{\bf a}\in K} - \Prob{{\bf a}\in D}{{\bf a}-u\in K}\leq 0$.
  Indeed, given both (1) and (2), the triangle inequality implies for any Borel set $K\subseteq\mathbb{R}^n$,
  \begin{align*}
  &\card{\Prob{{\bf a}\in D}{{\bf a}\in K} - \Prob{{\bf a}\in D}{{\bf a}-u\in K}}\\
  &\leq
  \card{
  \Prob{{\bf a}\in D}{{\bf a}\in K\cap D}-\Prob{{\bf a}\in D}{{\bf a}-u\in K\cap D}
  +\Prob{{\bf a}\in D}{{\bf a}\in K\setminus D}-\Prob{{\bf a}\in D}{{\bf a}-u\in K\setminus D}}
  \leq \eps_u+\eps_{-u}.
  \end{align*}

  To prove (1), note that
  $\Prob{{\bf a}\in D}{{\bf a}\in K}=\mu(K)$
  and
  \[
  \Prob{{\bf a}\in D}{{\bf a}-u\in K} = \Prob{{\bf a}\in D}{{\bf a}\in K+u} = \mu((K+u)\cap D).
  \]
  This is at most $\mu(K+u) = \mu(K)$ (hence the expression in (1) is non-negative)
  and at least $\geq \mu(K+u) - \mu((K+u)\setminus D)= \mu(K) - \mu(K\setminus (D-u))$.
  Therefore
  \[
  0\leq \Prob{{\bf a}\in D}{{\bf a}\in K} - \Prob{{\bf a}\in D}{{\bf a}-u\in K}
  \leq \mu(K\setminus (D-u))
  \leq \mu(D\setminus (D-u))
  = \Prob{{\bf a}\in D}{{\bf a}+u\not\in D}
  \leq \eps_u.
  \]

  To prove (2), note that $\Prob{{\bf a}\in D}{{\bf a}\in K} = 0$ (hence the expression in (2) is non-positive)
  and
  \[
  \Prob{{\bf a}\in D}{{\bf a}-u\in K} \leq \Prob{{\bf a}\in D}{{\bf a}-u\not\in D} \leq \eps_{-u}.\qedhere
  \]
\end{proof}

\begin{claim}
	\label{cl:5}
	$\eps\geq 0.1$.
\end{claim}
\begin{proof}
	Let $E_1$ be the event that ${\bf a} + {\bf u}[0,1]\subseteq D$, let
    $E_2$ be the event that $\Psi({\bf a})\leq 1$, let $E_3$ be the event
	that $|{\bf u}_i|>1/20$ for some $i$ and let $E_4$ be the event that
	$\Psi({\bf a},{\bf u})>\Psi({\bf a})+\frac{(Z\si)^4}{2}\Psi({\bf a})$. Finally, let $E_5$ be the event that
	$\Psi({\bf a}+{\bf u})>\Psi({\bf a})$ and denote $E({\bf a},{\bf u}) = E_1\cap (\neg E_2)\cap (\neg E_3)\cap E_4$.
    Note that if the event $E$ holds for $a,u$, then $E_5$ also holds, since by Claim~\ref{cl:1}:
	\[
	\Psi(a+u)\geq \Psi(a,u)-n^2\cdot e^{-Z/4}> \Psi(a)+\frac{(Z\si)^4}{2}\Psi(a)-n^2\cdot e^{-Z/4}\geq \Psi(a).
	\]

    By Claim \ref{cl:6} the probability of $E_2$ is at most the probability ${\bf a}$ is bad, hence it is at most $0.005$.
    By definition, the probability of $E_1$ is $1-\eps$.
    By the union bound and Chernoff inequality, the probability of $E_3$ is $o(1)$.
	Thus, by Claim~\ref{cl:4} we have
	\begin{equation}
	\Pr_{{\bf a},{\bf u}}[E({\bf a},{\bf u})]\geq 0.96 -\eps- 0.005 - o(1)\geq 0.95-\eps.
	\end{equation}
    Fix $u$. Using Claim~\ref{cl:tv5} we get that
    \[
    \Prob{{\bf a}}{E({\bf a}-u,u)} \geq
    \Prob{{\bf a}}{E({\bf a},u)} - \mathcal{D}_{TV}[{\bf a}; {\bf a}-u]
    \geq \Prob{{\bf a}}{E({\bf a},u)} - \eps_u-\eps_{-u}.
    \]
    By the union bound, we now conclude that
    \[
    \Prob{{\bf a}}{E({\bf a}-u,u)\cap E({\bf a},u)}\geq
    1 -  \Prob{{\bf a}}{\overline{E({\bf a}-u,u)}}
    -  \Prob{{\bf a}}{\overline{E({\bf a},u)}}
    \geq 2 \Prob{{\bf a}}{E({\bf a},u)} -1-\eps_u-\eps_{-u}.
    \]
    Taking expectation over ${\bf u}$, we get that
    \[
      \Prob{{\bf a},{\bf u}}{E({\bf a}-{\bf u},{\bf u})\cap E({\bf a},{\bf u})}\geq
     2\Prob{{\bf a},{\bf u}}{E({\bf a},{\bf u})} -1-2\Expect{{\bf u}}{\eps_{\bf u}}
     \geq 0.9 - 4\eps.
    \]
    Next, when both $E(a-u,u)$ and $E(a,u)$ hold, we have by the previous observation that $E_5$ holds for both pairs $(a-u,u)$ and $(a,u)$,
    and so $\Psi(a+u) > \Psi (a) = \Psi((a-u) + u) > \Psi(a-u)$. Thus, we get that
    $\Prob{{\bf a},{\bf u}}{\Psi({\bf a}+{\bf u}) > \Psi({\bf a}-{\bf u})}\geq 0.9 - 4\eps$. On the other hand, the probability on the left hand side is at most $0.5$;
    this follows as
    $\Prob{{\bf a},{\bf u}}{\Psi({\bf a}+{\bf u}) > \Psi({\bf a}-{\bf u})}
    = \Prob{{\bf a},{\bf u}}{\Psi({\bf a}-{\bf u}) > \Psi({\bf a}+{\bf u})}$ (since the distributions of ${\bf u}$ and $-{\bf u}$ are identical) and their sum is at most $1$.
    Combining the two inequalities we get that $\eps \geq 0.1$.
\end{proof}

\section{The upper bound: proof of Theorem~\ref{thm:ub}}\label{sec:ub}
In this section we prove a matching upper bound on the surface area of a symmetric foam by giving a
(probabilistic) construction of a symmetric tiling body $D$ of surface area $O(n/\sqrt{\log n})$.
The main technical result proved in this section, Lemma~\ref{lem:main_ub}, establishes a weaker statement,
and in Section~\ref{sec:ns_sa_reduce} we show how to deduce Theorem~\ref{thm:ub} from it.

\subsection{Reduction to constructing a rounding scheme}\label{sec:mapping_to_tiling}
Suppose $S$ is function mapping (multi-)sets of $n$ points from $\RR/\ZZ$, to $\RR/\ZZ$.
We further assume that for all (multi-)sets $A$, it holds that $S(A)\not\in\set{0}\cup A$.

Given such $S$, we may extend it to $\mathbb{R}^n$ by
$S(x_1,\ldots,x_n):=S(\{\{x_1\},\ldots,\{x_n\}\})$, where $\{x_i\}$ is the fractional part of $x$.
We can construct a rounding scheme
$R\colon \RR^n\to \ZZ^n$ using $S$ as follows.
\begin{itemize}
	\item
	On input $x=(x_1,\ldots,x_n)$, denote $z=S(x)$ and view
	$z$ as a number in $[0,1)$.
	\item For each $i\in [n]$:
\begin{itemize}
    \item
	if $\{x_i\}\in [0,z)$, set $R(x)_i = \lfloor x_i \rfloor$,
    \item otherwise, $\{x_i\}\in (z,1)$, and set $R(x)_i = \lceil x_i \rceil$.
\end{itemize}
\end{itemize}
First, $R$ is well-defined since $z\notin \{0,\{x_1\},\ldots,\{x_n\}\}$.
Next, note that for any $t\in\ZZ^n$ it holds that $R(x+t)=R(x)+t$, thus $R$ induces
that the body $D = \sett{x}{R(x) = 0}$ is tiling with respect to the lattice $\ZZ^n$.
Last, we note that since for any $\pi\in S_n$ we have that $S(\pi(x))=S(x)$, we also have
that $R(\pi(x))=\pi(R(x))$, and hence $D$ is symmetric.

In our proof we will define a distribution over mappings $S$, and we will want to study
the noise sensitivity of the resulting body $D$ using properties of the mappings $S$.
The following claim gives useful conditions to study noise sensitivity in terms of
mapping $S$.
\begin{claim}\label{cl:u1}
	Let $x$ and $x+\De$ two points in $\RR^n$. Suppose that
	\begin{enumerate}
		\item
		$S(x)=S(x+\De)=:z$; and
		\item
		for all $i$, $\{x_i+\la \De_i\} \neq z$, $\forall \la\in [0,1]$.
	\end{enumerate}
Then the points $x,x+\De$ fall in the same cell in the tiling induced by $D$.
\end{claim}
\begin{proof}
Suppose towards contradiction that the conclusion of the statement does not hold, i.e.\
$x$ and $x+\De$ belong to different cells in the tiling induced by $D$.
Thus, the rounding function $R$ when applied on $x$ and on $x+\De$ should produce
different lattice points, so there is an $i$ such that $R(x)_i \neq R(x+\De)_i$. We fix that $i$ and
assume without loss of generality that $\De_i\ge 0$ and that $x_i\in [0,1)$. We now consider two cases,
depending on the range $x_i$ falls into:
\begin{enumerate}
  \item If $x_i \in [0,z)$, then by definition of $R$ we get that
  $R(x)_i=0$, and $R(x+\De)_i=0$ unless $x_i+\De_i >z$, which leads to a contradiction to the second condition ($z$ is on the interval between $x_i$ and $x_i+\De_i$).
  \item 	If $x_i \in (z,1)$, then $R(x)_i=1$, and  $R(x+\De)_i=1$ unless $x_i+\De_i >1+z$, which again leads  to a contradiction to the second condition ($1+z$ is on the interval between $x_i$ and $x_i+\De_i$).\qedhere
\end{enumerate}
\end{proof}

Our main technical statement is the following lemma.
\begin{lemma}\label{lem:main_ub}
  There exists a distribution over mappings $(S_{\vec{r}})_{\vec{r}}$ ($\vec{r}$ is a vector
  of randomness) such that for small enough $\eps>0$, setting $\sigma = \eps\frac{\sqrt{\log n}}{n}$ we have
\[
\Expect{\vec{r}}{\Prob{{\bf x}, \bm{\De}\sim\mathcal{N}(0,\sigma^2 I_n)}{\text{Conditions of Claim~\ref{cl:u1} hold for ${\bf x}$ and ${\bf x}+\bm{\De}$}}}
\geq 1 - O(\eps).
\]
\end{lemma}
Deducing from Theorem~\ref{thm:ub} from Lemma~\ref{lem:main_ub} mostly involves measure-theoretic arguments,
and we defer this deduction to Section~\ref{sec:ns_sa_reduce}. We will actually need the following slightly
more informative version of Lemma~\ref{lem:main_ub} above, using the reduction from
mappings to tilings presented in the beginning of this section, and an inspection of the
bodies $D_{\vec{r}}$ our proof gives.

\begin{lemma}\label{lem:main_ub_use}
  There exists a distribution over tiling bodies $(D_{\vec{r}})_{\vec{r}}$ such that

  \begin{enumerate}
    \item For small enough $\eps>0$, we have
    \[
    \Expect{\vec{r}}{\Prob{{\bf x}, \bm{\De}\sim\mathcal{N}(0,\eps^2 I_n)}{\text{At least one of the conditions of Claim~\ref{cl:u1} fail for ${\bf x}$ and ${\bf x}+\bm{\De}$}}}
    \ll \frac{n}{\sqrt{\log n}} \eps.
    \]

    \item For each $\vec{r}$, $D_{\vec{r}}$ is a countable union of semi-algebraic sets (i.e., sets defined by finitely many polynomial inequalities).
  \end{enumerate}
\end{lemma}

\subsection{The construction of $S_{\vec{r}}$}
\subsubsection{Overview}
Before jumping into the technical details, we start with some intuition. Recall that on input $x$ (a set of $n$ points from $\RR/\ZZ$)
we must output a number $z\in\RR/\ZZ$, and our goal is to minimize the probability so that the conditions of Claim~\ref{cl:u1} fail
on a short needle $\ell_{x,\De}$. Note that it would not be beneficial for us to choose $z$ that are close to $x_i$. For example,
if we chose $z$ such that $\card{x_i-z_i}\leq \sigma$, then there is constant probability that the interval $\set{x_i + \lambda \De_i}_{\lambda\in[0,1]}$
would contain the point $z$, i.e.\ the second condition of Claim~\ref{cl:u1} would fail.

Thus, a natural candidate for the choice of $z$ would be the one that maximizes $\min_{i\in[n]}{\card{x_i - z_i}}$. It is not hard to see that this minimum
is typically of the order $\log n/n$, so intuitively the second condition of Claim~\ref{cl:u1} should hold with probability $\geq 1-\eps$. However, such choice
for $z$ would not be very stable: it is typically the case that there are numerous $z_1,\ldots,z_r$ that nearly achieve this maximum, thus the maximizer among
them could change when looking at $x+\De$ (i.e., this event would happen with probability significantly more than $\eps$), leading to a failure of the first
condition of Claim~\ref{cl:u1}.

We must therefore assign each one of the near-maximizer $z_1,\ldots,z_r$ some weight, so that the weight of each one of them does not significantly change when
moving to $x+\De$. A general form of construction of this type is to design a scoring function $f\colon[0,\infty]\to[0,1]$, and given an input $x$ to assign
the weight $w(z) = \prod\limits_{i}{f(\card{x_i - z})}$ to each $z$, and sample $z$ with probability proportional to $w(z)$.

We remark that this general recipe essentially captures our (natural) attempts so far. On the one hand, we want $f$ to penalize $z$
if it is very close to $x_i$, hence we want $f(t)$ at least mildly increasing. On the other hand, if $f$ is very sharply increasing (e.g exponential),
then one runs into the same problems as we had when we thought of picking $z$ that maximizes $\min_{i\in[n]}{\card{x_i - z_i}}$. We are thus
led to consider ``mildly increasing'' scoring functions $f$, and polynomials turn out to be good choice. Indeed, our scoring function $f$ will
be ``trivial'' if $\card{x_i - z}$ is too small or too large (i.e.\ it'll be $0$ if $\card{x_i-z}\leq \frac{\log n}{50 n}$ and $1$ if
$\card{x_i-z}\geq \frac{\log n}{25 n}$), and otherwise behaves cubically.

\subsubsection{A basic scoring function}
Our construction of $(S_{\vec{r}})_{\vec{r}}$ uses a non-negative scoring function $f$ with the following properties.
\begin{fact}\label{fact:std_fn}
There exists a function $f\colon [0,\infty)\to [0,1]$ that is twice differentiable with
continuous second derivative with the following properties:
\begin{enumerate}
  \item $f(t) = 0$ if $t\leq 1$.
  \item $f(t) = 1$ if $t\geq 2$.
  \item $f(t) \ee (t-1)^3$ if $1\leq t\leq 2$.
  \item $\card{f'(t)}\ll t^2$ and $\card{f''(t)}\ll t$ for all $t$.
\end{enumerate}
\end{fact}
Exhibiting function $f$ as in Fact~\ref{fact:std_fn} is not hard, and we omit the proof.
The function $f$ defined by $f(t) = (t-1)^3$ if $1\leq t\leq 2$ and $f(t) = 0$ for $t\leq 1$,
$f(t) = 1$ for $t\geq 2$ is almost enough, except that it is not differentiable at $t=1$. One
can fix by convolving a smooth bump function with compact support.

Next, we wish to define the mapping $S_{\vec{r}}$. We view the input $x$ as a multi-set, and
the randomness vector $\vec{r}$ as an infinite sequence of $(i,h)$ where $i$ is a uniformly
random element from $[m]$ and $h$ is a uniform real-number from $[0,1]$.

Set $m = n^{1/3}$, partition the circle
the circle $\RR/\ZZ$ into $m$ intervals of length $1/m$ each, $I_j:= \left[\frac{j-1}{m},\frac{j}{m}\right]$, and
let $z_j = \frac{j-1/2}{m}$ be the middle of $I_j$.
It will be convenient for us to define $g_j(t) = f(\frac{50n}{\log n}\card{t-z_j})$,
and subsequently $r_j(x):= \prod\limits_{y\in I_j\cap x} g_j(y)$.
There two cases:

\paragraph{Case (A): $r_i(x)\neq 0$ for some $i\in [m]$.}
In this case, we define a probability distribution $p_i(x)$ over the $i$'s proportionally to the $r_i(x)$'s, i.e.\
we define $p_i(x):=\frac{r_i(x)}{\sum_i r_i(x)}$. We now perform correlated sampling of $i\in[m]$ according to $p_i(x)$ using the randomness vector $\vec{r}$.
More precisely, we go over the randomness vector $\vec{r} = (i_1,h_1),(i_2,h_2),\ldots$ and find the smallest $j$ such that
$h_j\leq p_{i_j}(x)$, in which case we choose $i = i_j$. We then define $S_{\vec{r}}(x) = z_{i_j}$.

\paragraph{Case (B): $r_i(x) = 0$ for all $i\in [m]$.}
If $1/2\not\in x$, we define $S_{\vec{r}}(x) = 1/2$.
Otherwise, we define $S_r(x)=z$, where $z$ is the first element from $\{\frac{1}{4n}, \frac{3}{4n},\ldots,\frac{4n-1}{4n}\}$ that is at least $\frac{1}{4n}$-away from all the
entries of $x$.

\subsection{Estimating $g_j$ on close points}
\begin{fact}\label{fact:first_order_g}
Let $j\in[m]$ and
$x_i\in [z_j - \frac{\log n}{25 n}-\eps^{0.95},z_j+\frac{\log n}{25 n}+\eps^{0.95}]\setminus [z_j - \frac{\log n}{50 n},z_j+\frac{\log n}{50 n}]$,
$\De_i\in\mathbb{R}$, and denote
$\alpha_i = {\sf dist}\left(x_i, [z_j - \frac{\log n}{50 n}, z_j+\frac{\log n}{50n}]\right)$.

\begin{enumerate}
  \item If $\alpha_i\geq 2\card{\De_i}$, then $\card{g_j(x_i+\De_i) - g_j(x_i)}\ll \frac{\card{\De_i}}{\alpha_i} g_j(x_i)$.
  \item In general, $\card{g_j(x_i+\De_i) - g_j(x_i)}\ll n^3(\alpha_i^3+\card{\De_i}^3)$.
\end{enumerate}

\end{fact}
\begin{proof}
     Using
     Taylor's approximation with remainder, there is $y_i\in [x_i,x_i + \De_i]$ such that
     $g_j(x_i+\De_i) = g_j(x_i) + g_j'(y_i)\De_i$, hence
     \begin{align*}
     \card{g_j(x_i+\De_i) - g_j(x_i)}
     \ll \card{\De_i} \card{g_j'(y_i)}
     &\ll \card{\De_i}\frac{50n}{\log n}f'\left(\frac{50n}{\log n}\card{y_i-z_j}\right)\\
     &\ll \frac{n}{\log n} \card{\De_i} \left(\frac{50n}{\log n}\card{y_i-z_j}-1\right)^2.
     \end{align*}
     For the second item, since $y_i\in [x_i,x_i+\De_i]$,
     we get that $\card{\frac{50n}{\log n}\card{y_i-z_j}-1}\leq \frac{50n}{\log n}\left(\alpha_i + \card{\De_i}\right)$,
     and plugging that in yields
     \[
     \card{g_j(x_i+\De_i) - g_j(x_i)}
     \ll n^3\card{\De_i}(\alpha_i^2 + \De_i^2)
     \ll n^3(\alpha_i^3+\card{\De_i}^3),
     \]
     where the last inequality holds as $ab\ll a^3 + b^{3/2}$ for all $a,b>0$ (Young's inequality).
     For the first item, note that since $y_i\in [x_i,x_i+\De_i]$ we get that
     $\left(\frac{50n}{\log n}\card{y_i-z_j}-1\right)\geq \frac{50n}{\log n}\left(\alpha_i - \card{\De_i}\right)$,
     and by the lower bound on $\alpha_i$ this is $\geq \frac{25n}{\log n}\alpha_i$. Therefore we may continue
     as
     \[
     \card{g_j(x_i+\De_i) - g_j(x_i)}
     \ll \frac{n}{\log n} \card{\De_i} \left(\frac{50n}{\log n}\card{y_i-z_j}-1\right)^2
     \ll \frac{\card{\De_i}}{\alpha_i}\left(\frac{50n}{\log n}\card{y_i-z_j}-1\right)^3.
     \]
     Also, we have that
     $\left(\frac{50n}{\log n}\card{y_i-z_j}-1\right)\leq \frac{50n}{\log n}\left(\alpha_i + \card{\De_i}\right)
     \ll \frac{n}{\log n} \alpha_i$, so
     \[
     \card{g_j(x_i+\De_i) - g_j(x_i)}
     \ll \frac{\card{\De_i}}{\alpha_i}\left(\frac{n}{\log n} \alpha_i\right)^3
     \ll \frac{\card{\De_i}}{\alpha_i}g_j(x_i).\qedhere
     \]

\end{proof}
\subsection{Analysis of the construction}\label{sec:analysis_rare}
In this section we prove that Lemma~\ref{lem:main_ub} holds for the construction of $S_{\vec{r}}$ from the last section, and for
that we show that for small enough $\eps$, the expected probability of the complement event is $O(\eps)$, i.e.\ that
\begin{equation}\label{eq:main_ub}
\Expect{\vec{r}}{\Prob{{\bf x}, \bm{\De}\sim\mathcal{N}(0, \sigma^2 I_n)}
{\text{One of the conditions in Claim~\ref{cl:u1} fails for ${\bf x}$ and ${\bf x}+\bm{\De}$}}}
\ll \eps.
\end{equation}
We will think of $\eps$ as very small (say $\eps \leq 2^{-n^2}$), and analyze the contribution of $x$'s from case (A) and case (B) separately.
Case (A) is the main case that occurs often, and case (B) should be thought of rare.

\subsubsection{Analysis of case (B)}
First, we show that the probability ${\bf x}$ (or equivalently ${\bf x}+\bm{\De}$) falls into Case~(B) is at most $n^{-\omega(1)}$.
For this, it will be helpful for us to sample ${\bf x}$, a multi-set of $n$ uniformly chosen numbers in $[0,1]$ in the following equivalent way:
\begin{itemize}
	\item
	Sample ${\bf t}_1,\ldots,{\bf t}_m$ --- where ${\bf t}_i$ is the number of $i$'s such that ${\bf x}_i$'s that fall into interval $I_i$.
	\item Sample ${\bf t}_i$ points uniformly from $I_i$, for each $i=1,\ldots,m$.
\end{itemize}
Note that $\E[{\bf t}_i]=n/m$, hence by Chernoff bound $\Pr[{\bf t}_i\ge 2 \cdot n/m] =e^{-\Omega(n/m)}=n^{-\omega(1)}$. Thus, by
the union bound we have that
\[
\Prob{}{\forall i ~{\bf t}_i\leq 2\cdot n/m} = 1 - n\cdot n^{-\omega(1)} = 1-n^{-\omega(1)}.
\]
Next, we condition on ${\bf t}_i = t_i$, and assume that indeed $t_i<2\cdot n/m$ for all $i$. Let $E_i$ be the event that $r_i(x)=0$.
Note that conditioned on ${\bf t}_i = t_i$, the $E_i$'s are independent and that
\begin{align}
\Pr[\neg E_i|~t_1,\ldots,t_m]
=\Prob{{\bf a}\in I_i}{\frac{50 n}{\log n}\card{{\bf a}-z_i}\leq 1}^{t_i}
=\left( 1-\frac{\log n/25n}{1/m} \right)^{t_i}
&\geq
\left( 1-\frac{m \log n}{25 n} \right)^{2 \cdot n/m}\notag\\\label{eq:notEi}
&\geq e^{-4\log n/25}
= n^{-4/25},
\end{align}
where we used the fact that $e^{-\delta}\leq 1-\delta/2$ for small enough $\delta>0$.
Therefore,
$$
\Pr[E_1\wedge E_2 \wedge \ldots \wedge E_m|~t_1,\ldots,t_m] 
\leq (1-n^{-4/25})^m
=(1-n^{-4/25})^{n^{1/3}} = n^{-\omega(1)},
$$
as long as the $t_i$'s satisfy the condition $t_i<2\cdot n/m$. Therefore, the overall probability of case~(B) is $n^{-\omega(1)}$.

Next, we analyze the probability that the conditions of Claim~\ref{cl:u1} fail given we are in case (B).
Note that if the conditions of Claim~\ref{cl:u1} fail to hold, then either (I) exactly one of ${\bf x}$, ${\bf x}+\bm{\De}$ falls under Case~(B),
or (II) both ${\bf x}$ and ${\bf x}+\bm{\De}$ fall under Case~(B),  but $1/2\in {\bf x}+\la \bm{\De}$ for some $\la\in [0,1]$. We'll bound these cases separately.

\paragraph{Case (I).} Assuming ${\bf x}$ is under Case~(B), we know that each
of the $m$ intervals of the form $J_i := [z_i-\frac{\log n}{50n},z_i+\frac{\log n}{50n}]$ contains at least one point from ${\bf x}$.
Let ${\bf x}_i$ be that point (if there are multiple, pick one at random). Then ${\bf x}_i$ is uniformly distributed in $J_i$.
Therefore, the probability of ${\bf x}_i+\bm{\De}_i$ is outside $J_i$, where
$\bm{\De}_i \sim N(0,\sigma^2)$ and $\sigma^2\leq \eps^2$, is $O(\eps)$.
Given case (B) occurs with probability $\leq n^{-\omega(1)}$, we conclude that its contribution
to the conditions of Claim~\ref{cl:u1} failing is at most
\[
O(m \eps)\cdot n^{-\omega(1)} = O(\eps).
\]

\paragraph{Case (II).}
Fix $\bm{\De} = \De$, and consider ${\bf x}_j$ conditioned on being in case (B).
If ${\bf x}_j$ is in one of the intervals $J_i$, then its distribution is uniform over $J_i$, in which
case we get that the probability $1/2$ falls inside the interval $[{\bf x}_j,{\bf x}_j + \De_j]$ is at most
$m\card{\De_j}$. If ${\bf x}_j$ is not in one of the intervals $J_i$, then it is distributed uniformly
on $[0,1]\setminus\cup_{i=1}^{m} J_i$, and the probability $1/2$ is in $[{\bf x}_j,{\bf x}_j + \De_j]$ is
at most $2\card{\De_j}\leq m\card{\De_j}$.

Therefore by the union bound,
\[
\cProb{{\bf x}}{{\sf case (B)}, \De}{\exists j\in [n]~1/2\in [{\bf x}_j,{\bf x}_j+\De_j]}
\leq m\sum\limits_{j=1}^{n}{\card{\De_i}}.
\]
Taking expectation over $\De\sim\mathcal{N}(0,\sigma^2 I_n)$ and
using Cauchy-Schwarz we get that
\[
\cProb{{\bf x},\bm{\De}}{{\sf case (B)}}{\exists j\in [n]~~1/2\in [{\bf x}_j,{\bf x}_j+\bm{\De}_j]}
\leq m\Expect{\bm{\De}}{\sum\limits_{j=1}^{n}{\card{\bm{\De}_i}}}
\leq m\sqrt{n}\sqrt{\Expect{\bm{\De}\sim\mathcal{N}(0,\sigma^2 I_n)}{\norm{\bm{\De}}_2^2}}
= mn\sigma.
\]
Therefore, the contribution of this case is upper bounded as
\[
\Prob{{\bf x},\bm{\De}}{{\sf case (B)}\land \exists j\in [n]~1/2\in [{\bf x}_j,{\bf x}_j+\De_j]}
\leq \Prob{{\bf x},\bm{\De}}{{\sf case (B)}}mn\sigma
=n^{-\omega(1)}\cdot \sigma = O(\eps).
\]

\subsubsection{Analysis of case (A)}
We now analyze the contribution of $x$'s that fall into case (A) to the left hand side of~\eqref{eq:main_ub}.
\paragraph{Case~(A), Condition~2.}  If  ${\bf x}$  falls under Case~(A),
then the distance from all ${\bf x}_i$'s to $z=S({\bf x})$ is at least $\frac{\log n}{100 n}$.
Therefore, Condition~2 holds as long as $|\bm{\De}_i| < \frac{\log n}{100 n}$ for all $i$.
Since for each $i$ we have that
\[
\Prob{\bm{\De}\sim\mathcal{N}(0,\sigma^2 I_n)}{\card{\bm{\De}_i}\geq \frac{\log n}{100 n}}
=
\Prob{\bm{\De}\sim\mathcal{N}(0,\sigma^2 I_n)}{\bm{\De}_i^2\geq \frac{\log^2 n}{100^2 n^2}}
\ll \frac{\sigma^2}{\log ^2 n/n^2}
\ll \eps^2/\log n,
\]
we get by the union bound that
\[
\Prob{\bm{\De}\sim\mathcal{N}(0,\sigma^2 I_n)}{\exists i~\card{\bm{\De}_i}\geq \frac{\log n}{100 n}}
\ll n\eps^2/\log n
\ll \eps,
\]
for a sufficiently small $\eps$.

\subsection{Case~(A), Condition~1.}
This is the main part of the proof.
We show that in case (A), the probability that $S_{\vec{r}}({\bf x})\neq S_{\vec{r}}({\bf x}+\bm{\De})$
is at most $O(\eps)$. Note that the procedure describing $S_{\vec{r}}$ in this case
is the correlated sampling  procedure of Holenstein~\cite{Holenstein}, where
$S_{\vec{r}}(x)$ samples $i$ according to
the distribution $p(x) = (p_1(x),\ldots, p_m(x))$ and
$S_{\vec{r}}(x+\De)$ samples $i$ according to the distribution $p(x+\De)$. Therefore,
the probability they sample different $i$'s is at most the statistical distance
between the distributions, $\norm{p(x) - p(x+\De)}_1$. Therefore, we must show that
\begin{equation}\label{eq:mainA1}
\cbExpect{{\bf x},{\bf x}+\bm{\De}}{{\sf case (A)}}{\norm{p({\bf x}) - p({\bf x}+\bm{\De})}_1} = O(\ve).
\end{equation}

Before we turn to this task, we upper bound the contribution from several rare cases.

\subsubsection{Contribution from some rare cases}
First, we show that the case some $\bm{\De}_i$ is too large contributes at most $O(\eps)$
to the LHS of~\eqref{eq:mainA1}.
\begin{claim}\label{cl:A-1}
	$\Prob{\bm{\De}\sim\mathcal{N}(0,\sigma^2 I_n)}{\card{\bm{\De}_i}\geq \eps^{0.95}/n~\text{for some $i$}}\leq \eps$.
\end{claim}
\begin{proof}
For each $i$, we have that
\[
\Prob{\bm{\De}\sim\mathcal{N}(0,\sigma^2 I_n)}{\card{\bm{\De}_i}\geq \eps^{0.95}/n}
\leq 2^{-\Omega((\eps^{0.95}/n)^2/\sigma^2)} = 2^{-\Omega\left(\frac{1}{\eps^{0.1}\log n}\right)}
\leq \frac{\eps}{n},
\]
for small enough $\eps$,
and the claim follows from the union bound.
\end{proof}
From now on, we assume that the $\bm{\De}_i$'s are distributed from $N(0,\sigma^2)|_{|\bm{\De}_i| < \eps^{0.95}/n}$.
In particular, we can assume that if ${\bf t}_j$ is the number of $\bm{x}$'s that fall into interval $I_j$, these numbers stay the same under ${\bf x}+\bm{\De}$.
\footnote{Strictly speaking, $x_i + \De_i$ may be in a different interval than $x_i$, but in this case it doesn't affects the distribution $p(x)$.
Indeed, suppose $x_i$ is in $I_j$ but $x_i + \De_i$ is in $I_{j+1}$. Then
$\card{x_i + \De_i - z_{j+1}}\geq \card{z_{j+1} - j/m} - \card{\De_i} - \card{x_i - j/m}
\geq 1/m - 2\eps^{0.95}/n\geq 1/m - \eps^{0.95}$. Therefore, $\frac{50n}{\log n}\card{x_i + \De_i - z_{j+1}} > 2$, and
so $f(\frac{50n}{\log n}\card{x_i + \De_i - z_{j+1}}) = 1$.
}
Next, we handle the case in which $p({\bf x})$ is supported only on a single $j$.
Note that in this case, if $p({\bf x}+\bm{\De})$ is also only supported on this single $j$,
then the contribution of these cases to the LHS of~\eqref{eq:mainA1} is $0$.
We show that the contribution from the other case is $O(\eps)$.
\begin{claim}\label{cl:A-3}
	\[\Prob{{\bf x},\bm{\De}}{\text{$\exists j^{\star}$ such that $p({\bf x})$ is only supported on $j^{\star}$
        and the support of $p({\bf x}+\bm{\De})$ is different}}\ll\eps.\]
\end{claim}
\begin{proof}
    In case (B), we have shown that the probability that $r_j({\bf x}) = 0$ for all $j$ is $n^{-\omega(1)}$, and the same argument
    shows that the probability $r_j({\bf x}) = 0$ for all but a single $j^{\star}$ is still $n^{-\omega(1)}$. Denote this event by $E$.

    Let us condition on the event $E$, on $j^{\star}$ and the number ${\bf t}_1,\ldots,{\bf t}_m$ of ${\bf x}_i$'s that fall into $I_1,\ldots,I_m$.
    Note that for each $j\neq j^{\star}$, since $r_j({\bf x}) = 0$ there is $i$ such that ${\bf x}_i\in J_j \defeq [z_{j} - \frac{\log n}{50 n}, z_{j} + \frac{\log n}{50 n}]$,
    and we condition on that $i_j$ for each $j$ (if there is more than one, we choose one arbitrarily). Note that the distribution of ${\bf x}_{i_j}$
    is thus uniform over $J_j$.

    Now note that if for each $j\neq j^{\star}$ it holds that ${\bf x}_{i_j} + \bm{\De}_{i_j}\in J_j$, then $r_j({\bf x}+{\bf \De}) = 0$, so the only contribution to
    the probability of the event in question comes when ${\bf x}_{i_j} + \bm{\De}_{i_j}\not\in J_j$ (or from case (B), which we have already accounted for earlier).
    Conditioned on $\bm{\De} = \De$, the probability for that is
    at most
    \[
    \Expect{({\bf x}_{i_j})_{j\neq j^{\star}}}{\sum\limits_{j\neq j^{\star}}1_{{\bf x}_{i_j} + \De_{i_j}\not\in J_j}}
    = \sum\limits_{j\neq j^{\star}}\Expect{{\bf x}_j}{1_{{\bf x}_{i_j} + \De_{i_j}\not\in J_j}}
    \leq \sum\limits_{j\neq j^{\star}}\frac{\card{\De_j}}{\log n/(50n)},
    \]
    therefore taking expectation over $\De$ and using Cauchy-Schwarz we get that
    \[
    \Expect{\bm{\De},({\bf x}_{i_j})_{j\neq j^{\star}}}{\sum\limits_{j\neq j^{\star}}1_{{\bf x}_{i_j} + \bm{\De}_{i_j}\not\in J_j}}
    \ll
    \frac{n}{\log n}\sqrt{m}\sqrt{\sum\limits_{j\neq j^{\star}}
    \Expect{\bm{\De}}{\card{\bm{\De}_j}^2}}
    \leq \frac{n}{\log n}\sqrt{m}\sqrt{m\sigma^2} \leq n^2\sigma.
    \]
    Therefore, we get that
    \[
    \Prob{{\bf x},\bm{\De}}{\text{$p({\bf x})$ is only supported on $j^{\star}$,
        but the support of $p({\bf x}+\bm{\De})$ is different}}
        \leq \Prob{}{E}n^2\sigma
        \leq n^{-\omega(1)} n^2 \sigma
        \ll \eps.\qedhere
    \]
\end{proof}
Let $E$ be the event that the support of $p({\bf x})$ consists of at least two distinct $j$'s. We condition on the event $E$ in the subsequent argument.
The following claim shows that conditioned on $E$, the sum of the $r_j({\bf x})$'s is at least somewhat bounded away from $0$. It will only come into play
later in the proof.
\begin{claim}\label{cl:A-2}
  $\cProb{{\bf x}}{E}{\sum\limits_{j} r_j({\bf x}) \leq \eps^{1.6}} \ll \eps$.
\end{claim}
\begin{proof}
Since we conditioned on $E$, there are $j_1\neq j_2$'s such that
$r_{j_1}({\bf x}),r_{j_2}({\bf x}) > 0$. We condition on $j_1$ and $j_2$, and assume without loss of generality that
$j_1 = 1,j_2=2$. We show that
\begin{equation}\label{eq:cl:A-2:1}
\cProb{{\bf x}}{r_1({\bf x}),r_2({\bf x})>0}{r_{1}({\bf x})<\ve^{1.6} \wedge r_{2}({\bf x})<\ve^{1.6}} \ll \eps,
\end{equation}
and thus the result would follow.

Let ${\bf t}_1$ be the number of $i$'s such that ${\bf x}_i\in I_{1}$, and ${\bf t}_2$ be the number of $i$'s such that ${\bf x}_i\in I_{2}$.
Note that ${\bf t}_1, {\bf t}_2 \leq n$.
In addition, conditioned on ${\bf t}_1=t_1$ and ${\bf t}_2=t_2$, the events $r_1({\bf x})<\ve^{1.6}$ and $r_2(x)<\ve^{1.6}$ become independent.
Therefore, to prove \eqref{eq:cl:A-2:1}, it suffices to show for all $t_1\leq n, t_2\leq n$,
  	\begin{equation}
  \cProb{{\bf x}}{r_1({\bf x})>0,~{\bf t}_1=t_1,~{\bf t}_2=t_2}{r_1({\bf x})<\ve^{1.6}}\ll \eps^{0.5}.
  \label{eq:cl:A-2:2}
  \end{equation}
  Note that one way to sample $r_1({\bf x})| r_1({\bf x})>0,~{\bf t}_1=t_1$ is as follows.
  \begin{itemize}
  	\item Sample points ${\bf x}_1,\ldots,{\bf x}_{t_1}$ uniformly from $I_1$ conditioned on $|{\bf x}_i-z_1| > \frac{\log n}{50 n}$;
  	\item
  	$r_1({\bf x}) = \prod_{i=1}^{t_1} g_1({\bf x}_i)$.
  	\end{itemize}
  Let ${\bf Y}_i$ be the random variable ${\bf Y}_i := g_1({\bf x}_i)^{-0.32}$, where
  ${\bf x}_i$ is sampled as above (we need $0.32<1/3$).
  Let $E$ be the event that $\card{{\bf x}_i - z_1}\geq \frac{\log n}{25 n}$.
  If $E$ holds, then we get that $g_1({\bf x}_i) = 1$, and
  otherwise $g_1({\bf x}_i)\gg \card{\frac{50n}{\log n}\card{{\bf x}_i - z_1} - 1}^3$,
  so
  \[
  \Expect{}{{\bf Y}_i}
  \leq \Prob{}{E}\cdot 1 + \Prob{}{\bar{E}}\cExpect{}{\bar{E}}{g_1({\bf x}_i)^{-0.32}}
  \ll 1 + \cExpect{}{\bar{E}}{\card{\frac{50n}{\log n}\card{{\bf x}_i - z_1} - 1}^{-0.96}}.
  \]
  We write the last expectation as an integral, noting that $\card{{\bf x}_i-z_1}$ is distributed uniformly on
  $\left[\frac{\log n}{50 n},\frac{\log n}{25 n}\right]$, hence
  \[
  \cExpect{}{\bar{E}}{\card{\frac{50n}{\log n}\card{{\bf x}_i - z_1} - 1}^{-0.96}}
  \ll\frac{n}{\log n}\int_{\frac{\log n}{50n}}^{\frac{\log n}{25n}}{\card{\frac{50n}{\log n}t - 1}^{-0.96} dt}
  =\frac{1}{50}\int_{0}^{1}{y^{-0.96} dt}
  \ll 1,
  \]
  where we made the change of variables $y = \frac{50n}{\log n}t - 1$.
  Thus, $\E[{\bf Y}_i] \ll 1$, and so there is a constant $B$ such that $\E[{\bf Y}_i]\leq B$.
  Therefore by independence $\E\left[\prod_{i=1}^{t_1} {\bf Y}_i\right]\leq B^{t_1}\leq B^n$,
  and so writing $r_1({\bf x})$ in terms of the ${\bf Y}_i$'s and using Markov's inequality we get that
  \[
    \cProb{{\bf x}}{r_1({\bf x})>0,~{\bf t}_1=t_1,~{\bf t}_2=t_2}{r_1({\bf x})<\ve^{1.6}}
    = \Pr\left[\prod_{i=1}^{t_1} {\bf Y}_i > \ve^{-1.6\times 0.32} \right] \leq B^n\cdot \ve^{0.512} \ll \eps^{0.5}.
  \qedhere
  \]
\end{proof}

\subsubsection{Analyzing the typical case}
To expand out $\norm{p({\bf x}) - p({\bf x}+\bm{\De})}_1$, we will be using the following claim. The set-up one should
have in mind is that $r_j = r_j(x)$ and $d_j = r_j(x+\De)$ for some $x$ and $\De$ that are typical enough.
\begin{claim} \label{cl:A-6}
	Let $r_j\geq 0$, $d_j$ be real-numbers satisfying $\card{d_j}\leq r_j/2$ for all $j$.
    Denote $T=\sum r_j$, $T'=\sum (r_j+d_j)$, and let $p_j = r_j/T$ and $q_j = (r_i+d_i)/T'$ be two distributions.
    Then
	\begin{equation}\label{eq:cl:A-6:0}
	\norm{p-q}_1 \ll \sum_i \frac{|d_i|}{r_i} \cdot \frac{\min(r_i, T-r_i)}{T}.
	\end{equation}
\end{claim}
We defer the proof of Claim~\ref{cl:A-6} to Section~\ref{sec:def_proofs}. Morally speaking, it says that
\begin{equation}\label{eq:pretend1}
\Expect{{\bf x},\bm{\De}}{\norm{p({\bf x}) - p({\bf x}+\bm{\De})}_1}
\ll \sum\limits_{j=1}^{m}{\Expect{{\bf x}}{\Expect{\bm{\De}}{\frac{\card{r_j({\bf x}) - r_j({\bf x}+\bm{\De})}}{r_j({\bf x})}\cdot\frac{\min(r_j({\bf x}),T({\bf x})-r_j({\bf x}))}{T({\bf x})}}}},
\end{equation}
where $T(x) = \sum\limits_{j}{r_j(x)}$ (this is only morally because we are assuming that the supports of $p_j(x)$ and $p_j(x+\De)$
are the same, but formally speaking they may be different). In particular, to be able to handle with that we first must understand
the expectation of $\card{r_j(x) - r_j(x+\bm{\De})}$ over $\bm{\De}$.
\begin{claim}\label{cl:A-4}
    Let $j\in [m]$,  $x_1,\ldots,x_k \in [z_j-\frac{\log n}{25 n}-\eps^{0.95},z_j+\frac{\log n}{25 n}+\eps^{0.95}] \setminus [z_j-\frac{\log n}{50 n},z_j+\frac{\log n}{50 n}]$,
    and let $r(x)=\prod_{i=1}^c g_j(x_i)$.
	Denote $\al_i =\text{dist}\left(x_i, [z_j-\frac{\log n}{50n},z_j+\frac{\log n}{50n}]\right)$.
	and let $\bm{\De}_i \sim N(0,\sigma^2)|_{|\De_i|<\ve^{0.95}}$.
	Then
	\begin{equation}
	\E_{\bm{\De}}[|r(x+\bm{\De})-r(x)|] \ll \max\left(\eps^{2.65},r(x)\cdot \sigma\cdot\sqrt{\sum_{i=1}^c \frac{1}{\al_i^2}}\right).
	\end{equation}
\end{claim}

\begin{proof}
	We consider two cases.
	
	\paragraph{Case 1: $\al_i\le \ve^{0.9}$ for some $i$.}
    In this case, we have
    \[
    g_j(x_i)\ll
    \left(\frac{50n}{\log n}\alpha_i\right)^3
    \ll n^3\eps^{3\cdot 0.9}
    \ll \eps^{2.66}.
    \]
    Similarly, we have ${\sf dist}(x_i+\bm{\De}_i,[z_j-\frac{\log n}{50n},z_j+\frac{\log n}{50}])
    \leq \alpha_i + \card{\bm{\De}_i}\leq 2\alpha_i$, so $g_j(x_i+\bm{\De}_i)\ll \eps^{2.66}$.
    We conclude that $r(x_i),r(x_i+\bm{\De}_i)\ll \eps^{2.66}$, hence the contribution from these cases is at most $\eps^{2.65}$.
	
	 \paragraph{Case 2: $\al_i>\ve^{0.9}$ for all $i$.} In this case, we get that $x_i+\bm{\De}_i$ is also not in
     the interval $[z_j-\frac{\log n}{50n},z_j+\frac{\log n}{50n}]$, hence $g_j(x_i+\bm{\De}_i)\neq 0$,
     so $r(x+\bm{\De})>0$. Since $r(x)$ are defined using products, it would be more convenient for us to analyze
     $\log(r(x+\bm{\De})/r(x))$ as opposed to $r(x+\bm{\De})/r(x)-1$, and to justify we can do that we first argue that
     $r(x+\bm{\De})/r(x) = 1+o(1)$.

     To see that, note that as $\card{\bm{\De}_i}\leq \eps^{0.95}\leq \alpha_i/2$, we may use Fact ~\ref{fact:first_order_g} to conclude that
     \[
     \card{g(x_i+\bm{\De}_i) - g(x_i)}
     \ll \frac{\card{\bm{\De}_i}}{\alpha_i} \card{g(x_i)}
     \ll \eps^{0.05} \card{g(x_i)}.
     \]
     In particular, we get that
     $\frac{g_j(x_i+\bm{\De}_i)}{g_j(x_i)} = 1\pm O(\eps^{0.05})$, and hence
     $\frac{r(x+\bm{\De})}{r(x)} = 1\pm O(k\eps^{0.05})$. Writing $\frac{r(x+\bm{\De})}{r(x)} = 1+\bm{\eta}$,
     we get $\bm{\eta}$ is small in absolute value, and hence
     $\card{\log(r(x+\bm{\De})/r(x))}\gg \card{\bm{\eta}}\gg \card{\frac{r(x+\bm{\De})}{r(x)} - 1} = \card{\frac{r(x+\bm{\De})-r(x)}{r(x)}}$.
     I.e.,
	 \begin{equation}
	 \Expect{\bm{\De}}{\frac{|r(x+\bm{\De})-r(x)|}{r(x)}}
     \ll \Expect{\bm{\De}}{\card{\log \left(\frac{r(x+\bm{\De})}{r(x)}\right)}}
     = \Expect{\bm{\De}}{\left|\sum_{i=1}^k \log \frac{ g_j(x_i+\bm{\De}_i)}{g_j(x_i)} \right|}
     = \Expect{\bm{\De}}{\left|\sum_{i=1}^k {\bf Y}_i \right|},
	 	  \label{eq:cl:A-4:1}
\end{equation}
	 where we define the random variables ${\bf Y}_i = \log \frac{ g_j(x_i+\bm{\De}_i)}{g_j(x_i)}$.

     Observe that ${\bf Y}_i$'s are mutually independent, since each ${\bf Y}_i$ only depends on
	 the corresponding $\bm{\De}_i$. We wish to upper bound the average and variance of ${\bf Y}_i$,
     and to do that it would be more convenient to analyze ${\bf Z}_i = \frac{ g_j(x_i+\bm{\De}_i)-g_j(x_i)}{g_j(x_i)}$
     and then relate the two.

     Using second order Taylor's approximation, we have that there is ${\bf y}_i\in[x_i,x_i+{\bf \De}_i]$ such that
     \[
       g_j(x_i+{\bf \De}_i) = g_j(x_i) + g_j'(x_i){\bf \De}_i + \half g_j''({\bf y}_i){\bf \De}_i^2,
     \]
     hence
     \begin{equation}\label{eq:avg_Z}
       \card{\Expect{{\bf \De}}{{\bf Z}_i}}
       =\frac{1}{g_j(x_i)}\card{\Expect{\De}{g_j'(x_i){\bf \De}_i + \half g_j''({\bf y}_i){\bf \De}_i^2}}
       =\frac{1}{2g_j(x_i)}\card{\Expect{\De}{g_j''({\bf y}_i){\bf \De}_i^2}}.
     \end{equation}
     Using properties of $f$, we have
     \[
     \card{g_j''({\bf y}_i)} =
     \left(\frac{50n}{\log n}\right)^2 f''\left(\frac{50n}{\log n}\card{{\bf y}_i-z_j}\right)
     \ll \left(\frac{50n}{\log n}\right)^2 \card{\frac{50n}{\log n}\card{{\bf y}_i-z_j}-1}.
     \]
     Since ${\bf y}_i\in [x_i,x_i+{\bf \De}_i]$, we get that
     $
     \left(\frac{50n}{\log n}\card{{\bf y}_i-z_j}-1\right)
     \geq \frac{50n}{\log n}\alpha_i - \eps^{0.95}
     \geq \frac{25n}{\log n}\alpha_i
     $,
     and so we may continue the previous inequality as
     \[
     \card{g''({\bf y}_i)}
     \ll \left(\frac{50n}{\log n}\right)^2\frac{\card{\frac{50n}{\log n}\card{{\bf y}_i-z_j}-1}^3}{(\frac{25n}{\log n}\alpha_i)^2}
     \ll \frac{1}{\alpha_i^2} \card{g_j({\bf y}_i)}
     \ll \frac{1}{\alpha_i^2} \card{g_j(x_i)},
     \]
     where the last inequality is by Fact~\ref{fact:first_order_g}.
     Plugging this into~\eqref{eq:avg_Z} we get that
     \[
     \card{\Expect{\bm{\De}}{{\bf Z}_i}}
     \ll \frac{1}{\alpha_i^2} \Expect{\bm{\De}}{\bm{\De}_i^2}
     =\frac{1}{\alpha_i^2} \sigma^2.
     \]

	 In a similar fashion, we upper bound the second moment of ${\bf Z}_i$.
     Using Fact~\ref{fact:first_order_g}, we get that $\card{{\bf Z}_i}\leq \frac{\bm{\De}_i}{\alpha_i}$, and
     so $\Expect{\bm{\De}}{{\bf Z}_i^2}\ll \frac{1}{\alpha_i^2}\Expect{\bm{\De}}{\bm{\De}_i^2} =\frac{1}{\alpha_i^2} \sigma^2$.

     We can now upper bound the average of ${\bf Y}_i$ as follows. Recall that, $\card{{\bf Z}_i} = o(1)$ so by Taylor's approximation
     ${\bf Y}_i = \log(1+{\bf Z}_i) = {\bf Z}_i - \frac{1}{2(1+\bm{\xi}_i)^2} {\bf Z}_i^2$ for some $\bm{\xi}_i\in [1,1+{\bf Z}_i]$ and hence
	  \begin{equation}
	 \card{\Expect{}{{\bf Y}_i}} \ll \card{\E[{\bf Z}_i]} + \card{\E[{\bf Z}_i^2]|} \ll \frac{1}{\alpha_i^2}\sigma^2.
	 \label{eq:cl:A-4:4}
	 \end{equation}
	 This approximation (along with the fact that $\card{{\bf Z}_i} = o(1)$) also implies $\card{{\bf Y}_i}\ll \card{{\bf Z}_i}$, hence
	 \begin{equation}
	 \E[{\bf Y}_i^2] \ll \E[{\bf Z}_i^2]  \ll \frac{1}{\alpha_i^2}\sigma^2.
	 \label{eq:cl:A-4:5}
	 \end{equation}
	
	 We can now continue equation~\eqref{eq:cl:A-4:1} to upper bound the LHS there.
     Denoting $\mu_i:=\E[{\bf Y}_i]$, we have
	 \begin{align*}
	 \Expect{\De}{\card{\sum_{i=1}^k {\bf Y}_i}}
     \leq \sum_{i=1}^k |\mu_i|  +  \Expect{\De}{\card{\sum_{i=1}^k {\bf Y}_i-\mu_i}}
     \le \sum_{i=1}^k |\mu_i|
	 + \sqrt{\Expect{\De}{\sum_{i=1}^k ({\bf Y}_i-\mu_i)^2}},
	 \end{align*}
     where in the last inequality we used Cauchy-Schwarz and the fact that $Y_i$'s are independent.
     Using ~\eqref{eq:cl:A-4:4} we have that $\sum_{i=1}^k |\mu_i|\ll \sigma^2\sum_{i=1}^k \frac{1}{\alpha_i^2}$,
     and to upper bound the second term we use ~\eqref{eq:cl:A-4:5}:
     \[
     \Expect{\bm{\De}}{({\bf Y}_i-\mu_i)^2}
     \leq \Expect{\bm{\De}}{{\bf Y}_i^2}
     \ll \frac{1}{\alpha_i^2}\sigma^2.
     \]
     Together, we get that
     \[
     \Expect{\bm{\De}}{\card{\sum_{i=1}^k {\bf Y}_i}}
     \ll \sigma^2\sum_{i=1}^k \frac{1}{\alpha_i^2} + \sqrt{\sigma^2\sum_{i=1}^k \frac{1}{\alpha_i^2}}
     \ll \sigma  \sqrt{\sum_{i=1}^k \frac{1}{\alpha_i^2}},
     \]
     where the last inequality holds since
     $\sigma^2\sum_{i=1}^k \frac{1}{\alpha_i^2}\ll 1$
     (as $\sigma^2\ll \eps^2$ and $\alpha_i\geq \eps^{0.9}$).
\end{proof}

Next, using the previous claim we upper bound the expectation of each summand on the RHS of~\eqref{eq:pretend1}.
The following statement addresses a single term, and should be thought of as being applied after conditioning on
$x,\De$ being not-too untypical, and focusing only on $x_i$'s for which there is a chance that
$g_j(x_i + \De_i)\neq g_j(x_i)$.

\begin{claim}\label{cl:A-5}
Let $j\in [m]$, $k\leq n$, $S\geq 0$ and let $x_1,\ldots,x_k$
be chosen uniformly at random from
$[z_j-\frac{\log n}{25 n}-\eps^{0.95},z_j+\frac{\log n}{25 n}+\eps^{0.95}] \setminus [z_j-\frac{\log n}{50 n},z_j+\frac{\log n}{50 n}]$.
Let $\De_i \sim N(0,\sigma^2)|_{|\De_i|<\ve^{0.95}}$.
Then
\begin{align*}
\E_{{\bf x},\bm{\De}} \left[ \frac{|r_j({\bf x}+\bm{\De})-r_j({\bf x})|}{r_j({\bf x})}\cdot \frac{\min(r_j({\bf x}),S)}
 {r_j({\bf x})+S+\ve^{1.6}}\right]
\ll  \ve^{1.05}+ k\frac{\sigma n}{\log n}  \cdot \Pr_{{\bf x}}[r_j({\bf x})\ge S]\\
+\sigma \frac{n}{\log n}\sqrt{k} \Expect{{\bf x}}{\frac{r_j({\bf x})}{r_j({\bf x})+S}}.
\end{align*}
\end{claim}
\begin{proof}
	Upper bounding $\max(a,b)\leq a+b$ for $a,b\geq 0$, by Claim~\ref{cl:A-4}, we have
	\begin{align*}
	\E_{{\bf x},\bm{\De}} \left[ \frac{|r_j({\bf x}+\bm{\De})-r_j({\bf x})|}{r_j({\bf x})}\cdot \frac{\min(r_j({\bf x}),S)}
	{r_j({\bf x})+S+\ve^{1.6}}\right]
    &\ll \E_{{\bf x}}
	\left[ \frac{\ve^{2.65}+r_j({\bf x})\cdot \sigma\cdot\sqrt{\sum_{i=1}^k \frac{1}{\al_i^2}}}{r_j({\bf x})}\cdot \frac{\min(r_j({\bf x}),S)}
	{r_j({\bf x})+S+\ve^{1.6}}\right]\\
    &\ll  \eps^{1.05}+\sigma\Expect{{\bf x}}{\sqrt{\sum_{i=1}^k \frac{1}{\bm{\al}_i^2}}\cdot \frac{\min(r_j({\bf x}),S)}{r_j({\bf x})+S+\ve^{1.6}}},
    \end{align*}
    and it is enough to bound the second term. Note that while we expect that each $\bm{\alpha}_i$ to be of the order $\log n/n$, convexity works against
    us and it could still be the case that $\sum\limits_{i=1}^{k}\frac{1}{\bm{\alpha}_i^2}$ could be large.
    The point is that in this case, some $\bm{\alpha}_i$ must be close to $0$, in which case $g_j({\bf x}_i)$ is very small
    -- cubically with $\bm{\alpha}_i$ -- thereby balancing the $1/\bm{\alpha}_i^2$ term. The following proposition formalizes this intuition,
    and the proof is deferred to Section~\ref{sec:def_proofs}
    \begin{proposition}\label{prop:convex_error}
      There is an absolute constant $A>0$ such that for any $z>0$ and $r\leq 1$ such that $r_j(x) = r\cdot g_j(x_i)$,
      it holds that
      \begin{equation*}
	\Expect{{\bf x}_i}{\sqrt{z+\frac{1}{\bm{\al}_i^2}} \cdot\frac{\min(r\cdot g_j({\bf x}_i),S)}
	{r\cdot g_j({\bf x}_i)+S+\ve^{1.6}}}
    \leq
	\Expect{{\bf x}_i}{\sqrt{z+A\frac{n^2}{\log^2 n}} \cdot \frac{\min(r\cdot g_j({\bf x}_i),S)}{r\cdot g_j({\bf x}_i)+S+\ve^{1.6}} +A\frac{n}{\log n} \cdot
\mathbbm{1}_{r\cdot g_j({\bf x}_i)\ge S}}.
	\end{equation*}
    \end{proposition}

    Applying Proposition~\ref{prop:convex_error} iteratively $k$ times (once for each $i$, taking $r = \prod\limits_{i'\neq i}{g_{j}(x_{i'})}$
    and the appropriate $z$), we get that
    \[
    \Expect{{\bf x}}{\sqrt{\sum_{i=1}^k \frac{1}{\bm{\al}_i^2}}\cdot \frac{\min(r_j({\bf x}),S)}{r_j({\bf x})+S+\ve^{1.6}}}
    \leq \Expect{{\bf x}}{\sqrt{k\cdot A\frac{n^2}{\log^2 n}}\cdot \frac{\min(r_j({\bf x}),S)}{r_j({\bf x})+S+\ve^{1.6}}+k\cdot A\frac{n}{\log n} \cdot
    \mathbbm{1}_{r_j({\bf x})\geq S}}.
    \]
    The proof is concluded by noting that $\frac{\min(r_j(x),S)}{r_j(x)+S+\ve^{1.6}}\leq \frac{r_j(x)}{r_j(x)+S}$.
    \end{proof}

We are now ready to finish the proof of inequality~\eqref{eq:mainA1}.
\subsubsection*{Proof of inequality~\eqref{eq:mainA1}}
Let $E$ be the event that:
(1) the support of $p({\bf x})$ has size at least $2$,
(2) $\sum\limits_{j} r_j({\bf x})\geq \eps^{1.6}$ and also for ${\bf x}+\bm{\De}$,
and (3) $\card{\bm{\De}_i}\leq \eps^{0.95}$ for all $i\in[n]$.
As we argued in Claims~\ref{cl:A-1},~\ref{cl:A-3},~\ref{cl:A-2} the contribution $({\bf x},\bm{\De})\not\in E$ to
the LHS of inequality~\eqref{eq:mainA1} is $\ll \eps$, hence it is enough to analyze the contribution of $({\bf x},\bm{\De})\in E$.

Denote $T(x) = \sum\limits_{j\in [m]} r_j(x)$.
\begin{align*}
&\Expect{{\bf x},\bm{\De}}{\norm{p({\bf x}) - p({\bf x}+\bm{\De})}_1 1_E}
= \Expect{{\bf x},\bm{\De}}{\sum\limits_{j\in [m]}\card{\frac{r_j({\bf x})}{T({\bf x})}-\frac{r_j({\bf x}+\bm{\De})}{T({\bf x}+\bm{\De})}} 1_E}\\
&= \underbrace{\Expect{{\bf x},\bm{\De}}{\sum\limits_{j\in [m]}\card{\frac{r_j({\bf x})}{T({\bf x})}-\frac{r_j({\bf x}+\bm{\De})}{T({\bf x}+\bm{\De})}}
1_E 1_{r_j({\bf x}) \leq \eps^{2.7}}}}_{(\rom{1})}
+\underbrace{\Expect{{\bf x},\bm{\De}}{\sum\limits_{j\in [m]}\card{\frac{r_j({\bf x})}{T({\bf x})}-\frac{r_j({\bf x}+\bm{\De})}{T({\bf x}+\bm{\De})}}
1_E 1_{r_j({\bf x}) >\eps^{2.7}}}}_{(\rom{2})}
\end{align*}
First, we show that $(\rom{1})\ll \eps$. As $T({\bf x})\geq \eps^{1.6}$ (since $E$ holds) and $r_j({\bf x})\leq \eps^{2.7}$, we get that
$r_j({\bf x})/T({\bf x})\leq \eps^{1.1}$, and next we argue that $r_j({\bf x}+\bm{\De})/T({\bf x}+\bm{\De})\ll \eps^{1.05}$.
Fix $j$ and suppose ${\bf x}_1,\ldots,{\bf x}_{k_j}$ are the ${\bf x}_i$'s that fall inside $I_j$. The following easy fact will be helpful.

\begin{fact}\label{fact:general}
  For all $x,\De$ we have
  $r_j(x+\De) = \sum\limits_{\substack{S\subseteq[k_j]}}{
  \prod\limits_{r\in S}{g_j(x_r)}
  \prod\limits_{r\not\in S}{(g_j(x_r)-g_j(x_r+\De_r))}}$.
\end{fact}
\begin{proof}
  Write $r_j(x+\De) =
  \prod\limits_{r=1}^{k_j} g_j(x_r+\De_r)
  =\prod\limits_{r=1}^{k_j}\left(g_j(x_r) + (g_j(x_r+\De_r)-g_j(x_r))\right)$
  and expand out.
\end{proof}

Combining Fact~\ref{fact:general} and Fact~\ref{fact:first_order_g}, we get that
\begin{align*}
r_j(x+\De)
\leq \sum\limits_{S\subseteq[k_j]}{
  \prod\limits_{r\in S}{g_j(x_r)}
  \prod\limits_{r\not\in S}{\card{g_j(x_r)-g_j(x_r+\De_r)}}}
&\leq \sum\limits_{S\subseteq[k_j]}{
\prod\limits_{r\in S}{g_j(x_r)}
B^{\card{S}} n^{3\card{S}}\prod\limits_{r\not\in S}(\alpha_r^3 + \card{\De_r}^3)}\\
&\leq \sum\limits_{S\subseteq[k_j]}{
\prod\limits_{r\in S}{g_j(x_r)}
B^{\card{S}} n^{3\card{S}}\prod\limits_{r\not\in S}\alpha_r^3}\\
&\phantom{\leq_S}+4^{n}B^{\card{n}}n^{3n}\max_{r}\card{\De_r}^3.
\end{align*}
Consider the right hand side above. For the first term we use $\alpha_r^3\ll g_j(x_r)$ to get it is at most
\[
\sum\limits_{S\subseteq[k_j]}{B'^{\card{S}} n^{3\card{S}} r_j(x_r)}
\leq (B'')^n n^{3n} \eps^{2.7} \leq \eps^{2.65}/2.
\]
For the second term we use $\card{\De_r}\leq \eps^{0.95}$ to bound it by $\eps^{2.65}/2$ as well.
We thus get $r_j(x+\De)\leq \eps^{2.65}$, and so $r_j(x+\De)/T(x+\De)\leq \eps^{1.05}$. Combined, we get that
\[
{(\rom{1})}\leq m(\eps^{1.1}+ \eps^{1.05})\ll \eps.
\]

Next, we handle $(\rom{2})$. Denote $T'(x) = \sum\limits_{j}{r_j(x) 1_{r_j(x)\geq \eps^{2.7}}}$, and note that
$T'(x)\geq T(x) - m\eps^{2.7}\geq (1-m\eps^{1.1})T(x)$ and similarly for $T'(x+\De)$. Thus, we may replace
$T(x), T(x+\De)$ with $T'(x), T'(x+\De)$ and incur (by the triangle inequality) a loss of at most $m\eps^{1.1}\ll \eps$.
Thus, we want to upper bound
\[
\underbrace{\Expect{{\bf x},\bm{\De}}{\sum\limits_{j\in [m]}\card{\frac{r_j({\bf x})}{T'({\bf x})}-\frac{r_j({\bf x}+\bm{\De})}{T'({\bf x}+\bm{\De})}} 1_E 1_{r_j({\bf x}) >\eps^{2.7}}}}_{(\rom{3})}\ll \eps.
\]
We intend to apply Claim~\ref{cl:A-6} with $r_j = r_j(x)$ and $d_j = r_j(x+\De) - r_j(x)$ for each $x$ separately, but for that
we first have to argue that $\card{d_j}\leq r_j/2$.
For each $i\in [n]$ there is $j$ such that $x_i\in I_j$, and we denote
$\al_i =\text{dist}\left(x_i, [z_j-\frac{\log n}{50n},z_j+\frac{\log n}{50n}]\right)$. Note that
\[
\eps^{2.7}\leq r_j(x) \leq g_j(x_i)\ll \left(\frac{n}{\log n}\alpha_i\right)^3,
\]
hence $\alpha_i\gg \frac{\log n}{n}\eps^{0.9}$, and for small enough $\eps$ we get that $\alpha_i\geq \eps^{0.91}\geq 2\card{\De_i}$. Therefore,
Combining Fact~\ref{fact:general} and Fact~\ref{fact:first_order_g} we get
\[
\card{d_j(x)}
=\card{r_j(x) - r_j(x+\De)}
\leq \sum\limits_{\substack{S\subseteq[k_j]\\S\neq[k_j]}}{
  \prod\limits_{r\in S}{g_j(x_r)}
  \prod\limits_{r\not\in S}{\card{g_j(x_r)-g_j(x_r+\De_r)}}}
\leq
\sum\limits_{\substack{S\subseteq[k_j]\\S\neq[k_j]}}{ B^{\card{S}} r_j(x)
\prod\limits_{r\not\in S}{\frac{\card{\De_i}}{\alpha_i}}}.
\]
Bounding $\frac{\card{\De_i}}{\alpha_i}\leq \eps^{0.95}/\eps^{0.91} = \eps^{0.04}$ we get that
\[
\card{r_j(x) - r_j(x+\De)}\leq r_j(x) \eps^{0.04}
\sum\limits_{\substack{S\subseteq[k_j]\\S\neq[k_j]}}{B^{\card{S}}}
\leq B'^n \eps^{0.04} r_j(x)
\leq r_j(x)/2 = r_j/2.
\]

Therefore, we may apply Claim~\ref{cl:A-6} and get that
\begin{align}
(\rom{3})&\ll \Expect{{\bf x},\bm{\De}}{\sum\limits_{j=1}^{m} \frac{|r_j({\bf x}) - r_j({\bf x}+\bm{\De})|}{r_j({\bf x})} \cdot \frac{\min(r_j({\bf x}), T'({\bf x}))}{T'({\bf x})}
1_E 1_{r_j({\bf x}) >\eps^{2.7}}} \notag\\\label{eq:2}
&\ll \Expect{{\bf x},\bm{\De}}{\sum\limits_{j=1}^{m} \frac{|r_j({\bf x}) - r_j({\bf x}+\bm{\De})|}{r_j({\bf x})} \cdot \frac{\min(r_j({\bf x}), T'({\bf x}))}{T'({\bf x})+\eps^{1.6}}1_E 1_{r_j({\bf x}) >\eps^{2.7}}},
\end{align}
where the last inequality holds since $T'({\bf x})\gg \eps^{1.6}$. Next, we wish to discard ${\bf x}_i$ that are very far from their closest center $z_j$.
For each $j$, note that $\left[z_j-\frac{\log n}{50 n},z_j+\frac{\log n}{50 n}\right]$ is exactly the set of $y$'s on which
$g_j(y) = 0$, and let $R_j\subseteq I_j$ be $R_j = \left[z_j-\frac{\log n}{25 n}-\eps^{0.95},z_j+\frac{\log n}{25n}+\eps^{0.95}\right] \setminus
\left[z_j-\frac{\log n}{50 n},z_j+\frac{\log n}{50 n}\right]$.
Note that for each $y\in I_j\setminus R_j$, we have that either $g_j(y) = 0$ if $y\in \left[z_j-\frac{\log n}{50 n},z_j+\frac{\log n}{50 n}\right]$, and otherwise $g_j(y) = 1$.
Furthermore, in the latter case we also have that $g_j(y+\bm{\De}_i) = 1$ since $\card{\bm{\De}_i}\leq \eps^{0.95}$.

We sample ${\bf x}$ in the following way. First, sample ${\bf t}_1,\ldots,{\bf t}_m$ the number of ${\bf x}_i$'s in each interval $I_1,\ldots,I_m$,
then for each $j$ sample ${\bf k}_j$ to be the number of $x_i$'s inside the interval $I_j$ that fall inside $R_j$.
Finally, for each $j\in[m]$ sample ${\bf k}_j$ points uniformly from $R_j$, ${\bf t}_j - {\bf k}_j$ uniformly from $I_j\setminus R_j$, and let ${\bf x}$
be the (multi-)set of all the sampled points. We condition on the ${\bf t}_j$'s and ${\bf k}_j$'s henceforth in~\eqref{eq:2}. Furthermore, we condition
on the identity of the $i$'s for which ${\bf x}_i\in I_j$ for each $j$.

Since $i$'s for which ${\bf x}_i\in I_j\in [z_j-\frac{\log n}{25 n}-\eps^{0.95},z_j+\frac{\log n}{25n}+\eps^{0.95}]$
do not affect both $r_j({\bf x})$ and $r_j({\bf x}+\bm{\De})$, we may ignore them and hence take expectation only over $i$'s such
${\bf x}_i\in R_j$. Call these $y$'s. Then from~\eqref{eq:2} we get
\begin{align*}
(\rom{3})
&\ll
\Expect{\vec{{\bf t}},\vec{{\bf k}}}
{
\Expect{{\bf y},\bm{\De}}
{
\sum\limits_{j=1}^{m} \frac{|r_j({\bf y}) - r_j({\bf y}+\bm{\De})|}{r_j({\bf y})} \cdot \frac{\min(r_j({\bf y}), T'({\bf y}))}{T'({\bf y})+\eps^{1.6}}
}}\\
&\leq
\Expect{\vec{{\bf t}},\vec{{\bf k}}}
{
\sum\limits_{j=1}^{m}
\Expect{{\bf y},\bm{\De}}
{\frac{|r_j({\bf y}) - r_j({\bf y}+\bm{\De})|}{r_j({\bf y})} \cdot \frac{\min(r_j({\bf y}), T'_{-j}({\bf y}))}{r_j({\bf y})+T'_{-j}({\bf y})+\eps^{1.6}}
}},
\end{align*}
where $T'_{-j}(x) = \sum\limits_{j'\neq j}r_{j'}(x) 1_{r_{j'}(x)\geq \eps^{2.7}}$.
Note that conditioned on $\vec{{\bf t}}= \vec{t},\vec{{\bf k}} = \vec{k}$, the values of ${\bf y}_i$'s such that ${\bf y}_i\in I_j$ are independent of $T'_{-j}({\bf y})$,
and they are distributed uniformly over $R_j$. Therefore, using Claim~\ref{cl:A-5} we have
\begin{align*}
(\rom{3})
&\ll
\Expect{\vec{{\bf t}},\vec{{\bf k}}}{
\sum\limits_{j=1}^{m}
\eps^{1.05}+ {\bf k}_j\frac{\sigma n}{\log n}  \cdot \cProb{{\bf y}}{\vec{{\bf t}},\vec{{\bf k}}}{r_j({\bf y})\ge T'_{-j}({\bf y})}
+\Expect{{\bf y}}{\sigma \frac{n}{\log n}\sqrt{{\bf k}_j}\frac{r_j({\bf y}) }{T'({\bf y})}}
}.\\
&\leq
m\eps^{1.05}
+ n^2\sigma
\sum\limits_{j=1}^{m}\Prob{y}{r_j(y)\ge T'_{-j}(y)}
+ \sigma\frac{n}{\log n}
\Expect{\vec{t},\vec{k}}{\sqrt{\max_{j} k_j}}.
\end{align*}
Note that if $T'_{-j}(x)\leq r_j(x)$, then
\[
T(x) \leq
T'_{-j}(x) + r_j(x) + \sum\limits_{j'}{r_{j'}(x) 1_{r_{j'}(x)\leq \eps^{2.7}}}
\leq 2r_j(x) + m\cdot \eps^{2.7}
\leq 3,
\]
so we bound the sum on the right hand side by $m\Prob{x}{T(x)\leq 3}$. For the expectation,
we use Cauchy-Schwarz and overall we get
\[
(\rom{3})
\leq m\eps^{1.05} + n^3\sigma \Prob{{\bf x}}{T({\bf x})\leq 3} +
\sigma\frac{n}{\log n}\sqrt{\Expect{\vec{{\bf t}},\vec{{\bf k}}}{\max_{j} {{\bf k}_j}}}.
\]
The first term is clearly $\ll \eps$. For the second term, we use Claim~\ref{cl:A-7} below, that asserts that
$\Prob{{\bf x}}{T({\bf x})\leq 3}\leq n^{-\omega(1)}$, hence by the definition of $\sigma$ the second term is
also $\ll \eps$. For the third term, note that each ${\bf k}_j$ is a sum of $n$ independent Berounlli random variables
with parameter $p\leq\log n/n$, therefore by Chernoff bound
\[
\Prob{}{{\bf k}_j\geq 10\log n}
\leq e^{-\frac{1}{3}9^2 \log n}
\leq n^{-9}.
\]
The union bound now implies that $\Prob{}{\max_j {\bf k}_j\geq 10\log n}\leq n^{-8}$, and hence
\[
\Expect{\vec{{\bf t}},\vec{{\bf k}}}{\max_{j} {\bf k}_j}
\leq n^{-8}\cdot n + 10\log n\ll \log n.
\]
Using the definition of $\sigma$, we get that the third term is also $\ll \eps$. Combining all,
we get that $(\rom{3})\ll \eps$, and we are done.
\qed

\begin{claim}\label{cl:A-7}
\begin{equation}
\Pr_{{\bf x}} \left[\sum_j r_j({\bf x}) \le 3\right] < n^{-\omega(1)}.
\end{equation}
\end{claim}

\begin{proof}
	The proof is very similar to the analysis of Case~(B) above. In particular, similarly to inequality~\eqref{eq:notEi},
	$$
	\Pr[r_j({\bf x})<1] = \left( 1-\frac{2m \log n}{25 n} \right)^{t_j}
    \geq
	\left( 1-\frac{2m \log n}{25 n} \right)^{2 \cdot n/m} > e^{-2 \log n/25} = n^{-2/25},
	$$
	as long as $t_j < 2\cdot n/m$ (which is the case except with probability $n^{-\omega(1)}$. Since $m > n^{2/25} \cdot n^{\Omega(1)}$, the probability of not having at least three $r_j(x)$'s equal to $1$ is $n^{-\omega(1)}$.
\end{proof}

\section{The value of the $t$-fold symmetric odd cycle game}\label{sec:val_sym_cycle}
\subsection{The upper bound: Theorem~\ref{thm:val_sym_cycle_ub}}
Suppose that $n = 2m - 1$ and $A$ is a strategy for $C_n^{\otimes_{\sf sym} t}$. We will view $A$
as a symmetric function over ordered $t$ tuples, i.e. as  $A\colon C_{n}^{t} \to \power{t}$ satisfying
$A(\pi(x)) = \pi(A(x))$ for all permutations $\pi$ over $[t]$.

We identify $C_{n} = \sett{\frac{i}{n}}{i = 0,1,\ldots,n-1}$, consider the
lattice $L = (C_{n} + \mathbb{Z})^{t}$ and define a rounding map $R\colon L\to\mathbb{Z}^t$ on it as follows. For
$x\in C_{n}^t$, we define $R(x) = A(x) + n x \pmod{2}$, and then we extend $R$ to $L$ by
$R(x+z) = R(x) + z$ for $x\in C_{n}^t$ and $z\in\mathbb{Z}^t$.

Let $D = R^{-1}(0^t)$. The symmetry of $A$ implies that $D$ is symmetric, and we also note that $D$ is a tiling
of the lattice $L$.
\begin{definition}
  A random $\eps$-Bernouli direction, denoted by ${\bf u}\sim {\sf B}(\eps)$,
  is a random variable distributed on $\set{\pm \frac{1}{n}, 0}$, such that for each $i\in[t]$ independently,
  $\Prob{}{{\bf u}_i = 0} = 1-2\eps$ and $\Prob{}{{\bf u}_i = 1/n} = \Prob{}{{\bf u}_i = -1/n} = \eps$.
\end{definition}
We will mostly be concerned with $\eps = 1/4$, in which case the distribution of
${\bf x}, {\bf x}+{\bf u}\hspace{-1ex}\pmod{1}$ where ${\bf x}\in_R C_{n}^t$ and
${\bf u}$ is an independent $\frac{1}{4}$-Bernouli step, is exactly the distribution of challenges
to the players. Inspecting, we see that players succeed on these challenges if and only if
$R({\bf x}) = R({\bf x}+{\bf u})$, as the following claim shows.
\begin{claim}\label{claim:success_then_same_cell}
  Let $x\in C_{n}^t$ and $u\in\set{\pm \frac{1}{n},0}^t$. Then the players succeed on challenges
  $(x, x+u\hspace{-1ex}\pmod{1})$ if and only if $R(x) = R(x+u)$.
\end{claim}
\begin{proof}
  Note that $x$ and $x+u$ are either in the same cell of $D$ or in adjacent cells,
  so to prove the statement it is enough to show that the players succeed on the challenge
  if and only if $R(x) = R(x+u)\pmod{2}$.

  Write $x+u = d + z$ where $d \in C_{n}^t$ is $x+u\hspace{-0.6ex}\pmod{1}$, and $z\in\mathbb{Z}^t$.
  Note that
  \[
    R(x+u) = R(d) + z = A(d) + dn + z \pmod{2},
    \qquad\qquad
    R(x) = A(x) + nx\pmod{2}
  \]
  and subtracting the equations we get that
  \[
    R(x+u) - R(x) = A(d) - A(x) + dn + z - nx\pmod{2}.
  \]
  Multiplying the equality $x+u = d+z$ by $n$ and taking modulo $2$ we get that
  $nu + nx = nd + nz = nd + z\pmod{2}$ where the last transition used the fact that
  $n$ is odd. Thus, $R(x+u) - R(x) = A(d) - A(x) +nu\pmod{2}$. Note that the players succeed
  on the challenge if and only if $A(x) = A(d) + nu\pmod{2}$, and plugging that in we get that
  they succeed if and only if $R(x+u) - R(x) = 0\pmod{2}$, as desired.
\end{proof}
Claim~\ref{claim:success_then_same_cell} implies that the failure probability of the players is
\[
\Prob{{\bf x}\in C_n^t, {\bf u}\sim{\bf B}(1/4)}{{\bf x}, {\bf x}+{\bf u}\text{ are in different cells of $D$}}.
\]
Setting ${\bf y} = {\bf x}\pmod{D}$, it is easily seen that the distribution of ${\bf y}$ is uniform over $D$,
so the probability of the above event is equal to
\[
\eta \defeq \Prob{{\bf y}\in D, {\bf u}\sim{\sf B}(1/4)}{{\bf y}+{\bf u}\not\in D}.
\]
The rest of the proof is devoted to lower bounding $\eta$.
Setting $k = M \frac{n\sqrt{\log t}}{t}$ for large constant $M$ to be determined later, we show:
\begin{lemma}\label{lem:main_odd_up}
  $\eta\geq \Omega(1/k)$.
\end{lemma}
Below, we will assume $k$ is an integer,
otherwise we may multiply it by a constant factor close to $1$ and make it an integer. We then further assume $k$ is prime,
otherwise we may find a prime in $[k,2k]$ and replace $k$ by it.
Define
$\delta = \Prob{{\bf x}\in D, {\bf u}\sim {\sf B}(1/4)}{{\bf x}+k{\bf u}\not\in D}$ and observe the following easy relation
between $\delta$ and $\eta$.
\begin{claim}\label{claim:trivial_delta_eta}
  $\delta\leq k\eta$.
\end{claim}
\begin{proof}
  By the union bound
\[
\delta\leq
\sum\limits_{j=0}^{k-1}
\Prob{{\bf x}\in D, {\bf u}}{{\bf x}+j{\bf u}\in D, {\bf x}+(j+1){\bf u}\not\in D}.
\]
Note that for each $j$, the distribution of $y = {\bf x}+j{\bf u}\pmod{D}$ is uniform over $D$,
the $j$th term in the above sum is at most $\Prob{{\bf y}\in D, {\bf u}}{{\bf y} + {\bf u}\not\in D} = \eta$.
\end{proof}

\subsubsection{Disjoint Bernouli steps}
We will also consider the situation after making two Bernouli steps whose support is disjoint,
and for that we make the following definition.
\begin{definition}
The distribution of two disjoint $\eps$-Bernouli direction,
denoted by $({\bf u}^1$, ${\bf u}^2)\sim {\sf DB}(\eps)$,
is defined as follows. For each $i$ independently, set each one of the following options with probability $\frac{\eps}{2}$:
$({\bf u}^1_i, {\bf u}^2_i) = (1/n, 0)$, $({\bf u}^1_i, {\bf u}^2_i) = (-1/n, 0)$,
$({\bf u}^1_i, {\bf u}^2_i) = (0, 1/n)$, $({\bf u}^1_i, {\bf u}^2_i) = (0, -1/n)$;
otherwise, set $({\bf u}^1_i, {\bf u}^2_i) = (0, 0)$.
\end{definition}
We note that if $({\bf u}^1$, ${\bf u}^2)\sim{\sf DB}(\eps)$, then
${\bf u}^1 + {\bf u}^2$ is distributed as ${\sf B}(\eps)$. Therefore:
\begin{claim}\label{claim:two_small_steps}
  It holds that:
  \begin{itemize}
    \item $\Prob{{\bf x}\in D, {\bf u}\sim {\sf B}(1/4)}{{\bf x}+k{\bf u}\not\in D}\leq 2\delta$;
    \item $\Prob{{\bf x}\in D, {\bf u}\sim {\sf B}(1/4)}{{\bf x}+{\bf u}\not\in D}\leq 2\eta$.
  \end{itemize}
\end{claim}
\begin{proof}
  We prove the first item, and the second item is proved analogously.
  To sample ${\bf u}\sim {\sf B}(1/4)$, we sample $({\bf u}^1, {\bf u}^2)\sim{\sf DB}(1/4)$ and
  take ${\bf u} = {\bf u}^1 + {\bf u}^2$, so by the union bound the probability in the first item
  is at most
  \[
  \Prob{{\bf x}\in D, ({\bf u}^1, {\bf u}^2)\sim {\sf DB}(1/4)}{{\bf x}+k{\bf u}^1\not\in D}
  +\Prob{{\bf x}\in D, ({\bf u}^1, {\bf u}^2)\sim {\sf DB}(1/4)}{{\bf x}+k{\bf u}^1\in D, {\bf x}+k{\bf u}^1 + k{\bf u}^2\not\in D}.
  \]
  The first probability is $\delta$, and we argue that the second probability is at most the first. Indeed, setting ${\bf y} = {\bf x}+k{\bf u}^1$, this probability
  is at most the probability that ${\bf y}, {\bf y}+k{\bf u}^2$ are in different cells of $D$. Note that this occurs if and only
  if ${\bf y}\pmod{D}$ and ${\bf y}\pmod{D} + k{\bf u}^2$ are in different cells of $D$; note also that for every fixing of ${\bf u}^1$,
  the distribution of ${\bf y}\pmod{D}$ is uniform over $D$. Thus
  \[
  \Prob{{\bf x}\in D, ({\bf u}^1, {\bf u}^2)\sim {\sf DB}(1/4)}{{\bf x}+k{\bf u}^1\in D, {\bf x}+k{\bf u}^1 + k{\bf u}^2\not\in D}\leq
  \Prob{{\bf y}\in D, ({\bf u}^1, {\bf u}^2)\sim {\sf DB}(1/4)}{{\bf y} + k{\bf u}^2\not\in D} = \delta.\qedhere
  \]
\end{proof}

\begin{definition}
  Let $x\in D$ and $u$ be a direction. We say $(x,u)$ is decent if
  \[
  \cProb{({\bf u}^1, {\bf u}^2)\sim {\sf DB}(1/4)}
  {{\bf u}^1 + {\bf u}^2=u}
  {
  x+{\bf u}_1\not \in D\lor
  x+{\bf u}_2\not \in D\lor
  x+k{\bf u}_1\not\in D\lor
  x+k{\bf u}_2\not\in D} < \frac{1}{32}.
  \]
\end{definition}

\begin{claim}\label{claim:decent}
  $\Prob{{\bf x}\in_R D, {\bf u}\sim {\sf B}(1/4)}{({\bf x}, {\bf u})\text{ is decent}}\geq 1-64(\eta+\delta)$
\end{claim}
\begin{proof}
  Denote
  \[
  p(x,u) =
  \cProb{({\bf u}^1, {\bf u}^2)\sim {\sf DB}(1/4)}{{\bf u}^1 + {\bf u}^2=u}
  {x+{\bf u}_1\not \in D\lor
  x+{\bf u}_2\not \in D \lor
  x+k{\bf u}_1\not\in D\lor
  x+k{\bf u}_2\not\in D}.
  \]
  Note that
  \[
  \Expect{\substack{{\bf x}\in_R D\\ {\bf u}\sim {\sf B}(1/4)}}{p({\bf x}, {\bf u})}
  =\Prob{\substack{{\bf x}\in_R D\\ ({\bf u}^1, {\bf u}^2)\sim {\sf DB}(1/4)}}
  {
  {\bf x} + {\bf u}^1\not\in D\lor {\bf x} + {\bf u}^2\not\in D\lor
  {\bf x}+k{\bf u}^1\not\in D\lor {\bf x}+k{\bf u}^1\not\in D},
  \]
  which is at most $2(\delta+\eta)$ by the union bound. Thus, by Markov's inequality
  \[
  \Prob{{\bf x}\in_R D, {\bf u}\sim {\sf B}(1/4)}{({\bf x}, {\bf u})\text{ is not decent}}
  =
  \Prob{{\bf x}\in_R D, {\bf u}\sim {\sf B}(1/4)}{p({\bf x}, {\bf u})\geq \frac{1}{32}}\leq 64(\delta+\eta).
  \qedhere
  \]
\end{proof}

\subsubsection{Analyzing the potential function}
Our argument closely follows the argument in Section~\ref{sec:lb}, and below we focus on the necessary adjustments.
Set $Z = \frac{t}{10\log t}$. The definition of the potential function stays as is. We will have several constants
floating around in the proof which are not important for the most part, however we make the distinction between the
constants $c_1,\ldots,c_6$ that will be absolute (i.e. not depending on $M$), and the constants $t_0(M),t_1(M),t_2(M)$
that will depend on $M$.

The following is a variant of Claim~\ref{cl:1}, which is the main difference with the argument from Section~\ref{sec:lb}.
\begin{claim}\label{cl:1'}
	If $x, x+u, x-u, x+ ku, x-ku \in D$ and both $(x,u), (x,-u)$ are decent, then
	\[
	|\Psi(x+ku)-\Psi(x,ku)|\leq t^2\cdot
    e^{-Z/4}.
	\]
\end{claim}
\begin{proof}
   We consider the contribution of each pair $(i,j)$ to $\Psi(x+ku)$ and $\Psi(x,ku)$ separately.
   Without loss of generality we may only consider pairs $i,j$ that $\ga(x_i,x_j)=1$,
   and thus $d(x_i,x_j)= x_i - x_j + z$ for some $z\in \ZZ$, $z\neq 0$.
   Let $d= x_i - x_j + z + k(u_i-u_j)$.
   \begin{proposition}
     $d \geq 0$.
   \end{proposition}
   \begin{proof}
   Assume otherwise. Since $x_i-x_j + z \geq 0$ it follows by continuity that there is
   $\lambda\in[0,1)$ such that $x_i - x_j + z + \lambda k(u_i-u_j) = 0$. Note that $u_i - u_j$ can either be
   $0, \pm \frac{1}{n}, \pm \frac{2}{n}$. If $u_i - u_j = 0$, we get that $x_i - x_j + z = 0$, and as $x\in D$
   this contradicts Lemma~\ref{lem:diff}. Otherwise, multiplying by $n$, we get that $\lambda kn(u_i-u_j)$ is an
   integer. Note that $kn(u_i-u_j)$ is either $\pm k$ or $\pm 2k$, and as $k$ is prime we get that $\lambda = \frac{1}{2}$,
   $\lambda = \frac{1}{k}$ or $\lambda = \frac{1}{2k}$, and we analyze each case separately. If $\lambda = \frac{1}{k}$ then
   we get $x_i - x_j + u_i - u_j + z = 0$, so $x+u\in D$ has two coordinates differing by a non-zero integer,
   contradicting Lemma~\ref{lem:diff}. We next consider the other two cases separately, and assume that
   $u_i - u_j > 0$ --- otherwise we use $-u$ instead of $u$ in the argument below.

   If $\lambda = \frac{1}{2k}$, then necessarily $u_i - u_j = \frac{2}{n}$ and
   and we get that $x_i - x_j + z +\frac{1}{n} = 0$.
   Sample $({\bf u}^1, {\bf u}^2)\sim {\sf DB}(1/4)$ conditioned on ${\bf u}^1+{\bf u}^2 = u$. Note that the event
   that ${\bf u}^1_i = 1/n$ and ${\bf u}^1_j = 0$ occurs with probability $1/32$.
   Since $(x,u)$ is decent, we get that $x + {\bf u}^1\in D$ with probability strictly greater than $\frac{31}{32}$.
   Thus, the probability that $x + {\bf u}^1\in D$ and $({\bf u}^1_i, {\bf u}^1_j) = (1/n, 0)$ is positive, and in this
   case we get
   \[
     (x + {\bf u}^1)_i - (x + {\bf u}^1)_j
     =x_i - x_j + \frac{1}{n} = -z\neq 0,
   \]
   contradicting Lemma~\ref{lem:diff}.

   The case that $\lambda = \frac{1}{2}$ is similar. We must have that $u_i - u_j = \frac{2}{n}$, and thus we get $x_i - x_j + \frac{k}{n} + z = 0$.
   Sample $({\bf u}^1, {\bf u}^2)\sim {\sf DB}(1/4)$ conditioned on ${\bf u}^1+{\bf u}^2 = u$. Note that the event
   that ${\bf u}^1_i = 1/n$ and ${\bf u}^1_j = 0,$ occurs with probability $1/32$.
   Since $(x,u)$ is decent, we get that $x + k{\bf u}^1\in D$ with probability strictly greater than $\frac{31}{32}$.
   Thus, the probability that $x + k{\bf u}^1\in D$ and $({\bf u}^1_i, {\bf u}^1_j) = (1/n, 0)$ is positive, and in this
   case we get
   \[
     (x + k{\bf u}^1)_i - (x + k{\bf u}^1)_j
     =x_i - x_j + \frac{k}{n} = -z\neq 0,
   \]
   contradicting Lemma~\ref{lem:diff}.
\end{proof}
   We therefore get that $d\geq 0$, and the rest of the proof is identical to the proof of Claim~\ref{cl:1}.
%   \begin{itemize}
%     \item {\bf Case 1}: $d \in [0,0.5]$. In this case, we have $d(a_i+u_i, a_j+u_j) = d$, and thus the contribution of the pair $(i,j)$ to both sums is the same ($e^{-Z\cdot d}$).
%     \item {\bf Case 2}: $d > 0.5$. Since $\card{u_i-u_j}\leq 0.1$, it follows that $d(a_i,a_j) = d -(u_i-u_j) > 0.4$, which implies $d(a_i+u_i,a_j+u_j) > 0.3$. Therefore, the
%     contribution to $\Psi(a,u)$ from $i,j$ is at most $e^{-0.4\cdot Z}$ and to $\Psi(a+u)$ is at most $e^{-0.3 \cdot Z}$, and in particular $(i,j)$ contributes (in absolute value)
%     at most $e^{-Z/4}$ to the difference between the sums.
%   \end{itemize}
%   Taking a sum over all pairs $(i,j)$ concludes the proof.
\end{proof}

\begin{claim}\label{cl:2'}
There is an absolute constants $c_1>0$ and $t_0(M) > 0$, such that if $t\geq t_0$ then for every $x\in D$
\[
\Psi(x)\cdot
e^{c_1 k^2Z^2/n^2}
\leq
\Expect{{\bf u}\sim{\bf B}(1/4)}{\Psi(x,k {\bf u})}
\leq \Psi(x)\cdot e^{c_1^{-1} k^2Z^2/n^2}.
\]
\end{claim}
\begin{proof}
By linearity of expectation we have
	\[
	\Expect{{\bf u}\sim{\bf B}(1/4)}{\Psi(x,k{\bf u})}
    = \sum_{i<j} e^{-Z \cdot d(x_i,x_j)}\cdot \Expect{{\bf u}\sim{\bf B}(1/4)}{e^{-Z\cdot \ga(x_i,x_j)\cdot k({\bf u}_i-{\bf u}_j)}}.
	\]
    Note that the above expectation does not depend on $i,j$: for every $i,j$ the distribution of
    ${\bf u}_i - {\bf u}_j$ is ${\bf w}$, where
    $\Prob{}{{\bf w} = 2/n} = \Prob{}{{\bf w} = -2/n} = \frac{1}{16}$,
    $\Prob{}{{\bf w} = 1/n} = \Prob{}{{\bf w} = -1/n} = \frac{1}{4}$,
    $\Prob{}{{\bf w} = 0} = \frac{3}{8}$. In particular, this distribution is symmetric
    around $0$ and thus the sign $\ga(x_i,x_j)$ does not affect the expectation. Hence we have
    \[
    \Expect{{\bf u}}{\Psi(x,{\bf u})}   = \Psi(x) \cdot \E_{{\bf w}} [e^{kZ\cdot {\bf w}}]
	=\Psi(x)\cdot \E_{{\bf w}} \left[\frac{e^{kZ\cdot {\bf w}}+e^{-kZ\cdot {\bf w}}}{2}\right].
    \]
    Note that $\card{k Z \cdot {\bf w}}\leq M\frac{n\sqrt{\log t}}{t}\frac{t}{10\log t}\frac{1}{n}\leq 1$ for large enough $t$, so we have that
    \[
    e^{c_1(kZ\cdot {\bf w})^2}
    \leq
    \frac{e^{kZ\cdot {\bf w}}+e^{-kZ\cdot {\bf w}}}{2}
    \leq e^{c_1^{-1}(kZ\cdot {\bf w})^2}.
    \]
    Finally, the expectation of
    $e^{c (kZ\cdot {\bf w})^2}$ is at least $e^{c' k^2 Z^2/n^2}$
    and at most $e^{c'' k^2 Z^2/n^2}$, and the claim follows.
\end{proof}

The proofs of the following several claims are essentially identical to their analogs in Section~\ref{sec:lb},
and are therefore omitted.
We say a point $x$ is {\em good} if any interval of length $\frac{10 \log t}{t}$
on the circle contains at least $\log t$ and at most $100 \log t$ coordinates from $x \pmod{1}$.
By Chernoff bound, a random $x\in D$ is good with probability $>0.999$ given $t$ is large enough.

\begin{claim}\label{cl:6'}
There exists an absolute constant $c_2>0$, such that if
$x$ is good then $\Psi(x)> c_2 \log^2 t$.
\end{claim}
\begin{proof}
 The proof is identical to the proof of Claim~\ref{cl:6}.
\end{proof}

\begin{claim}\label{cl:7'}
There exists an absolute constant $c_3>0$, such that if $x$ is good,
then for all $i$ we have
$C_i < c_3 \frac{\Psi(x)}{\log t}$.
\end{claim}
\begin{proof}
	The proof is identical to the proof of Claim~\ref{cl:7}.
\end{proof}

\begin{claim}\label{cl:3'}
There exists an absolute constant $c_5,c_6>0$ and $t_1(M)>0$, such that if $t\geq t_1$ then for all good $x\in D$ we have
\[
	{\sf var}_{{\bf u}\sim{\bf B}(1/4)} [\Psi(x,{\bf u} )] \leq \frac{c_5}{\log t}
    \cdot \left(e^{c_6^{-1}\frac{k^2 Z^2}{n^2}} - e^{c_6\frac{k^2 Z^2}{n^2}}\right)\cdot \Psi(x)^2.
\]
\end{claim}
\begin{proof}
  The proof is a straightforward adaptation of the proof of Claim~\ref{cl:3}.
\end{proof}

Consequently, we have to adjust Claim~\ref{cl:4} as follows.
\begin{claim}
	\label{cl:4'}
There is an absolute constant $M>0$ and $t_2>0$ such that if
$k = M\frac{n\sqrt{\log t}}{t}$ and $t\geq t_1$, then for all good $x\in D$ we have
\[
\Pr_{{\bf u}\sim{\bf B}(1/4)}\left[\Psi(x,{\bf u})>\Psi(x)+\frac{c_1^2}{2}\frac{k^4 Z^4}{n^4}\Psi(x)\right]\geq 0.99.
\]
\end{claim}
\begin{proof}
    Let $c_1,\ldots,c_6$ be the constants from the previous claims, and choose
    $M = \sqrt{\frac{200 c_5}{c_1^2 c_6}}$. Then take $t_0(M)$, $t_1(M)$ from Claims~\ref{cl:2'}~\ref{cl:3'} and choose $t_2(M) = \max(t_0(M), t_1(M))$.
    We upper bound the probability of the complement event. Using Claim \ref{cl:2'}
    (and $e^t\geq 1+t+t^2/2$), we get
	\[
	\E_{{\bf u}\sim{\bf B}(1/4)}[\Psi(x,{\bf u})]\geq \Psi(x)\cdot \left(1+c_1\frac{k^2Z^2}{n^2}+\frac{c_1^2}{2}\frac{k^4Z^4}{n^4}\right).
	\]
	Hence
	\[
	\Pr_{{\bf u}\sim{\bf B}(1/4)}\left[\Psi(x,{\bf u})\leq \Psi(x)+\frac{c_1^2}{2}\frac{k^4Z^4}{n^4}\Psi(x)\right] \leq
    \Pr_{{\bf u}\sim{\bf B}(1/4)}\left[\card{\Psi(x,{\bf u})- \hspace{-2ex}\E_{{\bf u'}\sim{\bf B}(1/4)}[\Psi(x,{\bf u'})]}\geq \Psi(x)  c_1\frac{k^2Z^2}{n^2}\right].
	\]
    We want to upper bound the probability of the last event using Chebyshev's inequality.
    Since $x$ is good, the conclusion of Claim~\ref{cl:3'} holds, and so
	\[
	{\sf var}_{{\bf u}\sim{\bf B}(1/4)} [\Psi(x,{\bf u})]
    \leq
    \frac{c_5}{\log t}  \left(e^{c_6^{-1}\frac{k^2Z^2}{n^2}} - e^{c_6\frac{k^2Z^2}{n^2}}\right) \cdot \Psi(x)^2
    \leq \frac{c_5}{\log t} \cdot \frac{2 c_6^{-1} k^2Z^2}{n^2}\cdot \Psi(x)^2,
	\]
    for sufficiently large $t$. Therefore, applying Chebyshev's inequality we see the probability in question is at most
    \[
    \frac{{\sf var}_{{\bf u}\sim{\bf B}(1/4)} [\Psi(x,{\bf u})]}{\Psi(x)^2\cdot  c_1^2\frac{k^4 Z^4}{n^4}} \leq
	\frac{\frac{c_5}{\log t} \cdot \frac{2 c_6^{-1} k^2 Z^2}{n^2}\cdot \Psi(x)^2}
    {\Psi(x)^2\cdot  c_1^2\frac{k^4 Z^4}{n^4}}
    =\frac{2 c_5}{c_1^2 c_6}\frac{n^2}{k^2 Z^2\log t}
    =\frac{2 c_5}{c_1^2 c_6}\frac{1}{M^2}
    \leq 0.01.\qedhere
    \]
    \end{proof}

    \subsubsection{Finishing the argument}
    For each $u$, denote $\delta_u = \Prob{{\bf x}\in D}{{\bf x} + k u\not\in D}$,
    and note that $\delta = \Expect{{\bf u}}{\delta_{\bf u}}$.
\begin{claim}\label{cl:tv5'}
For each $u$, $\mathcal{D}_{TV}[{\bf x}; {\bf x}-ku]\leq \delta_u + \delta_{-u}$.
\end{claim}
\begin{proof}
  The proof is a direct conversion of the proof of Claim~\ref{cl:tv5} to the discrete setting, replacing the notion
  of ``Borel sets'' with finite sets.
\end{proof}

We can now prove Lemma~\ref{lem:main_odd_up}.
\begin{proof}[Proof of Lemma~\ref{lem:main_odd_up}]
    Take $M$ and $t_2$ from Claim~\ref{cl:4'}. We may assume that $t\geq t_2$, otherwise
    the lemma just follows from the fact that $\eta\geq \Omega(1/n)$, which holds as the value of the $t$-fold symmetric repeated game
    is at most the value of the original game, which is $1-\Theta(1/n)$.

    Take ${\bf x}\in_R D$, ${\bf u}\sim{\sf B}(1/4)$. Let $E_1$ be the event that $({\bf x}, {\bf u}), ({\bf x}, -{\bf u})$ are decent,
    $E_2$ be the event that $\Psi({\bf x})\leq c_2\log^2 t$,
    $E_3$ the event that ${\bf x}+k{\bf u}, {\bf x}-k{\bf u}, {\bf x}+{\bf u}, {\bf x}-{\bf u}\in D$, and let
    $E_4$ be the event that $\Psi({\bf x},{\bf u})\geq \Psi({\bf x})+\frac{c_1^2}{2}\frac{k^4 Z^4}{n^4}\Psi({\bf x})$.
    Finally, let $E_5$ be the event that
	$\Psi({\bf x}+{\bf u})>\Psi({\bf x})$ and denote $E({\bf x},{\bf u}) = E_1\cap \overline{E_2}\cap E_3\cap E_4$.
    Note that if the event $E$ holds for $x,u$, then $E_5$ also holds, since by Claim~\ref{cl:1'}:
	\[
	\Psi(x+u)\geq \Psi(x,u)-t^2\cdot e^{-Z/4}\geq \Psi(x)+\frac{c_1^2}{2}\frac{k^4 Z^4}{n^4}\Psi(x)-t^2\cdot e^{-Z/4}> \Psi(x).
	\]
    In the last inequality, we used the fact that if $E$ holds, then $\frac{c_1^2}{4}\frac{k^4 Z^4}{n^4}\Psi(x)\geq \Omega(1)$,
    and $t^2\cdot e^{-Z/4} = n^2 e^{-t/40\log t} = o(1)$ for large enough $t$.

    By Claim~\ref{claim:decent}, $\Prob{}{E_1}\geq 1-128(\delta+\eta)$.
    By Claim \ref{cl:6'} the probability of $E_2$ is at most the probability ${\bf x}$ is bad, hence it is at most $0.005$,
    by Claim~\ref{claim:two_small_steps} $\Prob{}{E_3}\geq 1-4(\delta+\eta)$, and by Claim~\ref{cl:4'},
    $\Prob{}{E_4}\geq 0.99$. We thus get
	\begin{equation}\label{eq:3}
	\Pr_{{\bf x},{\bf u}}[E({\bf x},{\bf u})]\geq 0.99 - 4(\delta+\eta) - 0.005 - 128(\delta+\eta) \geq 0.95 - 132(\delta+\eta).
	\end{equation}
    Fix $u$. Using Claim~\ref{cl:tv5'} we get that
    \[
    \Prob{{\bf x}}{E({\bf x}-u,u)} \geq
    \Prob{{\bf x}}{E({\bf x},u)} - \mathcal{D}_{TV}[{\bf x}; {\bf x}-u]
    \geq \Prob{{\bf x}}{E({\bf x},u)} - \delta_u-\delta_{-u}.
    \]
    By the union bound, we now conclude that
    \[
    \Prob{{\bf x}}{E({\bf x}-u,u)\cap E({\bf x},u)}\geq
    1 -  \Prob{{\bf x}}{\overline{E({\bf x}-u,u)}}
    -  \Prob{{\bf x}}{\overline{E({\bf x},u)}}
    \geq 2 \Prob{{\bf x}}{E({\bf x},u)} -1-\delta_u-\delta_{-u}.
    \]
    Taking expectation over a random step ${\bf u}$, we get that
    \[
      \Prob{{\bf x},{\bf u}}{E({\bf x}-{\bf u},{\bf u})\cap E({\bf x},{\bf u})}\geq
     2\Prob{{\bf x},{\bf u}}{E({\bf x},{\bf u})} -1- 2\Expect{{\bf u}}{\delta_{\bf u}}
     \geq 0.9 - 270(\delta+\eta),
    \]
    where we used~\eqref{eq:3}. Next, when both $E(x-u,u)$ and $E(x,u)$ hold, we have by the previous observation that $E_5$ holds for both pairs $(x-u,u)$ and $(x,u)$,
    and so $\Psi(x+u) > \Psi (x) = \Psi((x-u) + u) > \Psi(x-u)$. Thus, we get that
    $\Prob{{\bf x},{\bf u}}{\Psi({\bf x}+{\bf u}) > \Psi({\bf x}-{\bf u})}\geq 0.9 - 270(\delta+\eta)$.
    On the other hand, the probability on the left hand side is at most $0.5$; this follows as
    $\Prob{{\bf x},{\bf u}}{\Psi({\bf x}+{\bf u}) > \Psi({\bf x}-{\bf u})}
    = \Prob{{\bf x},{\bf u}}{\Psi({\bf x}-{\bf u}) > \Psi({\bf x}+{\bf u})}$, and their sum is at most $1$.
    Combining the two inequalities we get that $\eta+\delta \geq \Omega(1)$, which using Claim~\ref{claim:trivial_delta_eta}
    implies that $\eta = \Omega(1/k)$ as desired.
\end{proof}

\subsection{The lower bound: proof of Theorem~\ref{thm:val_sym_cycle_lb}}
In this section we use the symmetric body constructed in Theorem~\ref{thm:ub} in order
to prove Theorem~\ref{thm:val_sym_cycle_lb}.

\subsubsection{Tools}
We need the following isoperimetric inequality.
\begin{fact}\label{fact:iso}
  For all $\eps>0$ there is $\delta>0$ such that the following holds.
  Let $A\subseteq[0,1]^n$ be a measurable set such that
  $\eps\leq {\sf vol}(A)\leq 1-\eps$. Then ${\sf area}( A\cap {\sf interior}([0,1]^n))\geq \delta$.
\end{fact}
\begin{proof}
This is the combination of~\cite[Theorem 6, Theorem 7]{ros2001isoperimetric} as we explain below. Theorem 7 therein asserts
that if $A\subseteq[0,1]^n$ has Lebesgue measure $\alpha$ and surface area $S$, then there is a measurable set
in Gaussian space $B\subseteq \mathbb{R}^n$ with Gaussian measure $\alpha$ and (Gaussian) surface area at most $S$. Now~\cite[Theorem 7]{ros2001isoperimetric}
asserts among sets with Gaussian measure $\alpha$, the minimizers of surface area are halfspaces of the form
$B_{\beta} = \sett{z\in\mathbb{R}^n}{z_1\leq \beta}$ where $\beta$ is chosen so that the Gaussian measure of $B_{\beta}$
is $\alpha$, so $S\geq {\sf surface-area}(B_{\beta})$, which is bounded away from $0$ if $\alpha$ is bounded away from $0$ and $1$.
\end{proof}

Secondly, we need a slight strengthening of Theorem~\ref{thm:ub}. Recall that in Sections~\ref{sec:ub} and~\ref{sec:ns_sa_reduce} we have
constructed a semi-algebraic, bounded tiling body $D\subseteq\mathbb{R}^t$ whose surface area is $A = O(t/\sqrt{\log t})$, and for small enough
$\eps$ we have
\[
\Prob{{\bf x}\in D, \bm{\De}\sim N(0,\eps^2 I_t)}{{\bf x}+\bm{\De}\not\in D} \ll A\eps.
\]
We note that the argument in Section~\ref{sec:ub} holds in fact for more general class of $\bm{\De}$ (we only used the fact it
is independent of ${\bf x}$, has mean $0$ and is sub-Gaussian). Thus, we consider the distribution $\bm{\De}_{\eps}\in \set{0,\pm\eps/n}^t$ of Bernouli steps, namely
for each $i$ independently choosing $(\bm{\De}_{\eps})_i$ as $\Prob{}{(\bm{\De}_{\eps})_i = 0} = \half$,
$\Prob{}{(\bm{\De}_{\eps})_i = -\frac{\eps}{n}} = \frac{1}{4}$, $\Prob{}{(\bm{\De}_{\eps})_i = \frac{\eps}{n}} = \frac{1}{4}$.
Thus, running the argument therein we get:
\begin{lemma}\label{lem:main_ub_use'}
  The distribution over tiling bodies $(D_{\vec{r}})_{\vec{r}}$ from Lemma~\ref{lem:main_ub_use} satisfies, for small enough $\eps>0$
    \[
    \Expect{\vec{r}}{\Prob{{\bf x}, \bm{\De}_{\eps}}{\text{At least one of the conditions of Claim~\ref{cl:u1} fail for ${\bf x}$ and ${\bf x}+\bm{\De}_{\eps}$}}}
    \ll A \frac{\eps}{n}.
    \]
\end{lemma}
Slightly adapting the argument from Section~\ref{sec:ns_sa_reduce}, we may ensure that the chosen body $D$ also has small noise sensitivity
for Bernouli random steps $\bm{\De}_{\eps}$ for small enough $\eps$,\footnote{The proof is essentially the same, adapting
the definition of $G_k$ therein to be
    \[
    G_k =
    \left\{\vec{r} ~~~\left|
    \begin{array}{ll}
    &\Prob{\substack{{\bf x}\in D_{\vec{r}}\\ \bm{\De}\sim N(0,4^{-k}\cdot I_n)}}
    {{\bf x},{\bf x}+\bm{\De}\text{ lie in different cells of the tiling of }S_{\vec{r}}}\leq 4\cdot A 2^{-k},\\
    &\Prob{\substack{{\bf x}\in D_{\vec{r}}\\ \bm{\De}_{2^{-k}}}}
    {{\bf x},{\bf x}+\bm{\De}_{2^{-k}}\text{ lie in different cells of the tiling of }S_{\vec{r}}}\leq 4\cdot A 2^{-k}
    \end{array}
    \right.
    \right\}.
    \]
    }
but we will only need this to happen for a specific suitably chosen $\eps$ which can be ensured as follows. Take $\eps$ small enough for which
Lemma~\ref{lem:main_ub_use'} holds, and note that by Markov's inequality we get from Lemma~\ref{lem:main_ub_use'} that
\[
\Prob{\vec{r}}
{
\Prob{{\bf x}, \bm{\De}_{\eps}}{\text{At least one of the conditions of Claim~\ref{cl:u1} fail for ${\bf x}$ and ${\bf x}+\bm{\De}_{\eps}$}}
\geq C\cdot A\cdot \frac{\eps}{n}}\leq \frac{1}{4}
\]
for an absolute constant $C$. Thus, from Claim~\ref{claim:plenty_of_tilings} and the union bound we get that there is
$\vec{r}^{\star}\in \cap_{k\geq k_0} G_{k}$ such that the above event holds, and the rest of the proof in Section~\ref{sec:ns_sa_reduce} shows that
$D = D_{\vec{r}^{\star}}$ has surface area $O(A)$. We summarize this discussion with the following lemma.
\begin{lemma}\label{lem:use_for_strat}
  For all $t$, for small enough $\eps$, there is a symmetric, bounded tiling body $D$ with surface area $A = O(t/\sqrt{\log t})$ such that
  \[
  \Prob{{\bf x}, \bm{\De}_{\eps}}{\text{At least one of the conditions of Claim~\ref{cl:u1} fail for ${\bf x}$ and ${\bf x}+\bm{\De}_{\eps}$}}
  \ll A\cdot \frac{\eps}{n}.
  \]
\end{lemma}

\subsubsection{Decisive boxes}
In this section, we use Lemma~\ref{lem:use_for_strat} to devise a symmetric strategy for the players in the $t$-fold repeated game.
Take small enough $\eps$ so such Lemma~\ref{lem:use_for_strat} holds and assume that $k\defeq 1/\eps$ is an integer.
Let $D$ be the symmetric tiling body from Lemma~\ref{lem:use_for_strat}.
It will be convenient for us to think of challenges to the players as
$C_{n}^t = \sett{\frac{i}{n}}{i=0,1,\ldots,n-1}$.
Partition $[0,1)^t$ into the boxes $B_{\vec{a}} = \prod\limits_{i=1}^{t}\left[\frac{a_i}{n}, \frac{a_i}{n}+\frac{1}{n}\right)$
for $\vec{a}\in \set{0,1,\ldots,n-1}^t$; it will be convenient for us identify a challenge of a player ${\bf x}'$
with the box it belongs to, i.e. with $B_{\vec{a}}$ for $\vec{a} = n{\bf x}'$.
Consider the way $D$ further partitions the boxes $B_{\vec{a}}$.
\begin{definition}
  We say a box $B_{\vec{a}}$ is decisive if there exists $z\in\mathbb{Z}^n$ such that
  $\mu(B_{\vec{a}}\cap (D+z))\geq \frac{2}{3}\mu(B_{\vec{a}})$. Otherwise, we say $B_{\vec{a}}$
  is indecisive.
\end{definition}

We show that almost all boxes are decisive:
\begin{lemma}\label{lemma:most_are_decisive}
  The number of indecisive boxes is $O(A n^{t-1})$.
\end{lemma}
\begin{proof}
  Define
  $\Phi = \sum\limits_{z\in\mathbb{Z}^t}
  \sum\limits_{\vec{a}\in\set{0,1,\ldots,n-1}^t}{\sf area}(\partial(D+z)\cap {\sf interior}(B_{\vec{a}}))$.
  By considering the surface area of $D$, we will show that $\Phi\leq A$, and we will lower bound $\Phi$
  as a function of the number of the indecisive boxes, from which we will get the result.
  Let $B$ be such that $D\subseteq [-B,B]^t$, and take $m$ large enough.

  \paragraph{The upper bound.}
  For $\vec{a}\in \set{0,1,\ldots,mn-1}^t$, we define the box $B_{\vec{a}}$ as above,
  and define $\Phi_m = \sum\limits_{z\in\mathbb{Z}^t}
  \sum\limits_{\vec{a}\in\set{0,1,\ldots,mn-1}^t}{\sf area}(\partial(D+z)\cap {\sf interior}(B_{\vec{a}}))$.
  On the one hand, we clearly have that $\Phi_m = m^t \Phi$, and we next upper bound $\Phi_m$.
  Since $D\subseteq [-B,B]^t$, we have that
  \begin{align*}
  \Phi_m
  &=\sum\limits_{z\in\set{-B,-B+1,\ldots, B+m}^t}
  \sum\limits_{\vec{a}\in\set{0,1,\ldots,mn-1}^t}{\sf area}(\partial(D+z)\cap{\sf interior}(B_{\vec{a}}))
  \\
  &\leq \sum\limits_{z\in\set{-B,-B+1,\ldots, B+m}^t}{\sf area}(\partial(D+z))\\
  &=(m+2B+1)^t{\sf area}(\partial D)\\
  &\leq (m+2B+1)^t A.
  \end{align*}
  Combining the upper and lower bound we get $\Phi\leq \left(1+\frac{2B+1}{m}\right)^t A$, and
  sending $m$ to infinity gets that $\Phi\leq A$.

  \paragraph{The lower bound.} Interchanging the order of summation, we write
  \[
  \Phi =
  \sum\limits_{\vec{a}\in\set{0,1,\ldots,n-1}^t}
  \sum\limits_{z\in\mathbb{Z}^t}{\sf area}(\partial(D+z)\cap {\sf interior}(B_{\vec{a}})),
  \]
  and we show that if the box $B_{\vec{a}}$ is indecisive, then the innermost sum is at least $\Omega(1/n^{t-1})$.
  Indeed, if $B_{\vec{a}}$ is indecisive, then $\mu(B_{\vec{a}}\cap (D+z))\leq \frac{2}{3}\mu(B_{\vec{a}})$ for all
  $\vec{a}$. Thus, we may find $P\subseteq\mathbb{Z}^n$ such that for
  $H = B_{\vec{a}}\cap \bigcup_{z\in P} (D+z)$ we have that $\frac{1}{6}\mu(B_{\vec{a}})\leq\mu(H)\leq \frac{5}{6}\mu(B_{\vec{a}})$.
  We now scale and translate $H$, i.e. take $H' = nH - \vec{a}$, so that the above translates to $H'\subseteq[0,1]^n$ such that
  $\frac{1}{6}\leq \mu(H')\leq \frac{5}{6}$, and hence by Fact~\ref{fact:iso}
  ${\sf area}(\partial H' \cap {\sf interior}([0,1]^n))\geq \Omega(1)$. Removing the scaling, we get that
  ${\sf area}(\partial H\cap {\sf interior}(B_{\vec{a}}))\geq \Omega(n^{1-t})$. Therefore, we get that
  \[
  \Phi \geq
  \sum\limits_{\substack{\vec{a}\in\set{0,1,\ldots,n-1}^t\\ B_{\vec{a}}\text{ indecisive}}} \Omega(n^{1-t})
  =\Omega(n^{1-t}\cdot \#\set{\text{indecisive boxes}})
  \]
  Combining the upper and lower bound on $\Phi$, we get that
  the number of indecisive boxes is at most $O(A n^{t-1})$.
\end{proof}
Next, we show that if $B_{\vec{a}}$ is a typical decisive box, and
$\bm{\De}_1\in_R \set{0,\pm 1/n}$ is chosen randomly as above, then
$B_{\vec{a}+\bm{\De}_1}$ is very likely to be somewhat decisive, and
furthermore with the same cell of $D$.
\begin{lemma}\label{lem:neighbour_of_decisive}
It holds that
\begin{align}
\Prob{\substack{\bm{\De}_1\\ \vec{{\bf a}}\in\set{0,1,\ldots,n-1}^t}}{\exists z\in\mathbb{Z}^n, \mu(B_{\vec{{\bf a}}}\cap (D+z))\geq
\frac{2}{3}\mu(B_{\vec{a}}), \mu(B_{\vec{{\bf a}}+n\bm{\De}_1}\cap (D+z)) > \frac{1}{2}\mu(B_{\vec{{\bf a}}+n\bm{\De}_1})}\notag\\
\geq 1-O\left(\frac{A}{n}\right).\label{eq:4}
\end{align}
\end{lemma}
\begin{proof}
  Choose a random $\vec{{\bf a}}$, take a random ${\bf x}\in B_{\vec{{\bf a}}}$, and let ${\bf y} = {\bf x}\pmod{D}$. Note that as the distribution
  of ${\bf x}$ is uniform over $[0,1]^n$ and the distribution of ${\bf y}$ is uniform over $D$. Let
  $E_1(\vec{{\bf a}},{\bf x},\bm{\De}_1)$ be
  the event that ${\bf y}$ and ${\bf y} + \bm{\De}_1$ are in different cells of $D$. Then by the union bound and the choice of $D$
  \begin{align*}
  \Prob{\vec{{\bf a}}, {\bf x}, \bm{\De}_1}{E_1}
  &=\Prob{\vec{{\bf a}}, {\bf x}, \bm{\De}_{\eps}}{{\bf y}, {\bf y} + k\bm{\De}_{\eps}\text{ in different cells of }D}\\
  &\leq \sum\limits_{j=0}^{k-1} \Prob{{\bf y}, \bm{\De}_{\eps}}{{\bf y}+j\bm{\De}_{\eps}, {\bf y} + (j+1)\bm{\De}_{\eps}\text{ in different cells of }D}\\
  &= \sum\limits_{j=0}^{k-1} \Prob{{\bf w}\in D, \bm{\De}_{\eps}}{{\bf w}, {\bf w} + \bm{\De}_{\eps}\text{ in different cells of }D}\\
  &\leq  \sum\limits_{j=0}^{k-1} C\cdot A\cdot \frac{\eps}{n}
  = C\frac{A}{n}.
  \end{align*}
  Let $E_2(\vec{{\bf a}})$ be the event that $B_{\vec{{\bf a}}}$ is decisive, and if $E_2(\vec{{\bf a}})$ holds let ${\bf z}\in\mathbb{Z}^n$ be such that
  $\mu(B_{\vec{{\bf a}}}\cap (D+{\bf z}))\geq \frac{2}{3}\mu(B_{\vec{{\bf a}}})$. Then by Lemma~\ref{lemma:most_are_decisive}
  $\Prob{}{E_2(\vec{{\bf a}})}\geq 1-O(A/n)$. Denote
  \[
  p_{\vec{a}, \De_1} = \cProb{{\bf x}, {\bf a}, \bm{\De}_1}{\vec{{\bf a}}=\vec{a}, \bm{\De}_1 = \De_1}{E_1(\vec{{\bf a}},{\bf x},{\bf \De}_1)}.
  \]
  The expectation of $p_{\vec{{\bf a}},\bm{\De}_1}$ is the probability of $E_1(\vec{{\bf a}}, {\bf x},{\bf \De}_1)$, so
  \[
  \Prob{\vec{{\bf a}},\bm{\De}_1}{E_2(\vec{{\bf a}})\land p_{\vec{{\bf a}}, \bm{\De}_1}\leq \frac{1}{10}}
  \geq 1
  -\Prob{\vec{{\bf a}}}{\overline{E_2(\vec{{\bf a}})}}
  -\Prob{\vec{{\bf a}}, \bm{\De}_1}{p_{\vec{{\bf a}},\bm{\De}_1} > \frac{1}{10}}
  \geq 1-O\left(\frac{A}{n}\right) - \frac{\Prob{\vec{{\bf a}},{\bf x}, \bm{\De}_1}{E_1(\vec{{\bf a}},{\bf x}, \bm{\De}_1)}}{1/10},
  \]
  which is at least $1 - O\left(\frac{A}{n}\right)$.
  To finish the proof, we show that for every $\vec{a}$, $\De_1$ such that $E_2(\vec{a})$ holds and $p_{\vec{a},\De_1}\leq \frac{1}{10}$,
  we have the the event on the left hand side of~\eqref{eq:4} holds.

  Indeed, fix such $\vec{a}$, $\De_1$. Then there is a unique $z\in\mathbb{Z}^n$ such that
  $\mu(B_{\vec{a}}\cap (D+z)) = \Prob{{\bf x}\in B_{\vec{a}}}{{\bf x}\in (D+z)}$ is at least $\frac{2}{3}\mu(B_{\vec{a}})$.
  Note that if ${\bf y}$, ${\bf y}+\De_1$ are in the same cell of $D$, then ${\bf x}$, ${\bf x}+\De_1$ are in the same cell of $D$, so
  \begin{align*}
  \frac{\mu(B_{\vec{a}+n\De_1}\cap(D+z))}{\mu(B_{\vec{a}+n\De_1})}
  &=
  \frac{\mu(B_{\vec{a}+n\De_1}\cap(D+z))}{\mu(B_{\vec{a}})}\\
  &=
  \Prob{{\bf x}\in B_{\vec{a}}}{{\bf x}+\De_1\in (D+z)}\\
  &\geq
  \Prob{{\bf x}\in B_{\vec{a}}}{{\bf x}\in (D+z), {\bf x}+\De_1\in (D+z)}\\
  &\geq
  \Prob{{\bf x}\in B_{\vec{a}}}{{\bf x}\in (D+z)\text{ and }{\bf y}, {\bf y}+\De_1\text{ in the same cell of $D$}}\\
  &\geq
  \Prob{{\bf x}\in B_{\vec{a}}}{{\bf y}, {\bf y}+\De_1\text{ in the same cell of $D$}}
  -\Prob{{\bf x}\in B_{\vec{a}}}{{\bf x}\not\in (D+z)}\\
  &=
  1-p_{\vec{a},\De_1}
  -\Prob{{\bf x}\in B_{\vec{a}}}{{\bf x}\not\in (D+z)}\\
  &\geq 1 - \frac{1}{10}-\frac{1}{3} > \frac{1}{2}.\qedhere
  \end{align*}
\end{proof}

\subsubsection{Proof of Theorem~\ref{thm:val_sym_cycle_lb}}
In this section, we prove Theorem~\ref{thm:val_sym_cycle_lb}. For that, we show that the success probability of the
following players' strategy is at least $1-O(A/n)$.
\begin{enumerate}
  \item On challenge $x'\in C_n^t$, consider the box that $x'$ belongs to, i.e. $B_{\vec{a}}$ for $\vec{a} = n x'$.
  \item Check if there is $z\in\mathbb{Z}^t$ such that $\mu(B_{\vec{a}}\cap (D+z))> \frac{1}{2}\mu(B_{\vec{a}})$,
  and note that it is unique if such point exists. If there is no such $z$, abort. We refer to $z$ as the chosen lattice point
  of the player.
  \item Output $z + n x' \pmod{2}$.
\end{enumerate}
First, we argue that this strategy is symmetric. Indeed, the effect of permuting the entries of $x'$ by $\pi\in S_t$ is
that $a, z$ above also get permuted by $\pi$, and therefore the output also gets permuted by $\pi$. Next, we analyze the
success probability of this strategy.

Note the following equivalent way of picking challenges $({\bf x'}, {\bf y'})$:
sample $\vec{{\bf a}}\in \set{0,1,\ldots,n-1}^t$, set ${\bf x}'=\vec{{\bf a}}/n$, sample $\bm{\De}_1$ Bernouli as above and set ${\bf y}' = {\bf x}' + \bm{\De}_1\pmod{1}$.
Denote the box of ${\bf x}'$ by $B_{\vec{a}({\bf x}')}$, and consider the event $E$ defined in Lemma~\ref{lem:neighbour_of_decisive}.
We show that whenever the event $E$ holds, the players are successful with the above strategy, and as the probability of $E$
is at least $1-O(A/n)$, the proof would be concluded.

Fix $\vec{a},\De_1$ such that $E$ holds, and let $z\in\mathbb{Z}^t$ be the (unique) point such that
$\mu(B_{\vec{a}}\cap (D+z))\geq \frac{2}{3}\mu(B_{\vec{a}})$,
$\mu(B_{\vec{a}+n\De_1}\cap (D+z)) > \frac{1}{2}\mu(B_{\vec{a}+n\De_1})$.
The first condition implies that the $x'$-player does not abort and their chosen lattice point is $z$,
and we next show that the $y'$-player does not abort as well.
Note that the box of $y'$ is $B_{\vec{a}(y')}$ for $\vec{a}(y') = \vec{a}+n\De_1\pmod{1}$, and
write $\vec{a}+n\De_1 = \vec{a}(y') + w$ for $w\in \mathbb{Z}^t$.
Thus,
\[
\mu(B_{\vec{a}(y')}\cap (D+z-w))
=\mu(B_{\vec{a}(y')+w}\cap (D+z))
=\mu(B_{\vec{a}+n\De_1}\cap (D+z))
>\frac{1}{2}\mu(B_{\vec{a}+n\De_1}),
\]
which is equal to $\frac{1}{2}\mu(B_{\vec{a}(y')})$, so the $y'$-player also does not abort and their chosen lattice
point is $z-w$. We now analyze the answers of the players on each coordinate.
\begin{itemize}
  \item
  If $i$ is a coordinate such that $y'_i \neq x'_i$, then we may write
  $y'_i = x'_i +\De_1+b$ for $b\in\set{-1,0,1}$ and $\De_1\neq 0$. Then
  we get that $\vec{a}(y')_i = \vec{a}_i + n(\De_1)_i+nb$, so $w_i = -nb$.
  Thus, the answer of the $x'$-player is $z_i + n x'_i\pmod{2}$, whereas the answer of the $y'$-player is
  \[
  (z-w)_i + n y'_i
  =z_i + nb + n x'_i +n\De_1+nb
  =z_i + n x'_i +n\De_1+2nb
  = z_i + n {\bf x}'_i + 1\pmod{2},
  \]
  where we used $2nb = 0\hspace{-0.8ex}\pmod{2}$, and $n\De_1 = 1\hspace{-0.8ex}\pmod{2}$ (as $\De_1=\pm \frac{1}{n}$).
  Thus, the players are consistent on the $i$th coordinate.

  \item If $i$ is a coordinate such that $y'_i = x'_i$, then in the above notations we have $w_i = 0$, $\De_i = 0$ and
  we get that the answers of the players are the same on the $i$th coordinate, so they are consistent on $i$.\qed
\end{itemize}

\section{Open Problems}\label{sec:open}
In this section, we propose several challenges for further investigation of symmetric parallel repetition.

Recall from the introduction that on general games a strong parallel repetition theorem still fails, 
even for symmetric repetition. A simple example is the union of many disjoint, odd cycle games. It
would be interesting to understand for what instances of Max-Cut one has that a strong parallel holds with
symmetric repetition, motivating the following problem.

\begin{problem}
For the Max-Cut problem, extend the family of graphs for which symmetric parallel repetition
outperforms standard parallel repetition.
\end{problem}

Optimistically, one may hope that if symmetric parallel repetition would work for general enough
class of graphs, then one would be able to reduce any graph to a graph in that class
by mild preprocessing that doesn't affect the value of the game by much, and only then perform
symmetric repetition. If possible, that would establish the equivalence of the Max-Cut Conjecture
and UGC.

Secondly, there are well-known connections between parallel repetition and notions of mixing times and eigenvalues of
the underlying graph; for example, a strong parallel repetition theorem is known to hold for expander graphs~\cite{SS,AKKST},
and more generally for graphs with low threshold rank~\cite{tulsiani2014optimal}, i.e. graphs with only constantly many eigenvalues close to $1$.
We expect there could be stronger relations between symmetric parallel repetition and higher order eigenvalues of
$G^{\otimes_{{\sf sym}} k}$, the $k$-fold symmetric tensor product of $G$.
\begin{problem}
What is the relation between the performance of the $k$-fold symmetric parallel repetition
of a given instance of Max-Cut $G$, and the first $k+1$ eigenvalues of $G$?
\end{problem}

Finally, we believe that solving the foam problem for special classes of bodies may be an interesting geometric question
(albeit unrelated to the study of parallel repetition); a very natural class to study is the class of convex bodies.

\bibliographystyle{plain}
\bibliography{refs}

%\newpage
\appendix
\section{Deferred proofs}\label{sec:def_proofs}
\subsection{Proof of Claim~\ref{cl:A-6}}
We split the proof into two cases.
	\paragraph{Case 1: $r_i \le T/2$ for all $i$.} In this case,
	$\min(r_i, T-r_i)=r_i$ for all $i$, and the sum on the RHS of \eqref{eq:cl:A-6:0} is just $(\sum_i |d_i|)/T$. We have
	\begin{align*}
	\norm{p-q}_1 = \sum_i \left| \frac{r_i+d_i}{T'} - \frac{r_i}{T} \right|
    &\le
	\sum_i \left| \frac{r_i+d_i}{T} - \frac{r_i}{T} \right| +
	\sum_i \left| \frac{r_i+d_i}{T'} - \frac{r_i+d_i}{T} \right| \\
	&=\sum_i \frac{|d_i|}{T} + \left|\frac{1}{T'}-\frac{1}{T}\right|\cdot
	\sum_i (r_i+d_i)\\
    &= \sum_i \frac{|d_i|}{T}  + \left|1-\frac{T'}{T}\right|\\
    &=\sum_i \frac{|d_i|}{T} +\frac{1}{T} \cdot \left| \sum_i d_i \right| \\
	&\leq 2\cdot \sum_i \frac{|d_i|}{T}.
	\end{align*}
	
	\paragraph{Case 2: one of the $r_i$'s is greater than $T/2$.}
	Without loss of generality, $r_1>T/2$. Denote by $S:=\sum_{i>1} r_i = T-r_1$; $S':=\sum_{i>1} (r_i+d_i) = T'-r_1-d_1$.
	 In this case, the RHS of \eqref{eq:cl:A-6:0} is given by
	 \begin{equation}
	 \label{eq:cl:A-6:1}
	 \frac{|d_1| \cdot S}{r_1\cdot T} + \sum_{i>1} \frac{|d_i|}{T}.
	 \end{equation}
	 We will estimate $|p_1-q_1|$ and $\sum_{i>1} |p_i-q_i|$ separately. First,
     note that $T'\geq T - \sum_{j}\card{d_j}\geq T/2$.

For $|p_1-q_1|$, we have	
	 \begin{align*}
	 |p_1-q_1|= \left| \frac{r_1}{T} - \frac{r_1+d_1}{T'} \right| =
	  \left| \frac{r_1\cdot(S'-S)+ d_1\cdot S}{T\cdot T' } \right|
\le
	   2\left| \frac{S'-S}{T} \right|+
	   2\left| \frac{ d_1\cdot S}{T\cdot r_1 } \right|
\le \sum_{i>1} \frac{|d_i|}{T} +
	    \frac{|d_1| \cdot S}{r_1\cdot T}.
	 \end{align*}
In the third transition, we used the fact that $T' \geq T/2\geq r_1/2$.
	
	 For $\sum_{i>1} |p_i-q_i|$, by a similar calculation to the first case we have
	 \begin{align*}
	 \sum_{i>1}|p_i-q_i| =  \sum_{i>1} \left| \frac{r_i+d_i}{T'} - \frac{r_i}{T} \right|
     &\le
	 \sum_{i>1} \left| \frac{r_i+d_i}{T} - \frac{r_i}{T} \right| +
	 \sum_{i>1} \left| \frac{r_i+d_i}{T'} - \frac{r_i+d_i}{T} \right|\\
     &\leq
	 \sum_{i>1} \frac{|d_i|}{T}+\left|\frac{1}{T'}-\frac{1}{T}\right|\cdot
	 \sum_{i>1} r_i+d_i
     = \sum_{i>1} \frac{|d_i|}{T} + \card{\frac{1}{T'}-\frac{1}{T}} S',
     \end{align*}
     and it is enough to bound $\card{\frac{1}{T'}-\frac{1}{T}} S'$ by constant times the expression in~\eqref{eq:cl:A-6:1}.
     We have

  \begin{align*}
     \card{\frac{1}{T'}-\frac{1}{T}} S'=
	 \left|\frac{S'\cdot (S'-S) + S'\cdot d_1}{T'T}\right|
     &\leq
	 \left|\frac{(S'+d_1)\cdot (S'-S)}{T'T}\right|+
	 \left|\frac{ S\cdot d_1}{T'T}\right|\\
     &\leq
     \left|\frac{S'\cdot (S'-S)}{T'T}\right|
     +\left|\frac{d_1\cdot (S'-S)}{T'T}\right|
	 + 2\cdot
	 \left|\frac{ S\cdot d_1}{T^2}\right|,
	 \end{align*}
where in the last transition we used $T' \geq T/2 > 0$.
We bound each term separately.
For the first term, as $T'\geq T/2$, $\card{S'}\leq 2T$ (since $\card{d_i}\leq r_i$) we get
\[
\card{\frac{S'\cdot (S'-S)}{T'T}}
\leq 4\card{\frac{S'-S}{T}}
\leq 4\sum\limits_{i\geq 2}{\frac{\card{d_i}}{T}}.
\]

For the second term, we have $\card{d_1}\leq r_1\leq T$, $T'\geq T/2$ and so
\[
\card{\frac{d_1\cdot (S'-S)}{T'T}}
\leq 2\frac{\card{S'-S}}{T}
\leq 2\sum\limits_{i\geq 2}{\frac{\card{d_i}}{T}}.
\]

For the third term, we have, as $T\geq r_1$,
$\frac{ S\cdot d_1}{T^2}
\leq \frac{\card{d_1}}{r_1}\frac{S}{T}$.

\subsection{Proof of Proposition~\ref{prop:convex_error}}
	We will use the fact for points $x_i$ in our domain, $g_j(x_i) \ee \left(\frac{n}{\log n}\alpha_i\right)^3$.
	We consider two cases, based on the values of $S$ and $r$.
	\paragraph{Case 1: $\Pr_{{\bf x}_i} [r\cdot g_j({\bf x}_i)>S] < 1/2$.}
    We claim that for a sufficiently large constant $A>0$,
		\begin{equation*}
	\underbrace{\Expect{{\bf x}_i}{\sqrt{z+\frac{1}{\bm{\al}_i^2}} \cdot\frac{\min(r\cdot g_j({\bf x}_i),S)}
	{r\cdot g_j({\bf x}_i)+S+\ve^{1.6}}}}_{(\rom{1})}
    \leq
	\underbrace{\E_{{\bf x}_i}\left[ \sqrt{z+An^2/\log^2 n} \cdot \frac{\min(r\cdot g_j({\bf x}_i),S)}
	{r\cdot g_j({\bf x}_i)+S+\ve^{1.6}}  \right]}_{(\rom{2})}
    .
	\end{equation*}
	To do that, we compare both sides to $\Expect{{\bf x}_i}{\sqrt{z} \cdot\frac{\min(r\cdot g_j({\bf x}_i),S)}{r\cdot g_j({\bf x}_i)+S+\ve^{1.6}}} $.	
	For $(\rom{1})$, we have
    \[
	\Expect{{\bf x}_i}{\left(\sqrt{z+1/\bm{\al}_i^2} -\sqrt{z}\right)\cdot \frac{\min(r\cdot g_j({\bf x}_i),S)}
	{r\cdot g_j({\bf x}_i)+S+\ve^{1.6}}}
    \ll \Expect{{\bf x}_i}{\frac{1/\bm{\alpha}_i^2}{\sqrt{z+1/\bm{\al}_i^2}}\cdot \frac{\min(r\cdot g_j({\bf x}_i),S)}
	{r\cdot g_j({\bf x}_i)+S+\ve^{1.6}}}.
	\]
    Since $\bm{\alpha}_i\ll \log n/n$ always, we may further upper bound this by
    \[
    \ll \Expect{{\bf x}_i}{\frac{1/\bm{\alpha}_i^2}{\sqrt{z+A/(\log n/n)^2}}\cdot \frac{\min(r\cdot g_j({\bf x}_i),S)}
	{r\cdot g_j({\bf x}_i)+S+\ve^{1.6}}}
    \ll \Expect{{\bf x}_i}{\frac{1/\bm{\alpha}_i^2}{\sqrt{z+An^2/\log^2 n}}\cdot \frac{r\left(\frac{n}{\log n}\bm{\al}_i\right)^3}
	{S+\ve^{1.6}}},
	\]
    where we used $\min(r\cdot g_j({\bf x}_i),S)\leq rg_j({\bf x}_i)$ and the asymptotic we have for $g_j$. Simplifying and using
    $\Expect{{\bf x}_i}{\bm{\alpha}_i}\ll \log n/n$, we get that the last expression is equal to
    \[
    \frac{n^2}{\log^2 n}\frac{1}{\sqrt{z+An^2/\log^2 n}}\cdot \frac{r}{S+\ve^{1.6}}.
    \]

    For $(\rom{2})$, we have
    \[
    \Expect{{\bf x}_i}{\left(\sqrt{z+An^2/\log^2 n}-\sqrt{z}\right)\cdot \frac{\min(r\cdot g_j({\bf x}_i),S)}{r\cdot g_j({\bf x}_i)+S+\ve^{1.6}}}
    \gg
     \Expect{{\bf x}_i}{\frac{An^2/\log^2 n}{\sqrt{z+An^2/\log^2 n}}\cdot \frac{\min(r\cdot g_j({\bf x}_i),S)}{r\cdot g_j({\bf x}_i)+S+\ve^{1.6}}}.
    \]
    Restricting to the event $E$ that $r g_j({\bf x}_i)\leq S$ (that has probability at least $1/2$ by assumption), we have that the last expression
    is at least
    \[
    \gg\cExpect{{\bf x}_i}{E}{\frac{An^2/\log^2 n}{\sqrt{z+An^2/\log^2 n}}\cdot \frac{r\cdot g_j({\bf x}_i)}{S+\ve^{1.6}}}
    \gg \frac{An^2/\log^2 n}{\sqrt{z+An^2/\log^2 n}}\cdot \frac{r}{S+\ve^{1.6}},
    \]
    where the last inequality holds since $\cExpect{\alpha_i}{E}{g_j(x_i)}\gg 1$ (this is true for any event $E$ with constant probability
    in our range of interest of ${\bf x}_i$'s). Combining the bounds for $(\rom{1}),(\rom{2})$, we see that we may pick large enough $A$ so
    that
    \[
    \Expect{{\bf x}_i}{\left(\sqrt{z+1/\al_i^2} -\sqrt{z}\right)\cdot \frac{\min(r\cdot g_j({\bf x}_i),S)}
	{r\cdot g_j({\bf x}_i)+S+\ve^{1.6}}}
\leq \Expect{{\bf x}_i}{\left(\sqrt{z+An^2/\log^2 n}-\sqrt{z}\right)\cdot \frac{\min(r\cdot g_j({\bf x}_i),S)}{r\cdot g_j({\bf x}_i)+S+\ve^{1.6}}},
    \]
    and hence $(\rom{1})\leq (\rom{2})$. Let $A_1$ be a large enough value of $A$ so that this holds.

    \paragraph{Case 2: $\Prob{{\bf x}_i}{r\cdot g_j({\bf x}_i)>S} \geq 1/2$.}
    Using $\sqrt{a+b}\leq \sqrt{a}+\sqrt{b}$, we have
    \[
    \Expect{{\bf x}_i}{\sqrt{z+\frac{1}{\bm{\al}_i^2}} \cdot\frac{\min(r\cdot g_j({\bf x}_i),S)}
	{r\cdot g_j({\bf x}_i)+S+\ve^{1.6}}}
    \leq
    \underbrace{\Expect{{\bf x}_i}{\sqrt{z} \cdot\frac{\min(r\cdot g_j({\bf x}_i),S)}
	{r\cdot g_j({\bf x}_i)+S+\ve^{1.6}}}}_{(\rom{3})}
    +\underbrace{\Expect{{\bf x}_i}{\frac{1}{\bm{\alpha}_i} \cdot\frac{\min(r\cdot g_j({\bf x}_i),S)}
	{r\cdot g_j({\bf x}_i)+S+\ve^{1.6}}}}_{(\rom{4})}.
    \]
    Clearly,
    $(\rom{3})\leq \Expect{{\bf x}_i}{\sqrt{z+A\frac{n^2}{\log^2 n}} \cdot \frac{\min(r\cdot g_j({\bf x}_i),S)}{r\cdot g_j({\bf x}_i)+S+\ve^{1.6}}}$,
    and we upper bound $(\rom{4})$. Recall that $g_j({\bf x}_i) \ee \left(\frac{n}{\log n} \bm{\alpha}_i\right)^3$, so
    \[
    (\rom{4})
    \ll \Expect{{\bf x}_i}{\frac{1}{\bm{\alpha}_i} \cdot
    \frac{\min(r(n\bm{\alpha}_i/\log n)^3,S)}{B\cdot r (n\bm{\al}_i/\log n)^3+S+\ve^{1.6}}},
    \]
    for some absolute constant $B>0$. Writing the last expression as an integral, we note
    that $\bm{\alpha}_i$ is distributed uniformly on the interval $[0,\frac{\log n}{50n}+\eps^{0.95}]$, so
    we get
    \[
    (\rom{4})\ll\left(\frac{n}{\log n}\right)\int_{0}^{\frac{\log n}{25 n}}\frac{1}{t}\frac{\min(r(n t/\log n)^3,S)}{B\cdot r (nt/\log n)^3+S+\ve^{1.6}} dt.
    \]
    We break the range of integration into $R_1 = \left[0, (S/r)^{1/3} \frac{\log n}{n}\right]$,
    and $R_2 = \left[(S/r)^{1/3} \frac{\log n}{n},\frac{\log n}{25 n}\right]$. On $R_1$ our expression is equal to
    \[
    \left(\frac{n}{\log n}\right)^2\int_{0}^{\left(\frac{S}{r}\right)^{1/3} \frac{\log n}{n}}\frac{r(n t/\log n)^2}{B\cdot r (nt/\log n)^3+S+\ve^{1.6}} dt
    \ll
    \left(\frac{n}{\log n}\right)^4\int_{0}^{\left(\frac{S}{r}\right)^{1/3} \frac{\log n}{n}}\frac{r t^2}{S} dt
    \ll \frac{n}{\log n}.
    \]
    %Making the change of variables $y=\frac{n}{\log n}t$ and then $v=y^3$ we get the last
%    expression is
%    \[
%    \ll\frac{n}{\log n}\int_{0}^{S/r}\frac{r}{B\cdot r v+S+\ve^{1.6}} dv
%    \ll\frac{n}{\log n}\int_{0}^{S/r}\frac{r}{S} dv
%    =\frac{n}{\log n}.
%    \]

    On $R_2$ our expression is at most
    \[
    \left(\frac{n}{\log n}\right)\int_{\left(\frac{S}{r}\right)^{1/3} \frac{\log n}{n}}^{\frac{\log n}{25n}}
    \frac{1}{t}\frac{S}{B\cdot r (nt/\log n)^3} dt
    \ll\frac{S}{r}\left(\frac{\log n}{n}\right)^2\int_{\left(\frac{S}{r}\right)^{1/3} \frac{\log n}{n}}^{\frac{\log n}{25n}}
    \frac{1}{t^4} dt.
    \]
    Computing the integral, we see it is at most $(\left(\frac{S}{r}\right)^{1/3} \frac{\log n}{n})^{-3}$, hence
    the overall expression is $\ll n/\log n$, and since $\Expect{}{\mathbbm{1}_{r\cdot g_j(x_i)>S}}\geq 1/2$ we
    conclude that there is $A_2>0$ such that
    \[
    (\rom{4})\leq A_2 \frac{n}{\log n} \Expect{{\bf x}_i}{\mathbbm{1}_{r\cdot g_j({\bf x}_i)>S}}.
    \]

    The proposition is thus proven for $A = \max(A_1,A_2)$\qed

    \section{From Noise Sensitivity to Surface Area}\label{sec:ns_sa_reduce}
    Let $D_{\vec{r}}$ be a family of tilings of $\mathbb{R}^n$ that are constructed from Lemma~\ref{lem:main_ub_use}.
    I.e., the family $D_{\vec{r}}$ satisfies that the there is $A = O(n/\sqrt{\log n})$ such that for sufficiently small $\eps$, we have that
    \[
      \Expect{\vec{r}}{\Prob{\substack{{\bf x}\in D_{\vec{r}}\\ \bm{\De}\sim N(0,\eps^2\cdot I_n)}}
      {{\bf x},{\bf x}+\bm{\De}\text{ fall in different cells of the tiling induced by }D_{\vec{r}}}}\leq A\eps.
    \]
    Let $k_0$ be the first integer such that this condition holds for any $0 < \eps\leq 2^{-k_0}$.
    Thus, defining for each $k\geq k_0$ the set
    \[
    G_k = \Big\{\vec{r} \qquad\Big|\qquad \Prob{\substack{{\bf x}\in D_{\vec{r}}\\ \bm{\De}\sim N(0,4^{-k}\cdot I_n)}}
    {{\bf x},{\bf x}+\bm{\De}\text{ lie in different cells of the tiling of }S_{\vec{r}}}\leq 2\cdot A 2^{-k}\Big\},
    \]
    we have by Markov's inequality that $\Prob{\vec{r}}{\vec{r}\in G_k}\geq \half$.
    \begin{claim}
      The sets $G_k$ are monotone decreasing, i.e.\ for each $k$, $G_{k+1}\subseteq G_k$.
    \end{claim}
    \begin{proof}
      Fix $\vec{r}\in G_{k+1}$. Let $\De\sim N(0,4^{-k-1}\cdot I_n)$, and note that $\De' = 2\cdot\De\sim N(0,4^{-k}\cdot I_n)$.
      Thus,
      \begin{align}
      \Prob{\substack{{\bf x}\in D_{\vec{r}}\\ \bm{\De'}\sim N(0,4^{-k}\cdot I_n)}}{{\bf x},{\bf x}+\bm{\De'}\text{ in different cells}}
      &\leq \Prob{\substack{{\bf x}\in D_{\vec{r}}\\ \bm{\De}\sim N(0,A 4^{-k-1}\cdot I_n)}}{{\bf x},{\bf x}+\bm{\De}\text{ in different cells}}\notag\\\label{eq:apx_1}
      &+\Prob{\substack{{\bf x}\in D_{\vec{r}}\\ \bm{\De}\sim N(0,4^{-k-1}\cdot I_n)}}{{\bf x}+\bm{\De},{\bf x}+2\bm{\De}\text{ in different cells}}.
      \end{align}
      First, we argue that the second probability on the right hand side is equal to the first one. To see that,
      denote $y = x+\De$ and observe that
      the points $y,y+\De$ lie in different cells of the tiling induced by $D_{\vec{r}}$ if and only if the points
      $y\pmod{D_{\vec{r}}}$, $y\pmod{D_{\vec{r}}} + \De$ lie in different cells. Additionally, note for any fixed $\De$,
      the distribution of ${\bf y}\pmod{D_{\vec{r}}}$ when we take ${\bf x}\in_R D_{\vec{r}}$, is uniform over $D_{\vec{r}}$.

      Therefore, the bound we get from~\eqref{eq:apx_1} is (using the fact that $\vec{r}\in G_{k+1}$)
      \[
      2\cdot \Prob{\substack{{\bf x}\in D_{\vec{r}}\\ \bm{\De}\sim N(0,4^{-k-1}\cdot I_n)}}{{\bf x},{\bf x}+\bm{\De}\text{ in different cells}}\leq
      2\cdot 2\cdot A 2^{-(k+1)}
      = 2\cdot A 2^{-k},
      \]
      and so $\vec{r}\in G_{k}$.
    \end{proof}

    \begin{claim}\label{claim:plenty_of_tilings}
      It holds that $\Prob{\vec{r}}{\vec{r}\in \bigcap_{k\geq k_0} G_k}\geq \half$, and in particular $\bigcap_{k\geq k_0} G_k$
      is not empty.
    \end{claim}
    \begin{proof}
      Define the sequence of functions $g_m(\vec{r}) = \mathrm{1}_{\vec{r}\in \bigcap_{k_0\leq k\leq m} G_k}$,
      and also $g = \mathrm{1}_{\vec{r}\in \bigcap_{k\geq k_0} G_k}$. Clearly, on each $\vec{r}$, the sequence $g_m(\vec{r})$
      is monotonically decreasing to $g(\vec{r})$, and in other words we have monotone pointwise convergence of
      the non-negative functions $g_m$ to $g$.
      Thus, by the monotone convergence theorem
      \[
      \Prob{\vec{r}}{\vec{r}\in \bigcap_{k\geq 0} G_k}
      = \Expect{\vec{r}}{g(\vec{r})}
      = \Expect{\vec{r}}{\lim_{k\rightarrow\infty} g_k(\vec{r})}
      = \lim_{k\rightarrow\infty} \Expect{\vec{r}}{g_k(\vec{r})}.
      \]
      By the previous claim, $g_m = \mathrm{1}_{G_m}$, hence $\Expect{\vec{r}}{g_m(\vec{r})}\geq \half$
      and in particular the limit above is at least $\half$.
    \end{proof}

    Pick $\vec{r}^{\star}\in \bigcap_{k\geq k_0} G_k$, $\eps = 2^{-k_0}$ and denote $D = D_{\vec{r}^{\star}}$ for the rest of the proof. Clearly $D$ induces a tiling
    of the space $\mathbb{R}^n$, and next we will show that the surface area of $D$ is $O(A) = O(n/\sqrt{\log n})$, as desired.

    Towards this end, we will use Lemma~\ref{lem:santalo} that tells us that the surface
    area of $D$ is a constant multiple of
    \[
    \frac{1}{\eps}\Expect{\substack{{\bf x}\in_R D\\ \bm{\De}\sim N(0,\eps^2 I_n)}}{\card{({\bf x},{\bf x}+\bm{\De})\cap \partial D}},
    \]
    and we first observe that $({\bf x},{\bf x}+\bm{\De})\cap \partial D$ is almost surely countable.
    \footnote{The diligent reader may note that here, we are only considering intersections of the surface with the
    open interval $(x,x+\De)$ as opposed to the closed interval. This does not make any difference,
    since the contribution of the endpoints is proportional to the measure of $\partial D$. Hence, if the measure
    of $\partial D$ is $0$ they endpoints contribute $0$ to that expectation, and if the measure of $\partial D$
    is positive, then the expectation is infinite either way.}

    \begin{claim}\label{claim:countable}
      Let $\eps>0$ and sample ${\bf x}\in_R D$, $\bm{\De}\sim N(0,\eps^2 I_n)$. Then with probability $1$,
      $({\bf x},{\bf x}+\bm{\De})\cap \partial D$ is finite or countable.
    \end{claim}
    \begin{proof}
      Recall that by Lemma~\ref{lem:main_ub_use}, $D$ is a countable union of semi-algebraic sets, say $B_1,B_2,\ldots$.
      Note that for each semi-algebraic set $B_i$, the probability that $({\bf x},{\bf x}+\bm{\De})\cap \partial B_i$ is infinite is $0$, hence
      by the union bound, with probability $1$ all of these sets are finite, in which case $({\bf x},{\bf x}+\bm{\De})\cap \partial D$ is finite
      or countable.
    \end{proof}

    For a parameter $h$, a point $x\in\mathbb{R}^n$ and a direction $\De$, we say a
    point $y\in (x,x+\De)$ is $h$-isolated if
    \begin{enumerate}
      \item It holds that $y\in \partial D$.
      \item The neighbourhood of radius $h$ around $y$ does not contain $x,x+\De$ or any point from $\partial D$ (besides $y$).
    \end{enumerate}
    Define the quantity $g_{m}(x,\De)$ to be the number of $2^{-m}\norm{\Delta}_2$-isolated points in the interval $[x,x+\De]$.

    \begin{claim}\label{claim:converge}
      $g_{m}(x,\De)$ is an increasing sequence in $m$, and for any $x,\De$ for which Claim~\ref{claim:countable} holds, we have
      \[
      \lim_{m\rightarrow \infty} g_{m}(x,\De) = \card{(x,x+\De)\cap \partial D}.
      \]
    \end{claim}
    \begin{proof}
      The monotonicity of $g_m(x,\De)$ in $m$ is clear, and also that
      $g_{m}(x,\De)\leq \card{(x,x+\De)\cap \partial D}$. We set
      $\ell = g_{m}(x,\De)$ and split the rest of the proof according to
      whether $\ell$ is finite or not.

      \paragraph{Case 1: $\ell$ is finite.}
      In this case we argue that $g_{m}(x,\De) = \card{(x,x+\De)\cap \partial D}$
      for large enough $m$.  To see that, let $y_1,\ldots,y_\ell\in (x,x+\De)$ be all of the intersection
      points of $(x,x+\De)$ and $\partial D$, and take large enough $m$
      so that $2^{-m}\norm{\De}_2$ is smaller than all of the distances $\norm{y_i-y_j}_2$,
      $\norm{y_i-x}_2$, $\norm{y_i - (x + \De)}_2$ for all $i$ and $j$.

      \paragraph{Case 1: $\ell$ is infinite.}
      Consider the set $S = [x,x+\De]\cap \partial D$, and note that it is a closed.
      By Claim~\ref{claim:countable}, $S$ is countable, and we argue that $S$ must have an
      isolated point. Otherwise, $S$ is a closed set and has no isolated point, i.e.\
      it s a perfect set, but then it must be uncountable (e.g.\ see~\cite{perfect}).
      We thus conclude that $S$ has an isolated point $w_1$; we may remove it from
      $S$, have that the resulting set is again closed and countable, so we may again find
      an isolated point. Repeating this argument, for any $v\in\mathbb{N}$ we may find
      a collection of isolated points $w_1,\ldots,w_v\in S$ that are all different from
      $x$ and $x+\De$. As in case $1$, we conclude that $g_m(x,\De) \geq v$
      for large enough $m$, and since it holds for any $v$ we conclude that $\lim_{m\rightarrow \infty} g_m(x,\De) = \infty$.
    \end{proof}

    By Lemma~\ref{lem:santalo}, we have that the surface area of $D$ is at most a constant multiple of
    \[
    \frac{1}{\eps}\hspace{-0.5ex}\Expect{\substack{{\bf x}\in_R D\\ \bm{\De}\sim N(0,\eps^2 I_n)}}{\card{({\bf x},{\bf x}+\bm{\De})\cap \partial D}}
    =\frac{1}{\eps}\hspace{-0.5ex}\Expect{\substack{{\bf x}\in_R D\\ \bm{\De}\sim N(0,\eps^2 I_n)}}{\lim_{m\rightarrow \infty} g_{m}({\bf x},\bm{\De})}
    =\lim_{m\rightarrow \infty} \frac{1}{\eps}\hspace{-0.5ex}\Expect{\substack{{\bf x}\in_R D\\ \bm{\De}\sim N(0,\eps^2 I_n)}}{g_{m}({\bf x},\bm{\De})}.
    \]
    In the first transition we used Claims~\ref{claim:converge} and~\ref{claim:countable},
    and in the second one we used monotone convergence. Thus, if we assume that the surface area of
    $D$ is larger than $c\cdot A$ for a sufficiently large absolute constant $c$, then we get that
    $\lim_{m\rightarrow \infty} \frac{1}{\eps}\Expect{\substack{x\in_R D\\ \De\sim N(0,\eps^2 I_n)}}{g_{m}(x,\De)}\geq 10 A$.
    In the rest of the proof we will reach a contradiction and thereby show that for sufficiently large
    absolute constant $c$, the surface area of $D$ is at most $cA$, as required.

    By properties of limits, we conclude there exists $m$ such that
    \begin{equation}\label{eq:apx_2}
    \Expect{\substack{{\bf x}\in_R D\\ \bm{\De}\sim N(0,\eps^2 I_n)}}{g_{m}({\bf x},\bm{\De})}\geq 5 A\eps,
    \end{equation}
    and we fix this $m$ henceforth.

    Take $0<\delta\leq 2^{-m}$, and consider the following experiment. Take ${\bf x}\in_R D$ uniformly at random,
    $\bm{\De}\sim N(0,\eps^2 I_n)$ and take a uniformly random point ${\bf y}\in_R [{\bf x},{\bf x}+\bm{\De}]$. We
    consider the event $E$ in which the points ${\bf y}$ and ${\bf y}+\delta\bm{\De}$ lie in different cells in the tiling
    induced by $D$.

    \begin{claim}\label{claim:apx_one}
     For any $x,\De$ we have that $\cProb{{\bf y}}{x,\De}{E}\geq \delta g_m(x,\De)$.
    \end{claim}
    \begin{proof}
    Let $\ell = g_m(x,\De)$, and let $z_1,\ldots,z_{\ell}$ be the $2^{-m}\norm{\De}_2$-isolated points on the interval
    $(x,x+\De)$. For each $j$, let
    $I_j = (z_j - \delta\De,z_j)$, and note that as $\delta\leq 2^{-m}$ and the isolation of the points, we conclude that
    the intervals $I_j$ are disjoint and contained in $(x,x+\De)$.
    Also, note that if we pick $y\in I_j$, then $y$ and $y+\delta \De$ lie in different cells of the tiling induced by $D$;
    this holds since the interval between them contains exactly one point from $\partial D$ (namely, the point $z_j$). Therefore,
    \[
    \cProb{{\bf y}}{x,\De}{E}
    \geq\sum\limits_{j=1}^{\ell} \cProb{{\bf y}}{x,\De}{{\bf y}\in I_j}
    \geq\sum\limits_{j=1}^{\ell} \frac{\delta\norm{\De}_2}{\norm{\De}_2}
    =\delta\ell.\qedhere
    \]
    \end{proof}

   \begin{claim}\label{claim:apx_two}
     $\Prob{{\bf x},{\bf \De},{\bf y}}{E}\leq 2A\delta\eps$.
   \end{claim}
   \begin{proof}
      Consider ${\bf x},\bm{\De},{\bf y}$ the random variables in the definition of the event $E$. Let ${\bf z} = {\bf y}\pmod{D}$,
      and note that the points ${\bf y}$ and ${\bf y}+\delta\bm{\De}$ fall in different cells if and only if the points ${\bf z}$ and
      ${\bf z}+\delta\bm{\De}$ fall in different cells. Therefore, the probability of $E$ is exactly the probability that ${\bf z}$, ${\bf z}+\delta\bm{\De}$
      fall in different cells. Further, note that conditioned on $\bm{\De}$, the distribution of ${\bf z}$ is uniform over $D$, so
      \[
      \Prob{\bm{\De}\sim N(0,\eps^2 I_n)}{{\bf z}, {\bf z}+\delta\bm{\De}\text{ lie in different cells of } D}
      =\Prob{{\bf \De'}\sim N(0,\delta^2\eps^2 I_n)}{{\bf z}, {\bf z}+\bm{\De'}\text{ lie in different cells of }D},
      \]
      which is at most $2A\delta\eps$ by the choice of $D$ and the fact that $\delta\eps\leq \eps\leq 2^{-k_0}$.
    \end{proof}

    Combining the above claims we reach a contradiction:
    \[
    2A\delta\eps
    \geq
    \Prob{{\bf x},\bm{\De},{\bf y}}{E}
    =\Expect{{\bf x},\bm{\De}}{\cProb{{\bf y}}{{\bf x},\bm{\De}}{E}}
    \geq
    \Expect{{\bf x},\bm{\De}}{\delta g_m({\bf x},\bm{\De})}
    \geq \delta\cdot 5A\eps,
    \]
    and contradiction. The first transition is by Claim~\ref{claim:apx_one}, the second
    transition is by conditional probability formula, the third transition is by Claim~\ref{claim:apx_two}
    and the final one is by equation~\eqref{eq:apx_2}.\qed
\end{document}